\documentclass[11pt, letterpaper]{article}
\usepackage{setspace}
\usepackage{amsfonts}
\usepackage{amsmath}
\usepackage{amsthm}
\usepackage{amssymb}
\usepackage{mathrsfs}
\usepackage{mathtools}
\usepackage[hyphens]{url}
\usepackage{hyperref}
\usepackage[x11names]{xcolor}
\usepackage{graphicx}
\usepackage{xspace}
\usepackage{listings}
\usepackage{units}
\usepackage{makecell}
\usepackage{rotating}
\usepackage{tabularx}
\usepackage{multirow}
\usepackage{soul}
\usepackage[capitalize,noabbrev,nameinlink]{cleveref}
\usepackage[margin=1in]{geometry}
\usepackage{microtype}
\usepackage{pifont}
\usepackage{tcolorbox}
\usepackage{changepage}

\newtheorem{theorem}{Theorem}[section]
\newtheorem{lemma}[theorem]{Lemma}
\newtheorem{corollary}[theorem]{Corollary}

\newtheorem{proposition}[theorem]{Proposition}
\newtheorem{condition}{Condition}[]

\newcommand{\cmark}{\ding{51}}%
\newcommand{\xmark}{\ding{55}}%

\newcommand{\mypar}[1]{\noindent\textbf{#1}}

\definecolor{Azure1}{rgb}{0.95,1,1}

\lstdefinestyle{mystyle}{
    backgroundcolor=\color{Azure1},
    columns=fixed,
    basicstyle=\ttfamily\footnotesize,
    basewidth=0.5em,
    breakatwhitespace=false,
    breaklines=true,
    captionpos=b,
    frame=single,
    keepspaces=true,
    numbers=left,
    numberstyle=\small,
    numbersep=5pt,
    showspaces=false,
    showstringspaces=true,
    showtabs=false,
    tabsize=2
}

\crefname{lstlisting}{Listing}{Listings}

\begin{document}

\title{Optimal Computation in Anonymous Dynamic Networks}

\author{
Giuseppe A. Di Luna\thanks{
DIAG, Sapienza University of Rome, Italy, \texttt{g.a.diluna@gmail.com}}
\and
Giovanni Viglietta\thanks{
Department of Computer Science and Engineering, University of Aizu, Japan, \texttt{viglietta@gmail.com}}
}
\date{}
\singlespacing

\maketitle

\begin{abstract}
We give a simple characterization of the functions that can be computed deterministically by anonymous processes in dynamic networks, depending on the number of leaders in the network. In addition, we provide efficient distributed algorithms for computing all such functions assuming minimal or no knowledge about the network. Each of our algorithms comes in two versions: one that terminates with the correct output and a faster one that stabilizes on the correct output without explicit termination. Notably, these are the first deterministic algorithms whose running times scale \emph{linearly} with both the number of processes and a parameter of the network which we call \emph{dynamic disconnectivity} (meaning that our dynamic networks do not necessarily have to be connected at all times). We also provide matching lower bounds, showing that all our algorithms are asymptotically \emph{optimal} for any fixed number of leaders.

While most of the existing literature on anonymous dynamic networks relies on classic mass-distribution techniques, our work makes use of a novel combinatorial structure called \emph{history tree}, which is of independent interest. Among other contributions, our results make conclusive progress on two popular fundamental problems for anonymous dynamic networks: leaderless \emph{Average Consensus} (i.e., computing the mean value of input numbers distributed among the processes) and multi-leader \emph{Counting} (i.e., determining the exact number of processes in the network). In fact, our approach unifies and improves upon several independent lines of research on anonymous networks, including Yamashita--Kameda, IEEE T.\ Parall.\ Distr.~1996; Boldi--Vigna, Discrete Math.~2002; Nedi\'c et al., IEEE T.\ Automat.\ Contr.~2009; Kowalski--Mosteiro, J.\ ACM~2020.

Our contribution not only opens a promising line of research on applications of history trees, but also demonstrates that computation in anonymous dynamic networks is practically feasible and far less demanding than previously conjectured.
\end{abstract}

\clearpage

\tableofcontents

\clearpage

\section{Introduction}\label{s:1}

\subsection{Background}\label{s:1.1}

\mypar{Dynamic networks.} The study of theoretical and practical aspects of highly dynamic distributed systems has garnered significant attention in recent years~\cite{CFQS12,KO,MS18}. These systems involve a constantly changing network of computational devices called \emph{agents} (sometimes referred to as ``processes'' or ``processors''). Pairs of agents may send each other messages only through \emph{links} that appear or disappear unpredictably. This dynamicity is typical of modern real-world systems and is the result of technological innovations, such as the spread of mobile devices, software-defined networks, wirelesses sensor networks, wearable devices, smartphones, etc.

\smallskip
\mypar{Connected networks.} There are several models of dynamicity~\cite{CFQS12}; a popular choice is the \emph{1-interval-connected} network model~\cite{KLO10, DW05}. Here, a fixed set of $n$ agents communicate through links forming a time-varying graph, i.e., a graph whose edge set changes at discrete time units called \emph{rounds} (thus, the system is synchronous); such a graph changes unpredictably, but is assumed to be connected at all times.

\smallskip
\mypar{Disconnected networks.} The 1-interval-connected network model may not be a suitable choice for many real systems, due to the very nature of dynamic entities (think of P2P networks of smart devices moving unpredictably) or due to transient communication failures, which may compromise the network's connectivity. A more realistic assumption is that the union of all the network's links across any $\tau$ consecutive rounds induces a connected graph on the agents~\cite{KM22,NOOT09,OT09}. We say that such a network is \emph{$\tau$-union-connected}, and the smallest such parameter $\tau\geq 1$ is called \emph{dynamic disconnectivity}. Another widely used parameter for dynamic networks is the \emph{dynamic diameter} $d$. The relationship between $\tau$ and $d$ is discussed at the end of \cref{ss:disc}. In particular, it is worth noting that every occurrence of the parameter $\tau$ in the running times of our algorithms can be safely replaced with $d$. Observe that 1-interval-connected networks can equivalently be characterized as 1-union-connected networks. We remark that non-trivial (terminating) computation requires some conditions on temporal connectivity to be met, as well as some a-priori knowledge by agents (refer to \cref{o:timeimposs,o:timeimposs2}).

\smallskip
\mypar{Networks with unique IDs.} A large number of research papers have considered dynamic systems where each agent has a distinct identity (\emph{unique IDs})~\cite{DPR13,HF12,KLO10}. In this setting, there are efficient algorithms for consensus~\cite{KOM11}, broadcast~\cite{CFMS15}, counting~\cite{KLO10, DW05}, and many other fundamental problems~\cite{KLO11,MS18}.

It should be noted that networks with unique IDs allow for very simple algorithms for a wide variety of problems. Indeed, in a 1-interval-connected network of $n$ agents, if each agent broadcasts the set of its ``known IDs'' at every round (initially, an agent only knows its own ID), then every agent can learn the IDs of all other agents in $n$ rounds. If input values are attached to these IDs, then the network can compute any function of such input values within $n$ rounds.

For this reason, most research on networks with unique IDs has focused on the ``congested'' model, which limits the size of every message to $O(\log n)$ bits~\cite{PD00}. In this model, the \emph{all-to-all token dissemination} problem (where each agent begins with a unique token, and the objective is for all agents to collect every token) can be solved in $O(n^2)$ rounds using a \emph{token forwarding} strategy, where tokens are merely stored, copied, and sent without any alteration~\cite{KLO10}. Notably, any token-forwarding solution to the all-to-all token dissemination problem requires $\Omega(n^2/\log n)$ rounds~\cite{DPR13}.

\smallskip
\mypar{Anonymous networks.} Networks without IDs, also called \emph{anonymous systems}, pose an additional challenge compared to networks with unique IDs. In this model, all agents have identical initial states, and may only differ by their inputs. In the last thirty years, a large body of works have investigated the computational power of anonymous static networks, giving characterizations of what can be computed in various settings~\cite{BV01,BV02,CDS06,CGM08,JMM12,FPP00,SUW15,YK88}.

Anonymous systems are not only important from a theoretical standpoint, but also have a remarkable practical relevance. In a highly dynamic system, IDs may not be guaranteed to be unique due to operational limitations~\cite{DW05}, or may compromise user privacy. Indeed, users may not be willing to be tracked or to disclose information about their behavior; examples are COVID-19 tracking apps~\cite{SM20}, where a threat to privacy was felt by a large share of the public even if these apps were assigning a rotating random ID to each user. In fact, an adversary can easily track the continuous broadcast of a fixed random ID tracing the movements of a person~\cite{LAOOO20}. Anonymity is also found in insect colonies and other biological systems~\cite{GMTN15}.

We remark that, while in congested networks with unique IDs the difficulty is to cope with the severely limited message sizes, in (non-congested) anonymous networks the difficulty is to overcome the symmetry introduced by the anonymity of agents. Hence, the tools used in these two settings are radically different.

\smallskip
\mypar{Randomization.} One may wonder how randomization could help in breaking symmetry in anonymous networks. A simple strategy would be to sample unique IDs from a source of randomness, and then leverage such unique IDs. However, this approach would require an estimate on the network size to succeed with high probability, and is therefore unsuitable for safety-critical systems (where failures are not allowed, no matter how unlikely) or systems where reliable sources of randomness are not available, or when nothing is known about the network size. Furthermore, assigning unique IDs to agents compromises their anonymity, thereby forfeiting the security and privacy benefits of having an anonymous network.

Thus, in this paper we do not consider randomized algorithms, but only 
\emph{deterministic} ones. The only work we are aware of that proposes a randomized algorithm for determining the size of the network $n$ in the model considered in this paper is~\cite{KLO10}. However, this algorithm only provides approximate solutions and relies on prior knowledge of an upper bound on $n$. Additionally, its correctness is guaranteed only with high probability.

The problem of estimating the network size using randomized protocols has also been studied in the peer-to-peer community~\cite{massolino,jealino,massolino2}. We note that in these models, exact counting is not possible because agents continuously join and leave the network.

\smallskip
\mypar{Networks with leaders.} In order to deterministically solve several non-trivial problems in anonymous systems, it is necessary to have some form of initial ``asymmetry''~\cite{A80,BV02,MCS13,YK88}. The most common assumption is the existence of a single distinguished agent called \emph{leader}~\cite{AAE08,ABBS16,BBCD15,BBK11,DFIISV19,FPP00,KM18,KM20,MCS13,Sa99,YK96} or, less commonly, a subset of several leaders, and knowledge of their number~\cite{KM19,KM21,KMarxiv,KM22}.

Note that a network with an unknown number of leaders is equivalent to a network with no leaders at all. Also, if the leaders are distinguishable from each other, then any one of them can be elected as a unique leader. Hence, the only genuinely interesting multi-leader case is the one with a known number of indistinguishable leaders.

The presence of leaders is a realistic assumption: examples include base stations in mobile networks, gateways in sensor networks, etc. For these reasons, the computational power of anonymous systems enriched with one or more leaders has been extensively studied in the traditional model of static networks~\cite{FPP00,Sa99,YK96}, as well as in population protocols~\cite{AAE08,ABBS16,BBCD15,BBK11,DFIISV19}.

Apart from the theoretical importance of generalizing the usual single-leader scenario, studying networks with multiple leaders also has practical impacts in terms of privacy. Indeed, while the communications of a single leader can be traced, the addition of more leaders provides differential privacy for each of them.

\smallskip
\mypar{Leaderless networks.} In some networks, the presence of reliable leaders may not always be guaranteed. For example, in a mobile sensor network deployed by an aircraft, the leaders may be destroyed as a result of a bad landing; also, the leaders may malfunction during the system's lifetime. This justifies the extensive existing literature on networks with no leaders~\cite{CL22,C11,NOOT09,NOR18,O17,T84,YSSBG13}.
Notably, a large portion of works on leaderless networks have focused on the \emph{Average Consensus problem}, where the goal is to compute the mean of a list of numbers distributed among the agents~\cite{BT89,CL18,CL22,FSO18,OT09,OT11}.

\smallskip
\mypar{Problem classes.} Assume that each agent is assigned an input at the beginning of the computation, i.e., at round~$0$. The \emph{Input Frequency function} is the function that returns the percentage of agents that have each input value. Being able to compute the Input Frequency function allows a system to compute a wide class of functions called \emph{frequency-based} with no loss in performance~\cite{HOT}. Thus, we say that the Input Frequency function is \emph{complete} for this class of functions (refer to \cref{s:2.2}). The most prominent representative of this class of functions is given by the aforementioned Average Consensus problem, since the mean of a multiset of numeric input values is a frequency-based function (the percentage of agents that have each input can be used as weight to compute the mean of all inputs).

The \emph{Input Multiset function} is the function that returns the number of agents that have each input value. This function is complete for a class of functions called \emph{multiset-based}, which strictly includes the frequency-based ones (see \cref{s:2.2} for definitions). A well-studied example of a multiset-based function that is not frequency-based is given by the \emph{Counting problem}, which asks to determine the total number of agents in the system. The Counting problem is practically interesting in real-world scenarios such as large-scale ad-hoc sensor networks, where the individual agents may be unaware of the size of the system~\cite{DW05}.

\subsection{State of the Art}\label{s:1.2}
\mypar{Average Consensus problem.} In the Average Consensus problem, each agent starts with an input value, and the goal is to compute the average of these initial values. The typical approach found in the literature is based on local averaging algorithms, where each agent updates its local value at each round based on a convex combination of the values of its neighbors. This and similar techniques for Average Consensus have been studied for decades by the distributed control and distributed computing communities
~\cite{BT89,CL18,CL22,C11,KM21,NOOT09,O17,OT09,OT11,T84,YSSBG13}.
Existing works can be divided into those that give \emph{convergent} solutions, those that give \emph{stabilizing} solutions, and those that give \emph{terminating} solutions.

In convergent algorithms, the consensus is not reached in finite time, but each agent's local value asymptotically converges to the average. A state-of-the-art $\epsilon$-convergent algorithm based on Metropolis rules with a running time of $O\left(\tau n^{3} \log(1/{\epsilon})\right)$ communication rounds is given in~\cite{NOOT09}. However, this algorithm rests on the assumption that the degree of each agent in the network has a known upper bound; more generally, averaging algorithms based on Metropolis rules cannot be applied to our network model, because they require all agents to know their outdegree at every round prior to sending their messages. The running time of the algorithm in~\cite{NOOT09} can be improved to $O\left(n^2\log (1/\epsilon)\right)$ communication rounds if the network is always connected (i.e., $\tau=1$) and may only change every three rounds (as opposed to every round).

In stabilizing algorithms, the consensus is reached in a finite number of rounds, but no termination criteria are specified, and therefore all agents are assumed to run indefinitely. To the best of our knowledge, all stabilizing Average Consensus algorithms for leaderless networks are either probabilistic or assume the network to be static~\cite{FSO18,HOT,O17,YSSBG13}. However, some of these algorithms stabilize in a linear number of communication rounds~\cite{O17,YSSBG13}.

As for terminating Average Consensus algorithms, the one in~\cite{KM21} terminates in $O\left(n^{5+\epsilon} \log^{3} ({n})/\ell\right)$ communication rounds (for $\epsilon>0$) assuming the presence of a known number $\ell\geq 1$ of leaders and an always connected network.

A major research question left open in previous works is the following.

\begin{tcolorbox}
{\bf Research Question 1.} In anonymous dynamic leaderless networks, are there deterministic Average Consensus algorithms that stabilize (or terminate) in linear time?
\end{tcolorbox}

\mypar{Counting problem.} A long series of papers have focused the Counting problem in connected anonymous networks with a unique leader~\cite{CMM16,DB15,DB16,DBBC14a,DBBC14b,DBCB13,KM18,KM19,KM22,MCS13}. (We remark that the Counting problem cannot be solved in leaderless networks; this fact was implied by Angluin in~\cite{A80} and is now folklore.) These works have achieved better and better running times, with the first polynomial-time algorithm presented in J.\ ACM in 2020, having a running time of $O\left({n^{5}} \log^{2} (n)\right)$ rounds~\cite{KM20}. This was later improved and extended to networks with $\ell\geq 1$ leaders in~\cite{KM22}, solving the Counting problem in $O\left({n^{4+\epsilon}} \log^{3} (n)/\ell\right)$ communication rounds (for $\epsilon>0$).

The only result for $\tau$-union-connected networks is the recent preprint~\cite{KMarxiv}, which gives an algorithm that terminates in $\widetilde{O}\left({n^{3+2\tau(1+\epsilon)}}/{\ell}\right)$ communication rounds, assuming the presence of $\ell\geq 1$ leaders. We remark that this algorithm is exponential in $\tau$, and for $\tau=1$ it has a running time of $\widetilde{O}\left({n^{5+\epsilon}}/{\ell}\right)$.

Almost all of these works share the same basic approach of implementing a mass-distribution mechanism similar to the local averaging used to solve the Average Consensus problem. We should point out that such a mass-distribution approach requires agents to communicate numbers whose representation size grows at least linearly with the number of rounds; thus, all cited algorithms require agents to exchange messages of polynomial size. In spite of the technical sophistication of this line of research, there is still a striking gap in terms of running time between state-of-the-art algorithms for anonymous networks and networks with unique IDs. The same gap exists with respect to static anonymous networks, where the Counting problem is known to be solvable in $O(n)$ rounds~\cite{MCS13}. Given the current state of the art, solving non-trivial problems in large-scale dynamic networks is still impractical.

\begin{tcolorbox}
{\bf Research Question 2.} In anonymous dynamic networks with a fixed number of leaders, are there deterministic Counting algorithms that stabilize (or terminate) in linear time?
\end{tcolorbox}

\subsection{Our Contributions}\label{s:1.3}

\begin{table*}
\renewcommand{\arraystretch}{1.4}
\footnotesize
\begin{tabular}{|c|c|c|c|c|c|}
\hline
{\em Problem class} &{\em Leaders} &{\em Terminating} &{\em Assumptions} & {\em Running time} & {\em Lower bound}\\
\hline \hline
\multirow{2}{*}[-0.7em]{\makecell{Frequency-based \\ (e.g., Average Consensus)}}
& \multirow{3}{*}[0em]{$\ell= 0$} & \xmark &  & \hyperref[t:noleadstab]{$\tau(2n-2)$} & \hyperref[t:lower1]{$\tau(2n - 6)$}\\  \cline{3-6}
& & \multirow{2}{*}[0em]{\cmark} & $\tau$ and $N \geq n$ known & \hyperref[t:noleadterm]{$\tau(n+N-2)$} & \hyperref[t:lower1]{$\tau(2n - 4)$}\\ \cline{4-6}
& & & $d$ known & \hyperref[c:noleadterm]{$\tau(n-1)+d \leq \tau(2n-2)$} & \hyperref[t:lower3]{$\tau(2n - 4)$}\\
\hline \hline
\multirow{2}{*}[0em]{\makecell{Multiset-based\\(e.g., Counting)}}
& \multirow{2}{*}[0em]{$\ell \geq 1$} & \xmark & $\ell$ known & \hyperref[t:multileadstab]{$\tau(2n-2)$} & \hyperref[t:lower2]{$\tau(2n-\ell-5)$}\\  \cline{3-6}
& & \cmark & $\ell$ and $\tau$ known & \hyperref[t:multileadterm]{$\tau((\ell^2+\ell+1)(n-1)+1)$} & \hyperref[t:lower2]{$\tau(2n-\ell-3)$}\\
\hline
\end{tabular}
\caption{Summary of the results in this paper. The variable $n$ indicates the number of agents in the system, $\ell$ is the number of leaders, $N$ is an upper bound on $n$, $\tau$ is the dynamic disconnectivity of the network (or an upper bound thereof), and $d$ is its dynamic diameter. If $\ell=0$, only the frequency-based functions are computable; if $\ell\geq 1$, the multiset-based functions are also computable. All the running times hold for multigraph networks, while all lower bounds hold even for simple networks.\label{tab:summres}}
\end{table*}

\mypar{Summary.} Focusing on anonymous dynamic networks, in this paper we completely elucidate the relationship between leaderless networks and networks with one or more leaders, as well as the impact of the dynamic disconnectivity $\tau$ on the efficiency of distributed algorithms. In particular, we characterize the solvable problems in each of these settings and we provide optimal linear-time algorithms for all solvable problems. 

\smallskip
\mypar{Technique.} Our approach departs radically from the mass-distribution and averaging techniques traditionally adopted by most previous works on anonymous dynamic networks. Instead, all of our results are based on a novel combinatorial structure called \emph{history tree}, which completely represents an anonymous dynamic network and naturally models the idea that agents can be distinguished if and only if they have different ``histories'' (see \cref{s:3}). Thanks to the simplicity of our technique, this paper is entirely self-contained, our proofs are transparent and easy to understand, and our algorithms are elegant and straightforward to implement.

In \cref{s:3.3} we will compare history trees with other structures commonly used in the analysis of static networks. A more in-depth comparison is also found in~\cite{VIG24}. In particular, it should be noted that from the history tree of a static network it is possible to construct the views of Yamashita--Kameda, and vice versa. Hence, a history tree encodes all the information that the agents can extract from the network.

An implementation of most of the concepts discussed in this paper can be found at \url{https://github.com/viglietta/Dynamic-Networks}. The repository includes a dynamic network simulator that can be used to test our algorithms and visualize the history trees of custom networks. A browser version of the same program is also available at \url{https://giovanniviglietta.com/projects/anonymity}.

\smallskip
\mypar{Computability.} We give an exact characterization of the set of functions that can be computed deterministically in anonymous dynamic networks with and without leaders, respectively. Namely, in networks with at least one leader, a function is computable if and only if it is multiset-based (another way of stating this result is that it is sufficient to know the size of \emph{any} subset of distinguished agents in order to compute all multiset-based functions); on the other hand, in networks with no leaders, a function is computable if and only if it is frequency-based (a similar result, limited to \emph{static} leaderless networks and functions with additional constraints, was obtained in~\cite{HOT}).

Observe that computability is independent of the dynamic disconnectivity $\tau$. Notably, this is a consequence of the fact that our algorithms work for multigraphs, as opposed to simple graphs, as discussed later in this section.

\smallskip
\mypar{Algorithms.} Furthermore, we give efficient deterministic algorithms for computing the Input Frequency function in leaderless networks (\cref{s:stableaderless,s:termleaderless}) and the Input Multiset function in networks with leaders (\cref{s:stableader,s:termleader}). Since these functions are complete for the class of frequency-based and multiset-based functions, respectively, we automatically obtain efficient algorithms for computing \emph{all} functions in these classes (\cref{s:2.2}).

For each of the aforementioned functions, we give two algorithms: a \emph{terminating} version, where all agents are required to commit to their output and never change it, and a slightly more efficient \emph{stabilizing} version, where agents are allowed to modify their outputs, provided that they eventually stabilize on the correct output. The running times of our algorithms are summarized in \cref{tab:summres}.

The stabilizing algorithms for both functions give the correct output within $2\tau n$ communication rounds regardless of the number of leaders, and do not require any knowledge of the dynamic disconnectivity $\tau$ or the number of agents $n$. Our terminating algorithm for leaderless networks runs in $\tau(n+N)$ communication rounds with knowledge of $\tau$ and an upper bound $N\geq n$; the terminating algorithm for $\ell\geq 1$ leaders runs in $\tau(\ell^2+\ell+1) n$ communication rounds with no knowledge of $n$. The latter running time is reasonable (i.e., linear) in most applications, as $\ell$ is typically a constant or negligible compared to $n$. Note that the case where all agents are leaders is not equivalent to the case with no leaders: in the former case, agents do not have the information that $\ell=n$, and have to ``discover'' that there are no non-leader agents in the network.

We emphasize that the running times of our algorithms are specified as precise values, rather than being expressed in big-O notation. In fact, we made an effort to optimize the multiplicative constants, as well as the asymptotic complexity of our algorithms. To our knowledge, this feature is unique within the entire body of literature on anonymous dynamic networks.

We remark that the local computation time and the amount of internal memory required by our terminating algorithms is only polynomial in the size of the network. Also, like in previous works on anonymous dynamic networks, agents need to send messages of polynomial size.

\smallskip
\mypar{Negative results.} Some of our algorithms assume agents to have a-priori knowledge of some parameters of the network. In \cref{s:negative} we show that none of these assumptions can be removed, in the sense that removing any one of them would make the corresponding problem unsolvable by any algorithm. Note that we are not implying that this specific a-priori knowledge is strictly necessary in an absolute sense: it remains possible that alternative information, different from what we assume, could also suffice.

We also provide lower bounds that asymptotically match our algorithms' running times, assuming that the number of leaders $\ell$ is a constant (which is a realistic assumption in most applications). \cref{tab:summres} summarizes our lower bounds, as well. We point out that our lower bounds of roughly $2\tau n$ rounds are the first non-trivial lower bounds for simple undirected anonymous dynamic networks (i.e., better than $n-1$). 

\smallskip
\mypar{Multigraphs.} All of our results hold more generally if networks are modeled as multigraphs, as opposed to the simple graphs traditionally encountered in nearly all of the existing literature. This is relevant in many applications: in radio communication, for instance, multiple links between agents naturally appear due to the multi-path propagation of radio waves.

Furthermore, this turns out to be a remarkably powerful feature in light of \cref{p:time}, which establishes a relationship between multigraphs and $\tau$-union-connected networks. This finding single-handedly allows us to generalize our algorithms to disconnected networks at the cost of a mere factor of $\tau$ in their running times, which is worst-case optimal.

It is worth remarking that, while our algorithms apply to multigraphs, all the impossibility results concerning our algorithms' assumptions and all the principal lower bounds already hold for simple graphs (see \cref{s:negative}).

\smallskip
\mypar{Significance.} Our general technique based on history trees enables us to approach all problems related to anonymous networks in a uniform and systematic manner. This technique also allows us to design straightforward algorithms that achieve previously unattainable running times. Indeed, our results make conclusive advancements on long-standing problems within anonymous dynamic networks, particularly with regards to the Average Consensus and Counting problems, in several aspects:
\begin{itemize}
\item \emph{Running time.} Our algorithms are asymptotically worst-case optimal, i.e., linear (for networks with a fixed number of leaders $\ell$ and a fixed dynamic disconnectivity $\tau$).
\item \emph{Assumptions on the network.} Unlike most previous solutions, our algorithms work in any dynamic network that has a finite dynamic disconnectivity (i.e., a finite dynamic diameter).
\item \emph{Knowledge of the agents.} Our algorithms assume that the agents have a-priori knowledge about certain properties of the network only when the absence of such knowledge would render the Average Consensus or Counting problems unsolvable.
\item \emph{Quality of the solution.} Unlike several previous works, our algorithms are deterministic rather than probabilistic, and stabilize or terminate rather than converge.
\end{itemize}

Altogether, we settle open problems from several papers, including Nedi\'c et al., IEEE Trans.\ Automat.\ Contr.~2009~\cite{NOOT09}; Olshevsky, SIAM J.\ Control Optim.~2017~\cite{O17}; Kowalski--Mosteiro, J.\ ACM~2020~\cite{KM20} and SPAA~2021~\cite{KM21}. In particular, we settle the two research questions stated in \cref{s:1.2}:

\begin{tcolorbox}
{\bf Contribution 1.} In anonymous dynamic leaderless networks, any problem that can be solved deterministically (including Average Consensus) has an algorithm that stabilizes in linear time. If an upper bound on the number of agents is known, the algorithm also terminates in linear time.
\end{tcolorbox}

\begin{tcolorbox}
{\bf Contribution 2.} In anonymous dynamic networks with a fixed number of leaders, any problem that can be solved deterministically (including Counting) has an algorithm that terminates (thus stabilizes) in linear time.
\end{tcolorbox}

Our findings indicate that computing within anonymous dynamic networks entails an overhead when compared to networks equipped with unique IDs. However, this overhead is only linear, which is substantially lower than previously conjectured~\cite{KM20,KM22}. In fact, our results demonstrate that general computations in anonymous and dynamic large-scale networks are both feasible and efficient in practice.

\smallskip
\mypar{Unique-leader case.} Our results concerning connected networks with a unique leader (i.e., $\tau=\ell=1$) are especially noteworthy. We have an algorithm for the Counting problem that stabilizes in less than $2n$ rounds, which is worst-case optimal up to a small additive constant. We also have a terminating algorithm with a running time of less than $3n$ rounds, which is remarkably close to the lower bound of $2n$ rounds. We recall that, prior to our work, the state-of-the-art algorithm for this problem had a running time of $O\left({n^{4+ \epsilon}} \log^{3} (n)\right)$ rounds (with $\epsilon >0$)~\cite{KM22}.

\smallskip
\mypar{Previous versions.} This work is an extension of two preliminary conference papers that appeared at FOCS~2022~\cite{DV22} and DISC~2023~\cite{DVdisc}, respectively. In preparing the current version, we have reworked several sections, included all missing proofs, and provided a more accurate analysis of the performance of our algorithms. In addition, the upper bound in terms of the dynamic diameter $d$ for leaderless networks is a new contribution (see \cref{tab:summres}). We also added comparisons between history trees and previous structures and concepts related to anonymous and dynamic networks (see \cref{s:3.3}).

\section{Definitions and Preliminaries}\label{s:2}
In \cref{s:2.1} we define the model of computation of anonymous dynamic networks. In \cref{s:2.2} we define two classes of functions and their respective ``complete'' functions, which will play an important role in the rest of the paper. Finally, in \cref{ss:disc} we discuss $\tau$-union-connected networks and the impact of the dynamic disconnectivity $\tau$ on the running time of algorithms.

\subsection{Model of Computation}\label{s:2.1}
\mypar{Agents and networks.} A \emph{dynamic network} on a finite non-empty set $V=\{p_1,p_2,\dots, p_n\}$ is an infinite sequence $\mathcal G=(G_t)_{t\geq 1}$, where $G_t=(V,E_t)$ is an undirected multigraph, i.e., $E_t$ is a multiset of unordered pairs of elements of $V$. In this context, the set $V$ is called \emph{system}, and its $n\geq 1$ elements are the \emph{agents}. The elements of the multiset $E_t$ are called \emph{links}; note that we allow any (finite) number of ``parallel links'' between two agents, as well as ``self-loops''. Note that, in the dynamic networks literature, $G_t$ is typically assumed to be a \emph{simple} graph, for at most one link between the same two agents is allowed. However, our results hold more generally for multigraphs.

\smallskip
\mypar{Input and internal states.} Each agent $p_i$ starts with an \emph{input} $\lambda(p_i)$, which is assigned to it at \emph{round~$0$}. It also has an internal state, which is initially determined by $\lambda(p_i)$. At each \emph{round~$t\geq 1$}, every agent composes a message (depending on its internal state) and broadcasts it to its neighbors in $G_t$ through all its incident links.\footnote{In order to model wireless radio communication, it is natural to assume that each agent in a dynamic network broadcasts its messages to all its neighbors (a message is received by anyone within communication range). The network's anonymity prevents agents from specifying single destinations.} By the end of round~$t$, each agent reads all messages coming from its neighbors and updates its internal state according to a local algorithm $\mathcal A$. That is, $\mathcal A$ takes as input the current internal state of the agent and the multiset of incoming messages, and returns the new internal state of the agent. Note that $\mathcal A$ is deterministic and is the same for all agents.

The input of each agent also includes a \emph{leader flag}. The agents whose leader flag is set are called \emph{leaders} (or \emph{supervisors}). We will use the symbol $\mathcal N_{n,\ell}$ to denote the set of all dynamic networks of $n$ agents with all possible input assignments, with the condition that exactly $\ell$ agents are leaders. We also define $\mathcal N_{n}=\bigcup_{\ell=0}^n\mathcal N_{n,\ell}$ as the set of networks of $n$ agents with an ``unknown'' number of leaders.

\smallskip
\mypar{Stabilization and termination.} Each agent returns an \emph{output} at the end of each round; the output is determined by its current internal state. A system is said to \emph{stabilize} if the outputs of all its agents remain constant from a certain round onward; note that an agent's internal state may still change even when its output is constant.

An agent may also decide to explicitly \emph{terminate} and no longer update its internal state. An agent does so by setting a special \emph{termination flag} in its state, which implies that no further state transitions are possible. Since the output of an agent only depends on its internal state, a terminated agent can no longer update its output, either. When all agents have terminated, the system is said to \emph{terminate}, as well.

\smallskip
\mypar{Computation.} We say that an algorithm $\mathcal A$ \emph{computes} a function $F$ if, whenever the agents in the system are assigned inputs $\lambda(p_1)$, $\lambda(p_2)$, \dots, $\lambda(p_n)$ and all agents execute the local algorithm $\mathcal A$ at every round, the system eventually stabilizes with each agent $p_i$ giving the desired output $F(p_i,\lambda)$. A stronger notion of computation requires the system to not only stabilize but also to explicitly terminate with the correct output.

Formally, a function computed by a system of $n$ agents maps $n$-tuples of input values to $n$-tuples of output values. Writing such a function as $F(p_i,\lambda)$ emphasizes that the output of an agent may depend on all agents' inputs, as well as on the agent itself. That is, different agents may give different outputs.\footnote{We consider only functions whose outputs do not depend on the communication topology.}

The (worst-case) \emph{running time} of an algorithm $\mathcal A$, as a function of $n$, is the maximum number of rounds it takes for the system to stabilize (and optionally terminate), taken across all possible dynamic networks of $n$ agents and all possible input assignments, assuming that all agents run $\mathcal A$.

\smallskip
\mypar{Relationship between dynamic networks and Population Protocols.} The dynamic network model considered in this paper bears some minor similarities with Population Protocols~\cite{aspnes}. Indeed, both models describe systems where a set of entities exchange messages in a non-deterministic fashion and modify their internal states according to a deterministic local algorithm. However, a key difference is that in Population Protocols communication occurs only between two entities at a time; moreover, the two communicating entities have different roles: transmitter and receiver. This feature of Population Protocols automatically breaks symmetry between communicating entities, greatly simplifying problems such as Leader Election, Average Consensus, etc. It also eliminates all potential ambiguities arising from multiple entities in the same state sending equal messages to the same entity.

On the other hand, Population Protocols typically require the set of possible states to be independent of the number of entities $n$, while in our network model we allow agents to have an arbitrary amount of internal memory (although the terminating algorithms presented in this paper have polynomial memory requirements, as quantified in \cref{s:3.3}).

Another feature of Population Protocols is \emph{global fairness}, which roughly requires that the sequence of interactions between entities must allow any potentially reachable global state to eventually occur. This guarantees that the system can make progress despite an ``adversarial'' scheduler that might otherwise indefinitely block certain transitions. Our dynamic network model does not adopt the same notion of global fairness; however, a typical requirement is that the communication graph $G_t$ should be connected for every $t\geq 1$, or more generally the dynamic disconnectivity $\tau$ should be finite (see \cref{ss:disc}).

\subsection{Classes of Functions}\label{s:2.2}

\mypar{Multiset-based functions.} Let $\mu_\lambda=\{(z_1,m_1), (z_2,m_2), \dots, (z_k,m_k)\}$ be the multiset of all agents' inputs. That is, for all $1\leq i\leq k$, there are exactly $m_i$ agents $p_{j_1}, p_{j_2}, \dots, p_{j_{m_i}}$ whose input is $z_i=\lambda(p_{j_1})= \lambda(p_{j_2})= \dots= \lambda(p_{j_{m_i}})$; note that $n=\sum_{i=1}^k m_i$. A function $F$ is \emph{multiset-based} if it has the form $F(p_i,\lambda)=\psi(\lambda(p_i), \mu_\lambda)$. That is, the output of each agent depends only on its own input and the multiset of all agents' inputs.\footnote{Throughout the paper, all functions under consideration are understood to be effectively computable.} Multiset-based functions are also known as \emph{multi-aggregation} or \emph{multi-aggregate} functions~\cite{DV22,DVdisc}.

The special multiset-based functions $F_C(p_i,\lambda)=n$ and $F_{IM}(p_i,\lambda)=\mu_\lambda$ are called the \emph{Counting} function and the \emph{Input Multiset} function, respectively. It is easy to see that, if a system can compute the Input Multiset function $F_{IM}$, then it can compute any multiset-based function in the same number of rounds: thus, $F_{IM}$ is \emph{complete} for the class of multiset-based functions.

\begin{proposition}\label{xth:compl}
If the Input Multiset function $F_{IM}$ can be computed (with termination), then all multiset-based functions can be computed (with termination) in the same number of rounds, as well.
\end{proposition}
\begin{proof}
Once an agent $p_i$ with input $\lambda(p_i)$ has determined the multiset $\mu_\lambda$ of all agents' inputs, it can immediately compute any desired function $\psi(\lambda(p_i), \mu_\lambda)$ within the same local computation phase.
\end{proof}

We remark that \cref{xth:compl} does not require a unique leader or a connected network.

\smallskip
\mypar{Frequency-based functions.} For any $\alpha\in \mathbb R^+$, we define $\alpha\cdot \mu_\lambda$ as $\{(z_1,\alpha\cdot m_1), (z_2,\alpha\cdot m_2), \dots, (z_k,\alpha\cdot m_k)\}$. We say that a multiset-based function $F(p_i,\lambda)=\psi(\lambda(p_i), \mu_\lambda)$ is \emph{frequency-based} if $\psi(z,\mu_\lambda)=\psi(z,\alpha \cdot \mu_\lambda)$ for every positive integer $\alpha$ and every input $z$ (see~\cite{HOT}). That is, $F$ depends only on the ``frequency'' of each input in the system, rather than on their actual multiplicities. Notable examples of frequency-based functions include statistical functions of input values, such as mean, variance, maximum, median, mode, etc. Observe that the sum of all input values is a multiset-based function but not a frequency-based function.

The problem of computing the mean of all input values is called \emph{Average Consensus}
~\cite{BT89,CL18,CL22,C11,FSO18,KM21,NOOT09,NOR18,O17,OT09,OT11,T84,YSSBG13}.
The frequency-based function $F_{IF}(p_i,\lambda)=\frac 1n \cdot \mu_\lambda$ is called \emph{Input Frequency} function, and is complete for the class of frequency-based functions.

\begin{proposition}\label{th:concentration}
If the Input Frequency function $F_{IF}$ can be computed (with termination), then all frequency-based functions can be computed (with termination) in the same number of rounds, as well.
\end{proposition}
\begin{proof}
Suppose that an agent $p_i$ has determined $\frac 1n \cdot \mu_\lambda=\{(z_1,m_1/n), (z_2,m_2/n), \dots, (z_k,m_k/n)\}$. Then it can immediately find the smallest integer $d>0$ such that $d\cdot(m_i/n)$ is an integer for all $1\leq i\leq k$. Note that $\frac dn \cdot \mu_\lambda$ is a multiset. Hence, in the same round, $p_i$ can compute any desired function $\psi(\lambda(p_i), \frac dn \cdot \mu_\lambda)$, and thus any frequency-based function, by definition.
\end{proof}

\subsection{Disconnected Networks}\label{ss:disc}

Although the network $G_t$ at each individual round may be disconnected, in this paper we assume dynamic networks to be \emph{$\tau$-union-connected}. That is, there is a (smallest) parameter $\tau\geq 1$, called \emph{dynamic disconnectivity}, such that the sum of any $\tau$ consecutive $G_t$'s is a connected multigraph (by definition, the sum of (multi)graphs is obtained by adding together their adjacency matrices). Thus, for all $i\geq 1$, the multigraph $\left(V, \bigcup_{t=i}^{i+\tau-1}E_t\right)$ is connected (by definition, the union of multisets is obtained by adding together the multiplicities of equal elements). We will use the symbol $\mathcal N^\tau_{n,\ell}$ (respectively, $\mathcal N^\tau_{n}$) to indicate the subset of $\mathcal N_{n,\ell}$ (respectively, $\mathcal N_{n}$) consisting of the networks whose dynamic disconnectivity is exactly $\tau$.

Our $\tau$-union-connected networks should not be confused with the \emph{$\tau$-interval-connected} networks from~\cite{KLO10}. In those networks, the \emph{intersection} (as opposed to the union) of any $\tau$ consecutive $E_t$'s induces a connected (multi)graph. In particular, a $\tau$-interval-connected network is connected at every round, while a $\tau$-union-connected network may not be, unless $\tau=1$. Incidentally, a network is 1-interval-connected if and only if it is 1-union-connected.

\smallskip
\mypar{Knowledge of $\tau$.} We will argue that terminating computations are impossible without some a-priori knowledge about $\tau$. We say that a function $F(p_i,\lambda)$ is \emph{trivial} if and only if it is of the form $F(p_i,\lambda)=\psi(\lambda(p_i))$. That is, the output of any agent $p_i$ depends only on its own input $\lambda(p_i)$ and not on the inputs of other agents; therefore, $F$ can be computed ``locally'' and does not require communications between agents.

\begin{proposition}\label{o:timeimposs}
There is no algorithm that computes a non-trivial function with termination on all simple networks within $\bigcup_{\tau\geq 1}\bigcup_{n\geq 1} \mathcal N^\tau_{n}$.
\end{proposition}
\begin{proof}
Assume for a contradiction that a non-trivial function $F(p_i,\lambda)$ is computed with termination by an algorithm $\mathcal A$ in all simple networks in $\bigcup_{\tau\geq 1}\bigcup_{n\geq 1} \mathcal N^\tau_{n}$. Since $F$ is non-trivial, there is an input value $z$ and two distinct output values $y$ and $y'$ with the following properties.
\begin{itemize}
\item[(i)] If $p_1$ is the only agent in a network (i.e., $n=1$), and $p_1$ is assigned input $z$ and executes $\mathcal A$, it terminates in $t$ rounds with output $y$. \item[(ii)] There exists a network size $\tilde{n}>1$ and an input assignment $\lambda$ with $\lambda(p_1)=z$ such that, whenever the agents in a network of size $\tilde{n}$ are assigned input $\lambda$ and execute $\mathcal A$, the agent $p_1$ eventually terminates with output $y'\neq y$.
\end{itemize}
Let us now consider a simple dynamic network of $\tilde{n}$ agents, where $p_1$ is kept disconnected from all other agents for the first $t$ rounds (hence $\tau>t$). Assign input $\lambda$ to the agents, and let them execute algorithm $\mathcal A$. Due to property~(i), since $p_1$ is isolated for $t$ rounds and has no knowledge of $\tau$, it terminates in $t$ rounds with output $y$. This contradicts property~(ii), which states that $p_1$ should terminate with output $y'\neq y$.
\end{proof}

\cref{o:timeimposs} can be strengthened and extended in various ways. For instance, the same result can be shown to hold even if all agents know $n$ in advance. For a fixed $n\geq 1$, a function $F(p_i,\lambda)$ is called \emph{trivial for $n$ agents} if and only if it is of the form $F(p_i,\lambda)=\psi(\lambda(p_i),n)$ for every $n$-agent input assignment $\lambda$ and every agent $p_i$.

\begin{proposition}\label{o:timeimposs2}
For any fixed $n>1$, there is no algorithm that computes a function that is not trivial for $n$ agents with termination on all networks within $\bigcup_{\tau\geq 1} \mathcal N^\tau_{n}$.
\end{proposition}
\begin{proof}
We follow the same proof structure and notation as \cref{o:timeimposs}, with $\tilde{n}=n$, except that property~(i) is now different. Instead of having a network with a single agent, we have a network of $n$ agents arranged in a cycle (interpreted as two parallel links if $n=2$), each of which has input $z$. By assumption, all agents terminate after $t$ rounds with output $y$.

Now consider a second network of $n$ agents, where $p_1$ is connected only to itself via a double self-loop for the first $t$ rounds. If we assign input $\lambda$ to the agents, with $\lambda(p_1)=z$, then from the viewpoint of $p_1$ this network is indistinguishable from the previous one. Hence $p_1$ still terminates with output $y$, while the correct output is $y'\neq y$.
\end{proof}

A limitation of \cref{o:timeimposs2} is that, unlike \cref{o:timeimposs}, it requires the presence of parallel links or self-loops in the network. However, we can remove this requirement and prove a similar impossibility result for simple networks, at the cost of slightly narrowing the class of ``non-trivial functions'' that cannot be computed.

Specifically, if we assume that $\lambda$ assigns the same input $z$ to at least three agents, we can arrange these agents in a proper cycle. From the perspective of the agents within the cycle, this configuration is indistinguishable from any cycle of arbitrary length in which all agents have input $z$. Therefore, the same impossibility result still applies, but is now limited to functions that are not constant over all input assignments where at least three agents share the same input.

\smallskip
\mypar{Sufficiency of $\tau=1$.} The following result will be used repeatedly to reduce computations on $\tau$-union-connected networks to computations on 1-union-connected networks. Essentially, it states that any algorithm for 1-union-connected networks can be extended to $\tau$-union-connected networks (where $\tau$ is known to all agents) at the cost of a factor $\tau$ in the running time, which is optimal.

\begin{proposition}\label{p:time}
If there is an algorithm $\mathcal A_1$ that computes a function $F$ in $\mathcal N^1_{n}$ within $f(n)$ rounds (with termination) for all $n\geq 1$, then for any fixed $\tau\geq 1$ there is an algorithm $\mathcal A_\tau$ that computes $F$ in $\mathcal N^\tau_{n}$ within $\tau\cdot f(n)$ rounds (with termination) for all $n\geq 1$. Moreover, if $\mathcal A_1$ is optimal, then $\mathcal A_\tau$ is optimal.
\end{proposition}
\begin{proof}
Subdivide time into \emph{blocks} of $\tau$ consecutive rounds, and consider the following algorithm $\mathcal A_\tau$. Each agent collects and stores all messages it receives within the same block, and updates its state all at once at the end of the block. This reduces any $\tau$-union-connected network $\mathcal G=((V,E_t))_{t\geq 1}$ to a 1-union-connected network $\mathcal G'=((V,E'_t))_{t\geq 1}$, where $E'_t =\bigcup_{i=(t-1)\tau+1}^{t\tau}E_i$. Thus, if $F$ can be computed within $f(n)$ rounds in all 1-union-connected networks (which include $\mathcal G'$), then $F$ can be computed within $\tau f(n)$ rounds in the original network $\mathcal G$.

Conversely, consider a 1-union-connected network $\mathcal G$, and construct a $\tau$-union-connected network $\mathcal G'$ by inserting $\tau-1$ empty rounds (i.e., rounds with no links at all) before the first round of $\mathcal G$, as well as between every two consecutive rounds of $\mathcal G$. Since no information circulates during the empty rounds, if $F$ cannot be computed in fewer than $f(n)$ rounds in $\mathcal G$, then it cannot be computed in fewer than $\tau f(n)$ rounds in $\mathcal G'$ (recall that running times are measured in the worst case across all possible networks). Therefore, if $\mathcal A_\tau$ is not optimal, then $\mathcal A_1$ is not optimal either.
\end{proof}

We point out that the argument in the proof of \cref{p:time} is correct because algorithms are required to work for all \emph{multigraphs}, as opposed to simple graphs only. Indeed, since an agent $p_i$ may receive multiple messages from the same agent $p_j$ within a same block, the resulting network $\mathcal G'$ may have multiple links between $p_i$ and $p_j$ in a same round, even if $\mathcal G$ does not.

\smallskip
\mypar{Relationship between dynamic disconnectivity and the dynamic diameter.} A concept closely related to the dynamic disconnectivity $\tau$ of a network is its \emph{dynamic diameter} (or \emph{temporal diameter}) $d$, which is defined as the maximum number of rounds it may take for information to travel from any agent to any other agent at any point in time~\cite{CFQS12,KO}. It is a simple observation that $\tau\leq d\leq \tau(n-1)$.

We chose to predominantly use $\tau$, as opposed to $d$, to measure the running times of our algorithms for several reasons. Firstly, $\tau$ is well defined (i.e., finite) if and only if $d$ is; however, $\tau$ has a simpler definition, and is arguably easier to directly estimate or enforce in a real network. Secondly, \cref{p:time}, as well as all of our theorems in \cref{s:intermediate,s:termleader}, remain valid if we replace $\tau$ with $d$; nonetheless, stating the running times of our algorithms in terms of $\tau$ is preferable, because $\tau\leq d$.

We will use the symbol $\mathcal N^{[d]}_{n,\ell}$ to denote the subset of $\mathcal N_{n,\ell}$ consisting of the networks (with $n$ agents of which $\ell$ are leaders) whose dynamic diameter is exactly $d$. The square brackets around $d$ are used to avoid confusion with $\mathcal N^{\tau}_{n,\ell}$.

\section{History Trees}\label{s:3}
In this section, we introduce \emph{history trees} as a natural tool of investigation for anonymous dynamic networks. In \cref{s:3.1} we give definitions and we discuss basic structural properties of history trees. In \cref{s:3.2} we give a local algorithm to incrementally construct history trees and we prove a \emph{fundamental theorem} about this construction. Finally, in \cref{s:3.3} we discuss related concepts found in previous literature.

\subsection{Definition of History Tree}\label{s:3.1}
We will describe the structure of a history tree and its basic properties. For reference, an example showing part of a history tree is found in \cref{fig:ht}.

\begin{figure}
\centering
\includegraphics[scale=0.65]{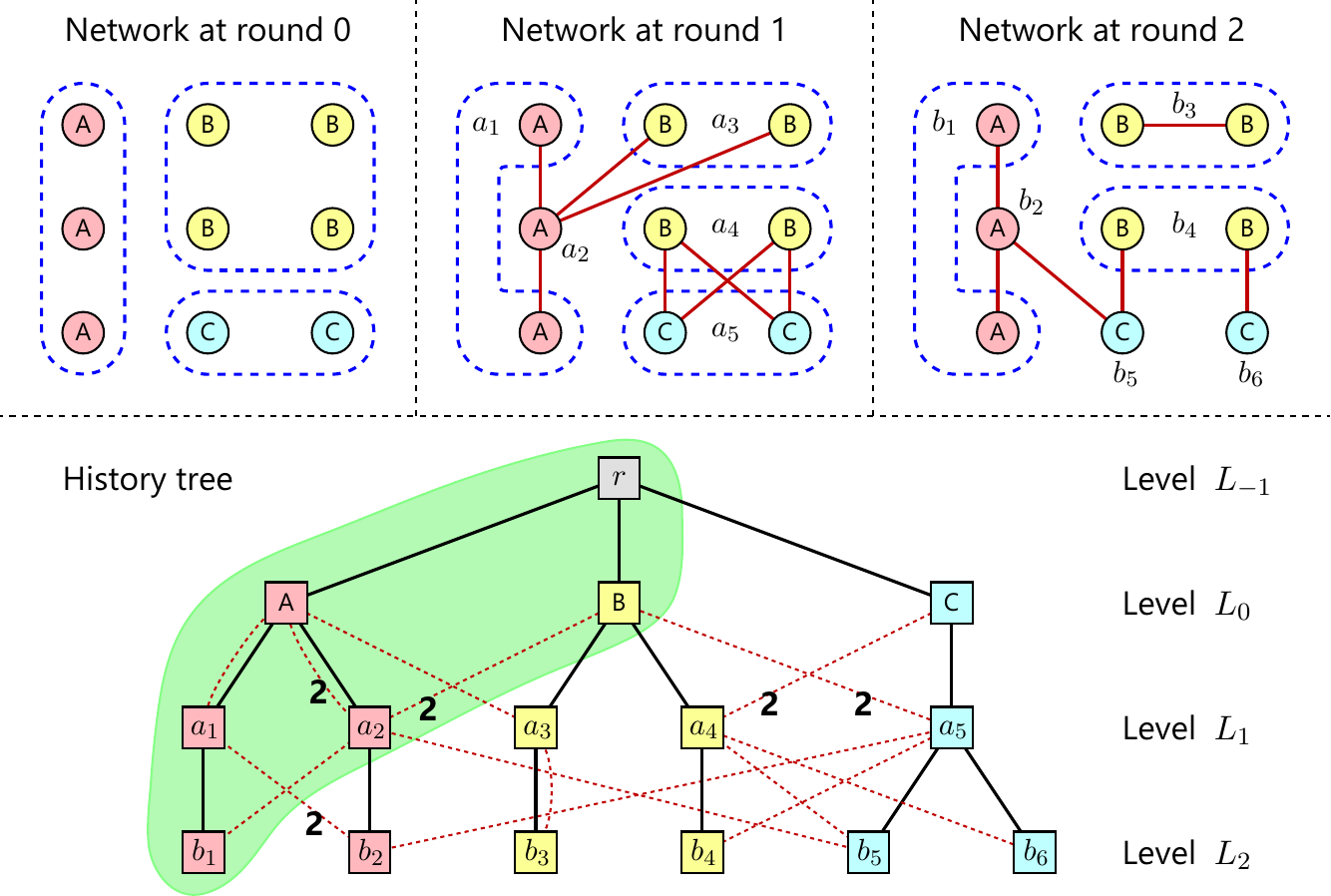}
\caption{The first rounds of a dynamic network with $n=9$ agents and the corresponding levels of the history tree. Level $L_t$ consists of all nodes at distance $t+1$ from the root $r$, which represent indistinguishable agents at the end of the $t$th communication round. The multiplicities of the red multi-edges of the history tree are explicitly indicated only when greater than $1$. The letters A, B, C denote agents' inputs; all other labels have been added for the reader's convenience, and indicate classes of indistinguishable agents (non-trivial classes are also indicated by dashed blue lines in the network). Note that the two agents in $b_4$ are still indistinguishable at the end of round~$2$, although they are linked to the distinguishable agents $b_5$ and $b_6$. This is because such agents were in the same class $a_5$ at round~$1$. The subgraph induced by the vertices in the green blob is the vista of the two agents in $b_1$. None of the levels of this vista is complete, except $L_{-1}$.}
\label{fig:ht}
\end{figure}

\smallskip
\mypar{Indistinguishable agents.} In an anonymous network, agents can only be distinguished by their inputs or by the multisets of messages they have received thus far. This leads to an inductive definition of \emph{indistinguishability}: two agents are indistinguishable at the end of round~$0$ if and only if they have the same input. At the end of round~$t\geq 1$, two agents $p$ and $q$ are indistinguishable if and only if they were indistinguishable at the end of round~$t-1$ and, for every equivalence class $A$ of agents that were indistinguishable at the end of round~$t-1$, both $p$ and $q$ receive an equal number of (identical) messages from agents in $A$ at round~$t$.

\smallskip
\mypar{Levels of a history tree.} A \emph{history tree} is a structure associated with a dynamic network. It is an infinite graph whose nodes are subdivided into \emph{levels} $L_{-1}$, $L_0$, $L_1$, $L_2$, \dots, where each node in layer $L_t$, with $t\geq 0$, represents an equivalence class of agents that are indistinguishable at the end of round~$t$. The level $L_{-1}$ contains a unique node $r$, representing all agents in the system.

In addition, each node of $L_0$ has a \emph{label} denoting the input of the agents it represents (by definition of indistinguishability at round~$0$, each node of $L_0$ represents all agents with a certain input). Nodes not in $L_0$ do not have any labels.

\smallskip
\mypar{Black and red edges.} A history tree has two types of undirected edges; each edge connects nodes in consecutive levels. The \emph{black edges} induce an infinite tree rooted at $r$ and spanning all nodes. A black edge $\{v, v'\}$, with $v\in L_{t}$ and $v'\in L_{t+1}$, indicates that the \emph{child node} $v'$ represents a subset of the agents represented by the \emph{parent node} $v$.

The \emph{red multiedges} represent messages. Each red edge $\{v, v'\}$ with multiplicity $m$, with $v\in L_{t}$ and $v'\in L_{t+1}$, indicates that, at round~$t+1$, each agent represented by $v'$ receives a total of $m$ (identical) messages from agents represented by $v$.

\smallskip
\mypar{Anonymity of a node.} The \emph{anonymity} $a(v)$ of a node $v$ of a history tree is defined as the number of agents represented by $v$. Since the nodes in a level represent a partition of all the agents, the sum of their anonymities must be equal to the total number of agents in the system, $n$. Moreover, by the definition of black edges, the anonymity of a node is equal to the sum of the anonymities of its children.

Observe that the problem of computing the Input Multiset function can be rephrased as the problem of determining the anonymities of all the nodes in $L_0$. Similarly, computing the Input Frequency function is the problem of determining the ratio $a(v)/n$ for each node $v$ in $L_0$.

\smallskip
\mypar{Vista of an agent.} A \emph{monotonic path} in the history tree is a sequence of nodes in distinct levels, such that any two consecutive nodes are connected by a black or a red edge. For any agent $p$, let $h(p,t)$ be the unique node representing $p$ in the level $L_t$ of the history tree. We define the \emph{vista}\footnote{The preliminary versions of this paper~\cite{DV22,DVdisc} used the term ``view'' instead of ``vista''. In the present version, we have adopted the new terminology to avoid confusion with the loosely related concept of \emph{view} introduced by Yamashita and Kameda in the context of static networks~\cite{YK88}. The relationship between vistas and views is discussed in \cref{s:3.3} and further explored in~\cite{VIG24}.} of $p$ at round~$t$ as the finite subgraph of the history tree induced by all the nodes spanned by monotonic paths with endpoints $h(p,t)$ and the root $r$.

A node of the history tree is \emph{missing} from a vista if it is not among the nodes in the vista. A level of a vista is \emph{complete} if no node is missing from that level. Observe that, if a level of a vista is complete, then all previous levels of the vista are also complete.

\subsection{Construction of Vistas}\label{s:3.2}
Intuitively, the vista of an agent represents the portion of the history tree that the agent ``knows'' at that time. In fact, there is an effective procedure that allows all agents in a network to locally construct a representation of their vista at all rounds. We will present this procedure and discuss some related results, including the \emph{fundamental theorem of history trees}, \cref{xth:view}.

\smallskip
\mypar{Merge operation.} The basic operation that is used to construct and update vistas is called \emph{merge}, and is illustrated in \cref{fig:merge}. Merging two vistas $\mathcal V$ and $\mathcal V'$ is a very natural operation whose result is the minimum graph that contains both $\mathcal V$ and $\mathcal V'$ as induced subgraphs.

Procedurally, the nodes of $\mathcal V'$ are matched and merged with the nodes of $\mathcal V$ level by level starting at the root $r$. That is, the procedure attempts to match the next vertex in level $L'_{t}$ of $\mathcal V'$ with an equivalent vertex already in level $L_t$ of $\mathcal V$. If no such vertex exists, it is created and connected with the appropriate vertices, which are already in the previous level. The procedure continues until all vertices of $\mathcal V'$ have been matched and merged.

\begin{figure}
\centering
\includegraphics[width=\linewidth]{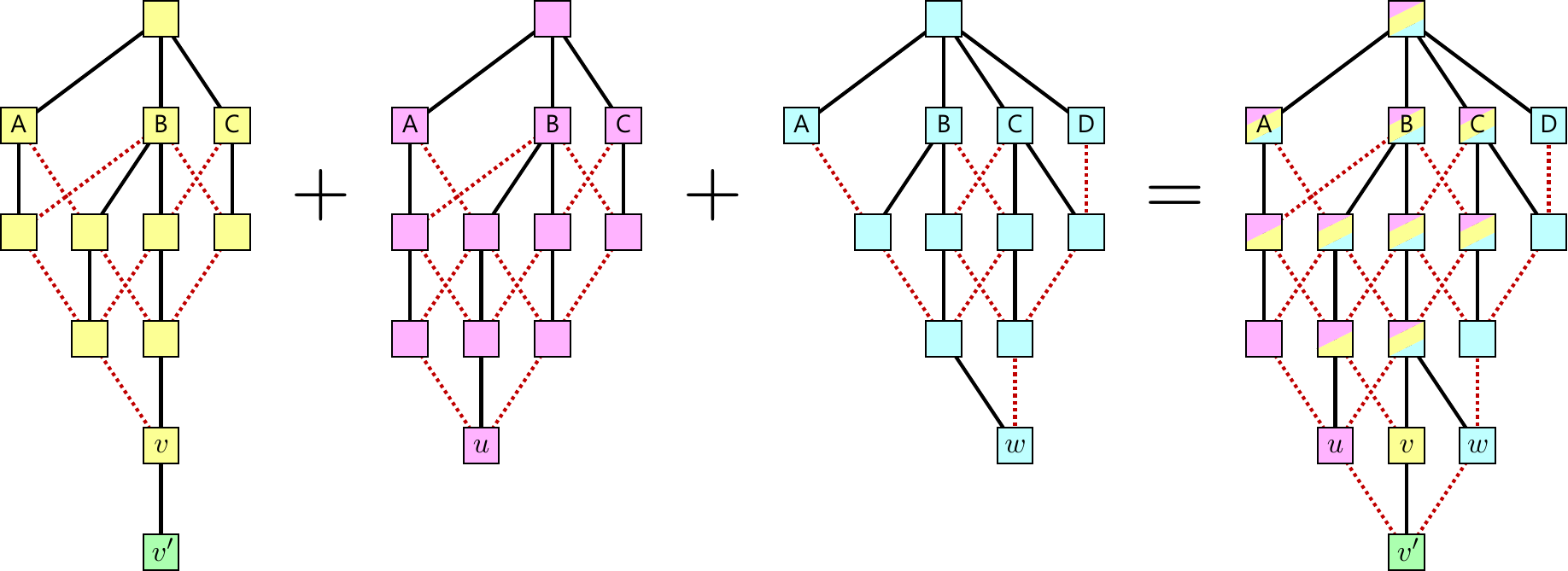}
\caption{Updating the vista of an agent represented by node $v$ after it receives messages from an agent represented by $u$ and an agent represented by $w$. The node $v'$ is created as a child of $v$, the vistas corresponding to $u$ and $w$ are merged with the old vista, and red edges are created connecting $v'$ with $u$ and $w$.}
\label{fig:merge}
\end{figure}

\smallskip
\mypar{Local construction.} We will now describe a local algorithm that allows agents to construct and update their vistas at every round. That is, assuming that each agent has its vista as its internal state at the beginning of round~$t$ and sends its vista to all its neighbors at round~$t$, there is a local algorithm $\mathcal A^\ast$ that allows the agent to construct its new vista at the end of round~$t$.

Constructing the vista at round~$0$ is simple, because the agent just has to read its own input and create a root $r$ with a single child bearing a label equal to its input.

After that, the vista can be updated at every round as shown in \cref{fig:merge}. Namely, the node $v$ that currently represents the agent gets a new child $v'$, representing the same agent in the next round. Then, if $m\geq 1$ copies of a vista $\mathcal V$ are received as messages, then $\mathcal V$ is merged into the current vista, and a red edge with multiplicity $m$ is added to connect $v'$ with the deepest node of $\mathcal V$. The correctness of this operation follows straightforwardly from the definition of vista.

\smallskip
\mypar{Size of the vista.} We remark that the size of the vista of an agent at round~$t$ is polynomial in $t$ and $n$, and so is the amount of local computation needed to update the agent's vista at round~$t$ based on the messages it has received. More specifically, the vista can be represented in $O(tn^2\log M)$ bits, where $M$ is the maximum number of messages that any agent may receive at any round (if the network is simple, then $M<n$).

Indeed, each of the $O(t)$ levels in the vista contains up to $n$ nodes. Thus, there are $O(n^2)$ possible red edges between any two consecutive levels, each of which has a multiplicity that can be represented in $O(\log M)$ bits.

\smallskip
\mypar{Fundamental theorem.} 
We will now give a concrete meaning to the idea that the vista of an agent at round~$t$ contains all the information that the agent can use at that round. This intuition is made precise by the following \emph{fundamental theorem of history trees}.

\begin{theorem}\label{xth:view}
If all agents in a system execute the same local algorithm $\mathcal A$ at every round, then the internal state of an agent $p$ at the end of round~$t$ is determined by a function $\mathcal F_\mathcal A$ of the vista of $p$ at round~$t$. The mapping $\mathcal F_\mathcal A$ depends entirely on the algorithm $\mathcal A$ and is independent of $p$.
\end{theorem}
\begin{proof}
Denote by $h(p,t)$ the node of the history tree that represents $p$ at round $t$, by $\mathcal V(p,t)$ the vista of $p$ at round $t$, and by $\sigma(p,t)$ the internal state of $p$ at round $t$. The construction of $\mathcal F_\mathcal A$ is done by induction on $t$ in such a way that $\mathcal F_\mathcal A(\mathcal V(p,t))= \sigma(p,t)$ for any agent $p$ and any round~$t$.

If $t=0$, then $\mathcal V(p,t)$ consists of a single node $h(p,t)$ in $L_0$ attached to the root $r$. Since each node of $L_0$ is labeled as the input of the agents it represents, $\mathcal F_\mathcal A$ can extract the input of $p$ from $h(p,t)$ and use it to compute its internal state $\sigma(p,t)$, which at round~$0$ is just a function of the input.

Let $t>0$ and assume that the inductive hypothesis $\mathcal F_\mathcal A(\mathcal V(p,t-1))= \sigma(q,t-1)$ holds for any agent $q$. Let $p$ be any agent, and let us prove that $\mathcal F_\mathcal A(\mathcal V(p,t))= \sigma(p,t)$. Observe that, by definition, if a vista contains a node in level $L_{i}$, it also contains the vista of the agents represented by that node at round~$i$.

It is easy to infer $h(p,t)$ from $\mathcal V(p,t)$, because it is the unique node of maximum depth. Thus, the node $h(p,t-1)$ is the parent of $h(p,t)$, which is in $\mathcal V(p,t)$ because it is connected to $h(p,t)$ by a black edge. It follows that $\mathcal F_\mathcal A$ can extract from $\mathcal V(p,t)$ the vista of $p$ at round~$t-1$ and, by the inductive hypothesis, it can compute $\sigma(p,t-1)$.

Also, for each agent $q$ that sends messages to $p$ at round~$t$, the node $h(q,t-1)$ is in $\mathcal V(p,t)$, because it is connected to $h(p,t)$ by a red edge with a certain multiplicity $m\geq 1$. Hence, $\mathcal F_\mathcal A$ can extract $m$ from $\mathcal V(p,t)$, as well as the vista of $q$ at round~$t-1$, and compute $\sigma(q,t-1)$ by the inductive hypothesis. From the internal state $\sigma(q,t-1)$, the message received by $p$ from $q$ at round~$t$ can also be computed.

Thus, $\sigma(p,t)$ can be computed by running the algorithm $\mathcal A$ with input $\sigma(p,t-1)$ (i.e., the previous internal state of $p$) and the messages received by $p$ at round~$t$ with the appropriate multiplicities. Since we have shown how $\mathcal F_\mathcal A$ can extract all this information from $\mathcal V(p,t)$, the theorem is proved.
\end{proof}

As a consequence, all the agents represented by the same node in level $L_t$ of the history tree must have the same state (and give the same output) at the end of round~$t$, regardless of the deterministic algorithm being executed. This is in agreement with the idea that the agents represented by the same node are indistinguishable.

\begin{corollary}\label{xcor:same}
At the end of any round~$t$, all agents represented by the same node in $L_t$ have the same internal state.
\end{corollary}
\begin{proof}
If two agents are represented by the same node of $L_t$, they have the same vista at round~$t$. Thus, by \cref{xth:view}, they have the same internal state at the end of round~$t$.
\end{proof}

\mypar{Significance.} The significance of the fundamental theorem is that it allows us to shift our focus from dynamic networks to history trees. Recall that there is a local algorithm $\mathcal A^\ast$ that allows agents to construct and update their vista at every round. Now, \cref{xth:view} guarantees that agents do not lose any information if they simply execute $\mathcal A^\ast$, regardless of their goal, and then compute their task-dependent outputs as a function of their respective vistas. Thus, in the following, we will assume without loss of generality that the internal state of every agent at every round, as well as all the messages it sends to all its neighbors, always coincide with its vista at that round.

\subsection{Related Concepts}\label{s:3.3}
We will now delve into related literature to place our history trees within the broader context of a line of research dating back to the 1980s.

\smallskip
\mypar{Graph coverings.} Borrowing from topological graph theory, Angluin was the first to use the notion of \emph{graph coverings} to prove that some problems cannot be solved in certain anonymous static networks~\cite{A80}. However, the conditions she stated were necessary but not sufficient.

\smallskip
\mypar{Views of Yamashita--Kameda.} Precise graph-theoretic characterizations of when certain fundamental problems for anonymous static networks are solvable were later given by Yamashita and Kameda~\cite{YK96}. The same authors also introduced the concept of \emph{view} of an agent in a static network whose links are endowed with \emph{port numbers} (i.e., the links incident to the same agent have distinct IDs)~\cite{YK88}. For Yamashita and Kameda, the view of an agent $p$ in a static network $G$ is an infinite tree rooted at $p$ that encodes all the information about the network that can possibly be gathered by $p$. In the language of topological graph theory, the view of $p$ is akin to the \emph{universal cover} of $G$ (i.e., the ``largest possible'' cover of $G$) rooted at $p$.

The views introduced by Yamashita and Kameda are loosely related to the vistas of history trees defined in \cref{s:3.1}, and their exact relationship has been detailed in~\cite{VIG24}. Notably, in static networks, one can construct an agent's vista at round~$t$ by inspecting the Yamashita--Kameda view of that agent truncated at depth~$t$, and vice versa. We stress that this correspondence holds only in static networks, as the Yamashita--Kameda views are not defined for dynamic networks.

If agents with isomorphic views are considered identical, one obtains the \emph{quotient graph} in the sense of Yamashita--Kameda. In the language of topological graph theory, this can also be characterized as the ``smallest possible'' graph that is covered by $G$.

\smallskip
\mypar{Minimum bases of Boldi--Vigna.} The above concepts were later extended by Boldi and Vigna to static networks with directed links and no port numbers~\cite{BV02}. In particular, they recognized that the appropriate topological tool for these networks is not the graph covering but the more general \emph{graph fibration}. In the language of Boldi--Vigna, the \emph{minimum base} $\widehat G$ of a static network $G$ is a structure analogous to the ``quotient graph'' of Yamashita--Kameda.

It is straightforward to construct the minimum base (equivalently, the quotient graph) of a static network $G$ given its history tree $\mathcal H$. In fact, the nodes in the level $L_t$ of $\mathcal H$ correspond to the isomorphism classes of the Yamashita--Kameda views truncated at round~$t$. It is easy to see that, if $G$ is a static network, the number of nodes in the levels of $\mathcal H$ strictly increases at every level until it becomes maximum, say, at level $L_s$. Now, the nodes of the minimum base $\widehat G$ are precisely the nodes of $L_s$, and their edges are given by the red edges between $L_s$ and $L_{s+1}$ (essentially, the endpoint in $L_{s+1}$ of each red edge is redirected to its parent in $L_s$).

In summary, previous structures and theories related to anonymous static networks can effectively be reinterpreted within the framework of history trees. The advantage of using history trees is that they also incorporate information about the time at which specific classes of agents become distinguishable. This ``temporal dimension'' makes history trees an ideal tool for dynamic networks, where one cannot rely on topological regularity to infer temporal information.

\smallskip
\mypar{Algorithms for static networks.} From a computational standpoint, it is important to consider how an agent in a network can algorithmically construct the aforementioned data structures and use them effectively to solve certain problems.

We already pointed out that the views of Yamashita--Kameda contain all the information that can be gathered locally by the agents; the same can be said of the minimum bases of Boldi--Vigna. Hence, constructing these structures allows an agent to access all the information that it can possibly use for computations.

Yamashita and Kameda showed that, at round~$t$, any agent can algorithmically construct its view truncated at depth~$t$. Thus, it can incrementally construct its view level by level. Moreover, they showed that, in a static network of $n$ agents, if two agents have isomorphic views up to depth~$n^2$, they have isomorphic views at any depth~\cite{YK88}. This leads to the conclusion that, if the number $n$ is known to the agents, they can effectively solve general problems and terminate in $O(n^2)$ communication rounds.

For example, since agents with isomorphic views are always indistinguishable, it is possible to elect a unique leader in a static network in a finite number of rounds if and only if there is a unique agent whose view truncated at depth~$n^2$ is distinct from all others.

\smallskip
\mypar{Linear bounds.} The $n^2$ upper bound was later improved by Norris, who showed that truncating views at depth~$n-1$ is sufficient and may be necessary in some networks. That is, if two agents have isomorphic views up to depth~$n-1$, they have isomorphic views at any depth~\cite{norris}.

Although this bound is worst-case optimal, it is far from being optimal in most cases. Later, Fraigniaud and Pelc further improved Norris' upper bound to $\widehat n-1$, where $\widehat n$ is the number of nodes in the minimum base $\widehat G$ of the network, and of course $\widehat n\leq n$~\cite{pelc}.

Proving these statements is immediate if one reframes them in terms of history trees. We have already argued that, in the history tree of a static network, the number of nodes in a level increases by at least one unit per level until level $L_s$, where it becomes maximum. That is, $|L_s|=\widehat n$, and clearly $s< \widehat n$, because $|L_{0}|\geq 1$. Thus, if two agents have isomorphic vistas (in the history tree) at round~$\widehat n-1$, they have isomorphic vistas at all rounds, which is equivalent to the statement of Fraigniaud and Pelc.

\smallskip
\mypar{Size of a view.} When designing efficient algorithms, one also wishes to optimize memory usage and message sizes, as well as running times. Unfortunately, in the worst case, a view of Yamashita--Kameda truncated at depth~$t$ grows exponentially in $t$, making their approach unsuitable for systems where memory limitation is an issue.

For this reason, Tani proposed a method to efficiently compress a truncated view by identifying isomorphic subtrees within the view itself. Using Tani's technique, a compressed view (of an undirected network with port numbers) truncated at depth $t$ can be stored in $O(tnM \log M)$ bits, where $M$ is the maximum degree of an agent in the network~\cite{tani}.

This essentially matches the upper bound on the size of a vista of a history tree of a $\tau$-union-connected network at round~$t$, which is $O(tn^2\log \tau M)$ bits (the small discrepancy is due to the fact that our networks are modeled by disconnected dynamic multigraphs, as opposed to connected static simple graphs). Thus, for instance, our linear-time terminating Counting algorithm based on history trees requires $O(n^3\log \tau M)$ bits of memory per agent in the worst case.

\smallskip
\mypar{Lamport causality.} In distributed systems, certain events are often caused by the occurrence of previous events; the concept of \emph{causal influence} was first introduced by Lamport in his seminal paper~\cite{lamport}.

A similar notion of causal influence can be formulated within the framework of history trees by stating that an agent $p$ exerts a causal influence on an agent $q$ at round~$t$ if the vista of $q$ at round~$t$ includes a node $v$ representing $p$. It is important to note that not all agents represented by $v$ necessarily initiate a sequence of messages reaching $q$ by round~$t$. However, at least one agent represented by $v$ does so, and the specific identity of this agent is irrelevant, since agents are anonymous.

\smallskip
\mypar{Causality in dynamic networks.} The concept of causal influence was further developed by Kuhn et al., who studied it in the context of dynamic networks with unique IDs~\cite{KLO10}. However, their findings apply equally well to anonymous dynamic networks.

Kuhn et al.\ established that, in a 1-union-connected dynamic network, the set of agents that are causally influenced by a given agent grows by at least one unit at every round, until all $n$ agents have been influenced. We can rephrase this observation in the language of history trees as follows.
\begin{itemize}
\item The vista of any agent at any round $t$ contains nodes representing at least $\min\{t+1, n\}$ agents.
\item For any agent $p$ and any round $t$, there are at least $\min\{t+1, n\}$ agents whose vista at round~$t$ contains a node representing $p$.
\end{itemize}

\smallskip
\mypar{Broadcasting speed.} As a consequence of the above, broadcasting a message, i.e., forwarding a message from a single agent to the whole dynamic network, takes at most $n-1$ rounds. Therefore, the \emph{dynamic diameter} $d$, defined as the maximum number of rounds it takes to forward a message from an agent to any other agent at any point in time, satisfies $d\leq n-1$ in any 1-union-connected dynamic network.

Since these facts will be used in the design of our algorithms, we give self-contained proofs below.

\begin{lemma}\label{xl:propamain} Let $P$ be a set of agents in a 1-union-connected dynamic network of size $n$, such that $1\leq |P|\leq n-1$, and let $t\geq 0$. Then, at every round~$t'\geq t+|P|$, in the vista of every agent there is a node at level $L_t$ representing at least one agent not in $P$.
\end{lemma}
\begin{proof}
Let $Q$ be the complement of $P$ (note that $Q$ is not empty), and let $Q_{t+i}$ be the set of agents represented by the nodes in $L_{t+i}$ whose vista contains a node in $L_t$ representing at least one agent in $Q$. We will prove by induction that $|Q_{t+i}|\geq |Q|+i$  for all $0\leq i\leq |P|$. The base case holds because $Q_t=Q$. The induction step is implied by $Q_{t+i}\subsetneq Q_{t+i+1}$, which holds for all $0\leq i < |P|$ as long as $|Q_{t+i}|<n$. Indeed, because $G_{t+i+1}$ is connected, it must contain a link between an agent $p\in Q_{t+i}$ and an agent $q\notin Q_{t+i}$. Thus, the vista of $q$ at round~$t+i+1$ contains the vista of $p$ at round~$t+i$, and so $Q_{t+i}\subsetneq Q_{t+i}\cup\{q\}\subseteq Q_{t+i+1}$.

Now, plugging $i:=|P|$, we get $|Q_{t+|P|}|=n$. In other words, the vista of each node in $L_{t+|P|}$ (and hence in subsequent levels) contains a node at level $L_t$ representing an agent in $Q$.
\end{proof}

\begin{corollary}\label{xl:propa}
In the history tree of a 1-union-connected dynamic network, every node at level $L_t$ is in the vista of every node at level $L_{t'}$, for all $t'\geq t+n-1$.
\end{corollary}
\begin{proof}
Let $v\in L_t$, and let $P$ be the set of agents not represented by $v$. If $P$ is empty, then all nodes in $L_{t'}$ are descendants of $v$, and have $v$ in their vista. Otherwise, $1\leq|P|\leq n-1$, and \cref{xl:propamain} implies that $v$ is in the vista of all nodes in $L_{t'}$.
\end{proof}

\section{Leaderless and Stabilizing Algorithms}\label{s:intermediate}
We will now present a general technique based on history trees, as well as three applications. Namely, we will give linear-time stabilizing and terminating algorithms for computing the Input Frequency function $F_{IF}$ in leaderless networks, as well as a linear-time stabilizing algorithm for computing the Input Multiset function $F_{IM}$ in networks with leaders. As we pointed out in \cref{s:2.2}, computing $F_{IF}$ or $F_{IM}$ is sufficient to compute the entire class of frequency-based functions or multiset-based functions, respectively, in the same number of rounds. 

This basic technique, however, falls significantly short when it comes to developing terminating algorithms for computing $F_{IM}$ in networks with leaders. This far more challenging problem will be discussed in \cref{s:termleader}.

Leveraging the theory developed in \cref{s:3.2}, we will assume, without any loss of generality, that all agents maintain their current vista of the history tree as their internal state and broadcast this vista across all available links at every round.

\subsection{Basic Technique}\label{s:technique}
Our basic technique makes use of the procedure in \cref{l:algorithm1}, which will be a subroutine of all algorithms presented in this section. The purpose of this procedure is to construct a homogeneous system of $k-1$ independent linear equations involving the anonymities of all the $k$ nodes in a level of an agent's vista (recall that a linear system is \emph{homogeneous} if all its constant terms are zero). We will first give some definitions.

\begin{figure}
\centering
\includegraphics[scale=0.65]{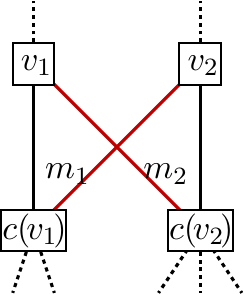}
\caption{The nodes $v_1$ and $v_2$ are \emph{exposed} with multiplicities $m_1$ and $m_2$, respectively.}
\label{fig:exposed}
\end{figure}

\smallskip
\mypar{Exposed nodes and strands.} In (a vista of) a history tree, if a node $v\in L_t$ has exactly one child (i.e., there is exactly one node $v'\in L_{t+1}$ such that $\{v,v'\}$ is a black edge), we say that $v$ is \emph{non-branching}. If $v$ has multiple children, it is \emph{branching} (thus, a leaf in a vista is neither branching nor non-branching). We emphasize that, if $v$ is a node of a vista $\mathcal V$ of a history tree $\mathcal H$, then $v$ may be branching in $\mathcal H$ but non-branching in $\mathcal V$. This happens if $v$ has multiple children in $\mathcal H$, but only one of them appears in $\mathcal V$.

We say that two non-branching nodes $v_1,v_2\in L_t$, whose respective children are $v'_1,v'_2\in L_{t+1}$, are \emph{exposed} with multiplicity $(m_1, m_2)$ if the red edges $\{v'_1, v_2\}$ and $\{v'_2, v_1\}$ are present with multiplicities $m_1\geq 1$ and $m_2\geq 1$, respectively (see \cref{fig:exposed}). Again, the same two nodes may be exposed in a vista of a history tree, but not in the history tree itself.

A \emph{strand} is a path $(w_1, w_2, \dots, w_k)$ in (a vista of) a history tree consisting of non-branching nodes such that, for all $1\leq i<k$, the node $w_i$ is the parent of $w_{i+1}$. We say that two strands $P_1$ and $P_2$ are \emph{exposed} if there are two exposed nodes $v_1\in P_1$ and $v_2\in P_2$.

Thanks to the following lemma, if we know the anonymity of a node in an exposed pair, we can determine the anonymity of the other node. (Recall that we denote the anonymity of $v$ by $a(v)$.)

\begin{lemma}\label{xl:guess1}
In a history tree $\mathcal H$, if the nodes $v_1$ and $v_2$ are exposed with multiplicity $(m_1, m_2)$, then $a(v_1)\cdot m_1=a(v_2)\cdot m_2$.
\end{lemma}
\begin{proof}
Let $v_1, v_2\in L_t$, and let $P_1$ and $P_2$ be the sets of agents represented by $v_1$ and $v_2$, respectively. Since $v_1$ is non-branching in $\mathcal H$, we have $a(c(v_1))=a(v_1)$, and therefore $c(v_1)$ represents $P_1$, as well. Hence, the number of links between $P_1$ and $P_2$ in $G_{t+1}$ (counted with their multiplicities) is $a(c(v_1))\cdot m_1 = a(v_1)\cdot m_1$. By a symmetric argument, this number is equal to $a(v_2)\cdot m_2$.
\end{proof}

Observe that \cref{xl:guess1} may not hold in a vista $\mathcal V$ of $\mathcal H$, because $v_1$ (or $v_2$) may be non-branching in $\mathcal V$ but have multiple children in $\mathcal H$. Thus, it is not necessarily true that $v_1$ (or $v_2$) and its unique child in $\mathcal V$ have the same anonymity.

\smallskip
\mypar{Main subroutine.} Intuitively, the procedure in \cref{l:algorithm1} searches for a long-enough sequence of levels in the given vista $\mathcal V$, say from $L_s$ to $L_t$, where all nodes are non-branching. That is, the nodes in $L_s \cup L_{s+1}\cup \dots \cup L_t$ can be partitioned into $k=|L_s|=|L_t|$ strands. Then the procedure searches for pairs of exposed strands, each of which yields a linear equation involving the anonymities of some nodes of $L_t$, until it obtains $k-1$ linearly independent equations. The reason why we have to consider strands spanning several levels of the history tree (as opposed to looking at a single level) is that the dynamic disconnectivity $\tau$ is not known, and thus \cref{p:time} cannot be applied directly. Note that the search may fail (in which case \cref{l:algorithm1} returns $t=-1$) or it may produce incorrect equations. The following lemma specifies sufficient conditions for \cref{l:algorithm1} to return a correct system of non-trivial equations for some $t\geq 0$.

\lstset{style=mystyle}
\begin{lstlisting}[caption={Constructing a system of equations in the anonymities of some nodes in a vista.\label{l:algorithm1}},captionpos=t,float,abovecaptionskip=-\medskipamount,mathescape=true]
# Input: a vista $\mathcal V$ with levels $L_{-1}$, $L_0$, $L_1$, $\dots$, $L_h$
# Output: $(t, S)$, where $t$ is an integer and $S$ is a system of linear equations

Assign $s := 0$
For $t := 0$ to $h$
   If $L_t$ contains a node with no children, return $(-1,\emptyset)$
   If $L_t$ contains a node with more than one child, assign $s := t+1$
   Else
      Let $k=|L_s|=|L_t|$ and let $u_1$, $u_2$, $\dots$, $u_k$ be the nodes in $L_t$
      Let $\mathcal P=\{P_1, P_2, \dots, P_k\}$, where $P_i$ is the strand starting in $L_s$ and ending at $u_i\in L_t$
      Let $G$ be the graph on $\mathcal P$ whose edges are pairs of exposed strands
      If $G$ is connected
         Let $G'\subseteq G$ be any spanning tree of $G$
         Assign $S := \emptyset$
         For each edge $\{P_i, P_j\}$ of $G'$
            Find any two exposed nodes $v_1\in P_i$ and $v_2\in P_j$
            Let $(m_1, m_2)$ be the multiplicity of the exposed pair $(v_1, v_2)$
            Add the equation $m_1 x_i = m_2 x_j$ to $S$
         Return $(t,S)$
\end{lstlisting}

\begin{lemma}\label{l:algo1corr}
Let $\mathcal V$ be the vista of an agent in a $\tau$-union-connected network of size $n$ taken at round $t'$, and let \cref{l:algorithm1} return $(t,S)$ on input $\mathcal V$. Assume that one of the following conditions holds:
\begin{enumerate}
\item $t\geq 0$ and $t'\geq t+\tau(n-1)+1$, or
\item $t'\geq \tau(2n-c)$, where $c$ is the number of distinct inputs held by agents at round $0$.
\end{enumerate}
Then, $0\leq t\leq \tau(n-c+1)-1$, and $S$ is a homogeneous system of $k-1$ independent linear equations (with integer coefficients) in $k=|L_{t}|$ variables $x_1$, $x_2$, \dots, $x_k$. Moreover, $S$ is satisfied by assigning to $x_i$ the anonymity of the $i$th node of $L_{t}$, for all $1\leq i\leq k$.
\end{lemma}
\begin{proof}
If $\tau=1$, it takes at most $n-1$ rounds for information to travel from an agent to any other agent, due to \cref{xl:propa}. In general, if $\tau\geq 1$, it takes at most $\tau(n-1)$ rounds, by \cref{p:time}. Therefore, since $\mathcal V$ is a vista taken at round $t'$, all levels of $\mathcal V$ up to $L_{t'-\tau(n-1)}$ are complete (recall that a level of a vista is complete if it has no missing nodes).

Assume Condition~2 first. Since $t'\geq \tau(2n-c)$, all levels of $\mathcal V$ up to $L_{\tau(n-c+1)}$ are complete. Thus, level $L_0$ is complete, and therefore it has exactly $c$ nodes, because two agents are distinguishable at round~$0$ if and only if they have distinct inputs. Since the sum of the anonymities of these $c$ nodes is $n$, there may be at most $n-c$ branching nodes in $\mathcal V$ (excluding the root). Hence, the levels from $L_0$ up to $L_{\tau(n-c+1)-1}$ contain an interval of at least $\tau$ consecutive levels, say from $L_r$ to $L_{r+\tau-1}$, where all nodes are non-branching and can be partitioned into $k=|L_r|=|L_{r+\tau-1}|$ strands $P_1$, $P_2$, \dots, $P_{k}$ (Lines~9--10).

Note that a link between two agents at any round $r'$ in the interval $[r+1,r+\tau]$ determines a pair of exposed nodes in $L_{r'-1}$. Thus, by definition of $\tau$-union-connected network, the graph of exposed strands between $L_r$ and $L_{r+\tau-1}$ (constructed as $G$ in Line~11) is connected. It follows that the execution of \cref{l:algorithm1} terminates at Line~19 (as opposed to Line~6) whenever $t\geq r+\tau-1$. Thus, the procedure returns a pair $(t,S)$ with $0\leq t\leq r+\tau-1\leq \tau(n-c+1)-1$. In particular, all levels of $\mathcal V$ up to $L_{t+1}$ are complete, since $t+1\leq \tau(n-c+1)$.

Now assume Condition~1. Since $t'\geq (t+1)+\tau(n-1)$, all levels of $\mathcal V$ up to $L_{t+1}$ are complete in this case, as well. Since $t\geq 0$ by assumption, the execution of \cref{l:algorithm1} terminates at Line~19. The termination condition is met when long-enough strands are found; as proved above, this event must occur when $t\leq \tau(n-c+1)-1$.

We have proved that, in both cases, the inequalities $0\leq t\leq \tau(n-c+1)-1$ hold, and all levels of $\mathcal V$ up to $L_{t+1}$ are complete. Let us now examine the linear system $S$. Observe that $S$ is homogeneous because it consists of homogeneous linear equations (cf.~Line~18). Also, since the spanning tree $G'$ constructed at Line~13 has $k-1$ edges, $S$ contains $k-1$ equations. We will prove that they are linearly independent by induction on $k$. If $k=1$, there is nothing to prove. Otherwise, let $P_i$ be a leaf of $G'$, and let $\{P_i, P_j\}$ be its incident edge. Then, $S$ contains an equation $Q$ of the form $m_1 x_i = m_2 x_j$ with $m_1m_2\neq 0$. Let $S'$ be the system obtained by removing $Q$ from $S$; equivalently, $S'$ corresponds to the tree obtained by removing the leaf $P_i$ from $G'$. By the inductive hypothesis, no linear combination of equations in $S'$ yields $0=0$. On the other hand, if $Q$ is involved in a linear combination with a non-zero coefficient, then the variable $x_i$ cannot vanish, because it only appears in $Q$. Therefore, the equations in $S$ are independent.

It remains to prove that a solution to $S$ is given by the anonymities of the nodes of $L_t$. Due to \cref{xl:guess1}, if $v_1$ and $v_2$ are exposed in $\mathcal V$, as well as in the history tree containing $\mathcal V$, with multiplicity $(m_1, m_2)$, then $m_1 a(v_1)=m_2 a(v_2)$. To conclude our proof, it is sufficient to note that, since the nodes of a strand $P_i$ are non-branching in $\mathcal V$ as well as in the underlying history tree (recall that all levels of $\mathcal V$ up to $L_{t+1}$ are complete), they all have the same anonymity, which is the anonymity of the ending node $w_i\in L_t$.
\end{proof}

\mypar{Diameter bounds.} We can also restate \cref{l:algo1corr} with respect to the dynamic diameter $d$ of the network.

\begin{corollary}\label{c:diameter}
The statement of \cref{l:algo1corr} remains valid if the two conditions are replaced with:
\begin{enumerate}
\item $t\geq 0$ and $t'\geq t+d+1$, or
\item $t'\geq \tau(n-c+1)+d$, where $c$ is the number of distinct inputs held by agents at round $0$,
\end{enumerate}
where $d$ is the dynamic diameter of the network (or an upper bound thereof).
\end{corollary}
\begin{proof}
Note that it takes at most $d$ rounds for information to travel across the network. Hence, the proof of \cref{l:algo1corr} can be repeated verbatim, substituting the term $\tau(n-1)$ with $d$ throughout.
\end{proof}

\subsection{Stabilizing Algorithm for Leaderless Networks}\label{s:stableaderless}
As a first application of \cref{l:algo1corr}, we will give stabilizing algorithm that efficiently computes the Input Frequency function $F_{IF}$ in all leaderless networks with a finite dynamic disconnectivity $\tau$, assuming no knowledge of $\tau$ or $n$. As a consequence, \emph{all} frequency-based functions are efficiently computable as well, due to \cref{th:concentration}. Moreover, \cref{t:scale} states that no other functions are computable in leaderless networks, and \cref{t:lower1} shows that our algorithm is asymptotically optimal.

\begin{theorem}\label{t:noleadstab}
There is an algorithm that computes the Input Frequency function $F_{IF}$ on all networks in $\bigcup_{\tau\geq 1}\bigcup_{n\geq 1}\mathcal N^\tau_{n,0}$ and stabilizes in at most $\tau(2n-2)$ rounds.
\end{theorem}
\begin{proof}
Our local algorithm is as follows. Each agent $p$ runs \cref{l:algorithm1} on its own vista $\mathcal V$, obtaining a pair $(t, S)$. If $t=-1$ or $S$ is not a homogeneous system of $k-1$ independent linear equations in $k$ variables, then $p$ outputs $\{(\lambda(p),1)\}$, where $\lambda(p)$ is its own input. Otherwise, since the rank of the coefficient matrix of $S$ is $k-1$, the general solution to $S$ has exactly one free parameter, due to the Rouch\'e--Capelli theorem. Therefore, by Gaussian elimination, it is possible to express every variable $x_i$ as a rational multiple of $x_1$, i.e., $x_i=\alpha_i x_1$ for some $\alpha_i\in \mathbb Q^+$ (recall that the coefficients of $S$ are integers). Let $L_{t}=\{w_1, w_2, \dots, w_k\}$ and $L_0=\{v_1, v_2, \dots, v_{k'}\}$. For every node $v_i\in L_0$, define $\beta_i\in\mathbb Q^+$ as $\beta_i=\sum_{w_j\in L_t \text{ descendant of }v_i} \alpha_j$, and let $\beta = \sum_i \beta_i$. Then, $p$ outputs
$$\{(\text{label}(v_1), \beta_1/\beta), (\text{label}(v_2), \beta_2/\beta), \dots, (\text{label}(v_{k'}), \beta_{k'}/\beta)\}.$$

The correctness and stabilization time of the above algorithm directly follow from \cref{l:algo1corr}. Specifically, if not all agents have the same input (hence $n\geq c\geq 2$), then, at any round $t'$ such that $t' \geq \tau(2n-2)\geq \tau(2n-c)$, Condition~2 of \cref{l:algo1corr} is met, and the system $S$ is satisfied by the anonymities of the nodes in $L_t$. Thus, $a(v_i)=\alpha_i a(v_1)$ for all $v_i\in L_0$, and therefore $\beta_i/\beta= a(v_i)/n$. We conclude that, for any input assignment $\lambda$, the algorithm stabilizes on the correct output $\frac 1n \cdot \mu_\lambda$ within $\tau(2n-2)$ rounds.

In the special case where all agents have the same input, the default output $\{(\lambda(p),1)\}$ is correct since round~$0$. Also, whenever $S$ is a homogeneous system of $k-1$ independent linear equations, the algorithm still returns the correct output $\{(\lambda(p),1)\}$, because all nodes in level $L_0$ have the same label $\lambda(p)$. Hence, in this particular case the output is always correct.
\end{proof}

\subsection{Terminating Algorithm for Leaderless Networks}\label{s:termleaderless}
We will now give a certificate of correctness that can be used to turn the stabilizing algorithm of \cref{s:stableaderless} into a terminating algorithm. The certificate relies on a-priori knowledge of the dynamic disconnectivity $\tau$ and an upper bound $N$ on the size of the network $n$; these assumptions are justified by \cref{o:timeimposs2} and \cref{t:knowN}, respectively. Again, \cref{t:lower1} shows that our algorithm is asymptotically optimal.

\begin{theorem}\label{t:noleadterm}
For every $\tau\geq 1$ and $N\geq 1$, there is an algorithm that computes the Input Frequency function $F_{IF}$ on all networks in $\bigcup_{n\leq N} \mathcal N^{\tau}_{n,0}$ and terminates in at most $\tau(n+N-2)$ rounds.
\end{theorem}
\begin{proof}
Given $\tau$ and $N$, our terminating local algorithm is as follows. Run \cref{l:algorithm1} on the agent's vista $\mathcal V$, obtaining a pair $(t, S)$, and then do the same computations as in the algorithm for \cref{t:noleadstab}. If $t\geq 0$ and the current round $t'$ satisfies $t'\geq t+\tau(N-1)+1$, then the output is certifiably correct, and the agent terminates. Moreover, if $t'=\tau(N-1)$ and the level $L_0$ of $\mathcal V$ contains a single node, then the agent outputs $\{(\lambda(p),1)\}$ and terminates, where $\lambda(p)$ is the agent's own input.

Due to \cref{xl:propa}, the latter condition is verified only if all agents have the same input (because level $L_0$ of $\mathcal V$ must be complete at round $\tau(N-1)\geq \tau(n-1)$). In this case, it is safe to conclude that the correct output is $\{(\lambda(p),1)\}$, and termination occurs within $\tau(N-1)\leq \tau(n+N-2)$ rounds, as desired (obviously, $n\geq 1$ holds).

In all other cases, the correctness of the algorithm is a direct consequence of \cref{l:algo1corr}. Indeed, if the algorithm terminates when $t\geq 0$ and $t'\geq t+\tau(N-1)+1\geq t+\tau(n-1)+1$, it gives the correct output because Condition~1 of \cref{l:algo1corr} is met. Thus, the algorithm cannot terminate with an incorrect output.

As for the running time, we may leverage \cref{l:algo1corr} with $c\neq 1$, since we have already shown that for $c=1$ termination occurs within $\tau(n+N-2)$ rounds. So, assume that $c\geq 2$ and $t'=\tau(n+N-2)\geq \tau(2n-2)\geq \tau(2n-c)$. Since Condition~2 of \cref{l:algo1corr} is met, we have $0\leq t\leq \tau(n-c+1)-1\leq \tau (n-1)-1$. Thus, $t'=\tau(n+N-2)\geq t+\tau(N-1)+1$, and the algorithm terminates at round $t'$, as desired.
\end{proof}

We can also trade the knowledge of $\tau$ and $N$ for the knowledge of the dynamic diameter $d$ of the network (or an upper bound thereof).

\begin{corollary}\label{c:noleadterm}
For every $d\geq 1$, there is an algorithm that computes the Input Frequency function $F_{IF}$ on all networks in $\bigcup_{n\geq 1} \mathcal N^{[d]}_{n,0}$ and terminates in at most $\tau(n-1)+d\leq \tau(2n-2)$ rounds, where $\tau$ is the dynamic disconnectivity of the network.
\end{corollary}
\begin{proof}
We can repeat the proof of \cref{t:noleadterm} verbatim, replacing the termination condition $t'\geq t+\tau(N-1)+1$ with $t'\geq t+d+1$ and the condition $t'=\tau(N-1)$ with $t'=d$. Then \cref{c:diameter} is invoked in lieu of \cref{l:algo1corr}, recalling that $d\leq \tau(n-1)$.
\end{proof}

Note that knowledge of an upper bound $D\geq d$ is actually enough for the algorithm in \cref{c:noleadterm} to work. Indeed, with the termination condition $t'\geq t+D+1$, the same algorithm computes $F_{IF}$ on all networks in $\bigcup_{d\leq D}\bigcup_{n\geq 1} \mathcal N^{[d]}_{n,0}$ and terminates in at most $\tau(n-1)+D$ rounds.

\subsection{Stabilizing Algorithm for Networks with Leaders}\label{s:stableader}
To conclude this section, we will give a stabilizing algorithm that efficiently computes the Input Multiset function $F_{IM}$ in all networks (of unknown size $n$) with a known number $\ell\geq 1$ of leaders and a finite (but unknown) dynamic disconnectivity $\tau$. Therefore, \emph{all} multiset-based functions are efficiently computable as well, due to \cref{xth:compl}. Moreover, \cref{t:aggr} states that no other functions are computable in networks with leaders, and \cref{t:lower2} shows that our algorithm is asymptotically optimal. We will once again make use of the subroutine in \cref{l:algorithm1}, this time assuming that the number of leaders $\ell\geq 1$ is known to all agents. This assumption is justified by \cref{t:knowL}.

\begin{theorem}\label{t:multileadstab}
For every $\ell\geq 1$, there is an algorithm that computes the Input Multiset function $F_{IM}$ on all networks in $\bigcup_{\tau\geq 1} \bigcup_{n\geq \ell} \mathcal N^{\tau}_{n,\ell}$ and stabilizes in at most $\tau(2n-2)$ rounds.
\end{theorem}
\begin{proof}
The algorithm proceeds as in \cref{t:noleadstab}, with two differences. First, the default output of an agent $p$, instead of being $\{(\lambda(p),1)\}$, is now $\{(\lambda(p),\ell)\}$ (recall that the function $F_{IM}$ does not return frequencies, but numbers of agents). Second, when the fractions $\beta_1$, $\beta_2$, \dots, $\beta_{k'}$ have been computed, as well as their sum $\beta$, the following additional steps are performed. Let $L_0=\{v_1, v_2, \dots, v_{k'}\}$, and let $\{v_{j_1}, v_{j_2}, \dots, v_{j_l}\} \subseteq L_0$ be the set of nodes in $L_0$ representing leader agents, i.e., such that $\text{label}(v_{j_i})$ has the leader flag set for all $1\leq i\leq l$ (observe that, in general, we have $l\leq \ell$, because some nodes of $L_0$ may represent more than one leader). If $l=0$, retain the default output. Otherwise, compute $\beta'=\sum_{i=1}^l \beta_{j_i}$ and $\gamma_i=\ell\beta_i/\beta'$ for all $1\leq i\leq k'$, and output
$$\{(\text{label}(v_1), \gamma_1), (\text{label}(v_2), \gamma_2), \dots, (\text{label}(v_{k'}), \gamma_{k'})\}.$$

The correctness follows from the fact that, as shown in \cref{t:noleadstab}, at any round $t'\geq \tau(2n-c)$ we have $\beta_i/\beta= a(v_i)/n$ for all $1\leq i\leq k'$. Adding up these equations for all $i\in\{j_1, j_2, \dots, j_l\}$, we obtain $\beta'/\beta=\ell/n$, and therefore $n=\ell\beta/\beta'$. We conclude that
$$\gamma_i=\frac{\ell\beta_i}{\beta'}=\frac{\ell\beta\beta_i}{\beta'\beta}=\frac{n\beta_i}{\beta}=a(v_i).$$
Thus, within $\tau(2n-c)$ rounds, the algorithm stably outputs the anonymities of all nodes in $L_0$. As observed in \cref{s:2}, this is equivalent to computing the Input Multiset function $F_{IM}$.

If $c\geq 2$, then the stabilization time is at most $\tau(2n-2)\geq \tau(2n-c)$ rounds, as desired. If $c=1$, all agents have the same input, including the same leader flag. Since $\ell\geq 1$, this implies that all agents are leaders and $\ell=n$; hence the default output $\{(\lambda(p),\ell)\}$ is correct at every round since round~$0$. Thus, the stabilization time is at most $\tau(2n-2)$ rounds in every case.
\end{proof}

\section{Terminating Algorithm for Networks with Leaders}\label{s:termleader}
We will now present the main result of this paper. As already remarked, giving an efficient certificate of correctness for the Input Multiset function with one or more leaders is a highly non-trivial task for which a radically new approach is required. In fact, there are two crucial difficulties to overcome.

Firstly, the strategies developed in \cref{s:intermediate} are too shallow and ineffective even for networks with a unique leader, which are the easiest to treat. In \cref{s:naive}, we will give counterexamples to some naive termination strategies that one may devise in an attempt to generalize those in \cref{s:intermediate}. This indicates that an entirely different technique is necessary.

Secondly, networks with more than one leader significantly add to the difficulty of the problem. For instance, once a terminating algorithm for networks with a unique leader has been designed, one may be tempted to simply adapt it to the multi-leader case by setting the anonymity of the leader node in the history tree to $\ell>1$ instead of $1$. Unfortunately, this approach neglects one important factor. Indeed, while the history tree of a network with a unique leader contains a single leader branch, all of whose nodes have anonymity $1$, this may not be the case in a network with multiple leaders. In such a network, as soon as some leaders get disambiguated, the leader node branches into several children nodes whose anonymities are unknown (we only know that their sum is $\ell$). Moreover, some leader branches may be missing from the vistas of other leaders for several rounds. Thus, in the case of multi-leader networks, we must deal with the fact that even the vista of a leader may have levels where the sum of the anonymities of any subset of nodes is unknown.

\cref{s:approxcount} contains the technical core of our algorithm. Here we develop a subroutine that, in $O(\tau\ell n)$ rounds, counts the number of agents in a network assuming that its history tree contains a leader node of known anonymity whose descendants are non-branching for sufficiently many rounds. In particular, in the case of networks with a unique leader, this subroutine yields a full-fledged linear-time terminating algorithm for the Counting problem (because in this case the leader nodes are necessarily non-branching).

In \cref{s:main}, we approach the general multi-leader problem by repeatedly guessing the anonymity of a leader node in the history tree and invoking the subroutine of \cref{s:approxcount} to confirm our guesses. This procedure introduces additional overhead, bringing the running time of our final algorithm to $O(\tau\ell^2n)$ rounds.

\subsection{Naive Termination Strategies Are Incorrect}\label{s:naive}
If we were to use a technique such as the one in \cref{s:stableader} to devise a terminating algorithm for the Counting problem (or, more generally, for computing the Input Multiset function), our attempts would be bound to fail.

\smallskip
\mypar{Basic naive strategy.} As a first example, consider the dynamic network in \cref{xfig:counter0}. Observe that the leader of this network always receives exactly four messages from indistinguishable non-leader agents. In turn, the agents that send messages to the leader in the first four rounds are still unaware of the agents labeled $p_3$. As a result, the vista of the leader up to round~$4$ consists of only two strands whose nodes are exposed with multiplicity $(4,1)$ at every level. Thus, according to the algorithm in \cref{s:stableader}, the leader assigns an anonymity of $4$ to all the non-leader nodes in its vista, concluding that there are only five agents in the network. Essentially, for the first four rounds, the leader cannot distinguish this network from a static star graph with the leader at the center.

\begin{figure}
\centering
\includegraphics[scale=0.5]{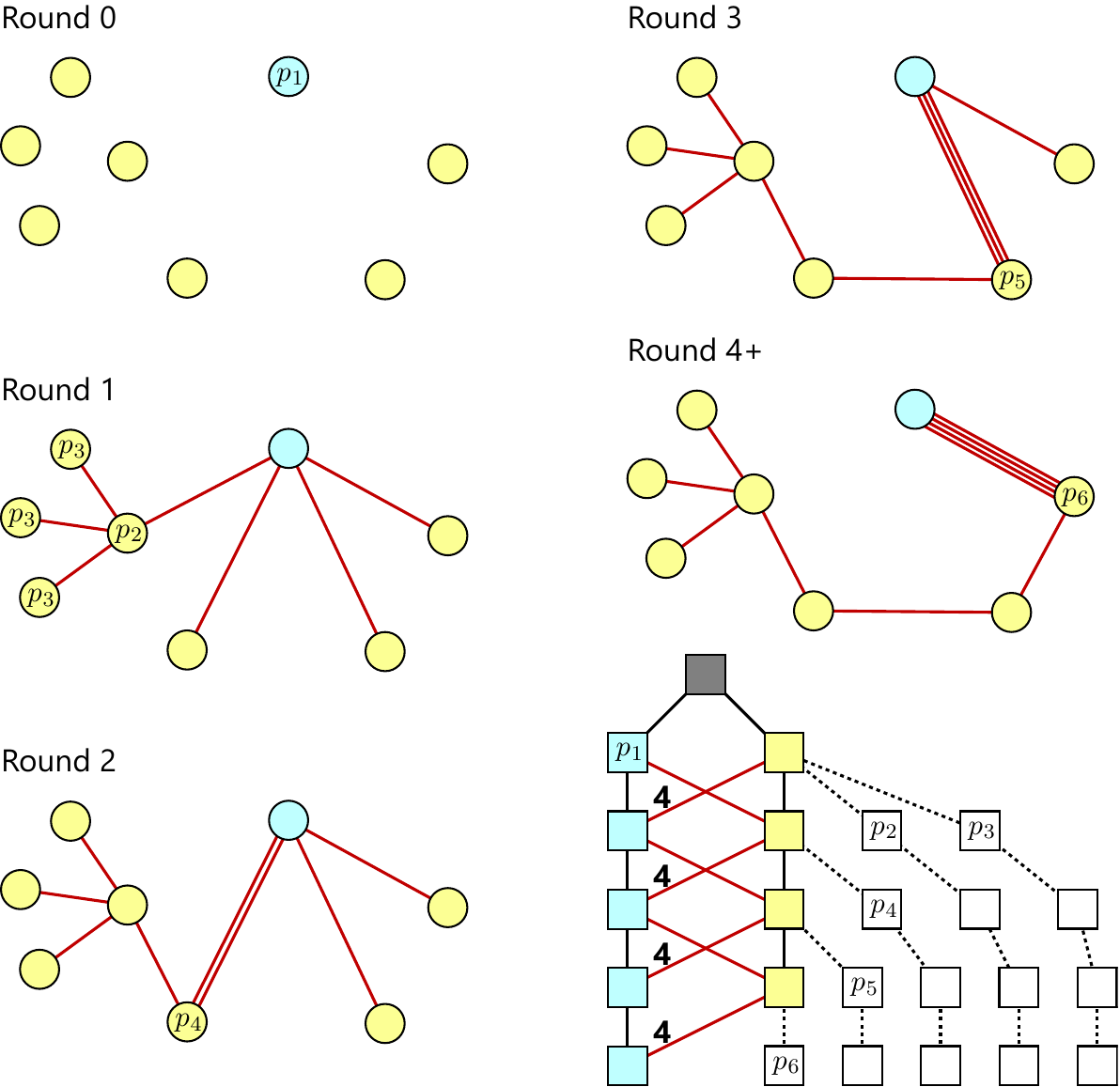}
\caption{An example of a dynamic network where the naive technique of \cref{s:stableader} fails to provide a correct termination condition. The white nodes in the history tree are not in the vista of the leader at the last round; the red edges not in the vista are not drawn. Same-colored agents have equal inputs. For the first four rounds, from the leader's perspective, this network is indistinguishable from the complete bipartite graph $K_{1,4}$. Thus, throughout this time, the leader is unaware of the agents labeled $p_3$, and therefore cannot compute the total number of agents.}
\label{xfig:counter0}
\end{figure}

It is straightforward to generalize this example to networks with $n=k_1+k_2+1$ agents, $k_1+1$ of which are counted by the leader, while the other $k_2$ are not discovered by the leader until round~$k_1+1$ (such as the agents labeled $p_3$ in \cref{xfig:counter0}). In these networks, our naive algorithm consistently returns $k_1+1$ for several rounds, which can be made arbitrarily far from the true value $k_1+k_2+1$ by increasing $k_2$ indefinitely.

\smallskip
\mypar{Improved naive strategy.} One may conjecture that a good termination certificate would be to compute the number of agents $n'$ according to the algorithm in \cref{s:stableader}, and then wait for a number of rounds depending on $n'$ to confirm that no relevant nodes were missing from the vista. In fact, in the previous example, simply waiting for $n'$ rounds would suffice.

\begin{figure}
\centering
\includegraphics[scale=0.5]{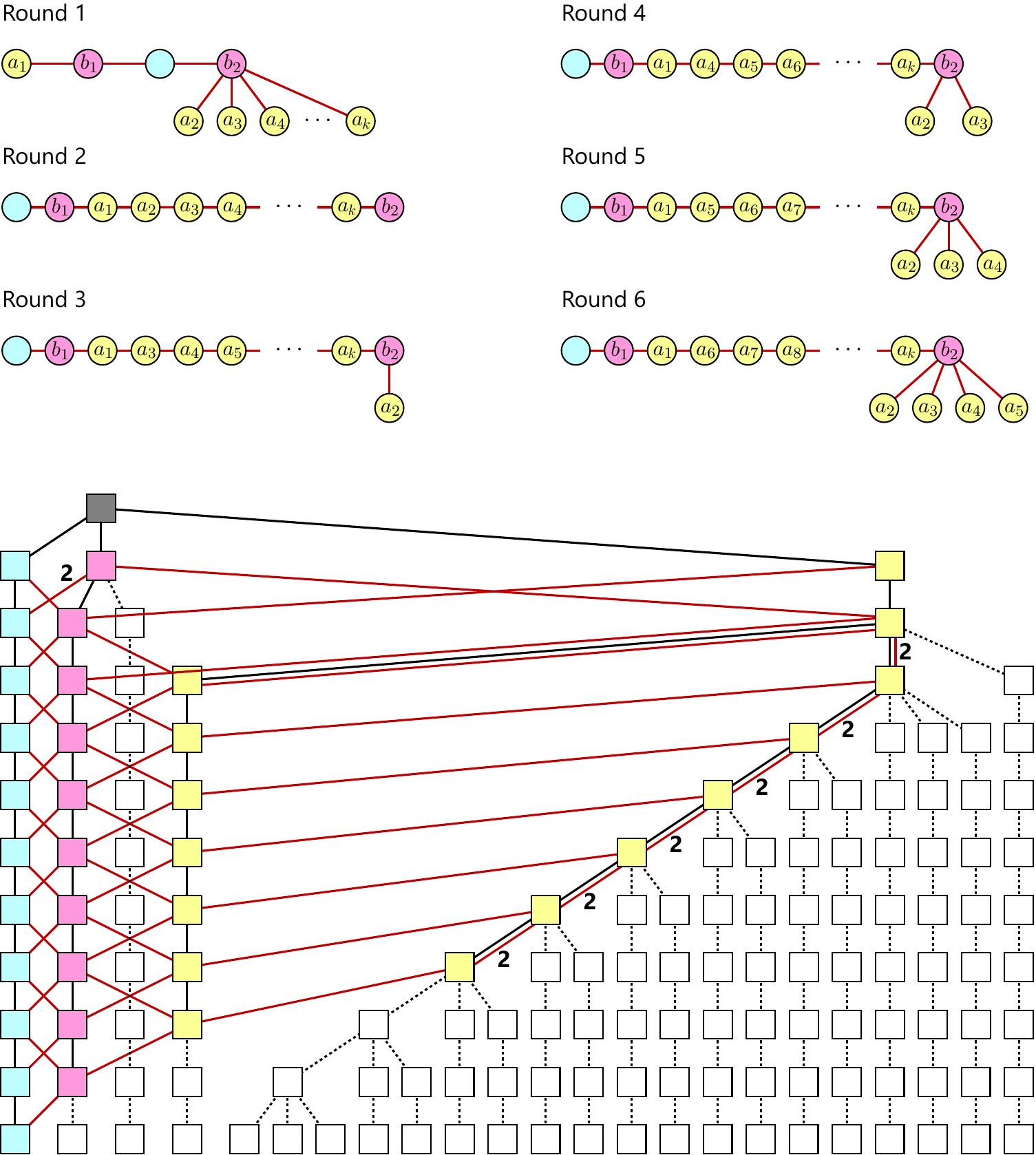}
\caption{A dynamic network where, after level $L_1$, all levels in the leader's vista are identical for an arbitrarily long sequence of rounds (depending on the parameter $k$). The color schemes and stylistic conventions are as in \cref{xfig:counter0}.}
\label{xfig:counter1}
\end{figure}

Unfortunately, this strategy fails on the network in \cref{xfig:counter1}. Here the leader's vista at levels $L_0$ and $L_1$ causes the algorithm to count only $n'=5$ agents (i.e., one leader, two agents represented by the purple node, and two agents represented by the yellow node). Afterwards, the leader has to wait until round $k-5$ for the appearance of the node that was missing from $L_1$. Since $k$ is arbitrary, this type of strategy is bound to fail no matter what the waiting time is.

\smallskip
\mypar{Challenges.} Essentially, the recurring issue is that an agent seems to have no way of knowing whether any level in its vista is complete; thus, it may end up terminating too soon with an incorrect output. To formulate a correct termination condition, we will have to considerably develop the theory of history trees. This will be undertaken in \cref{s:main,s:approxcount}.

\subsection{Main Algorithm}\label{s:main}
\mypar{The subroutine \texttt{ApproxCount}.} We first introduce the subroutine \texttt{ApproxCount}, whose formal description and proof of correctness are postponed to \cref{s:approxcount}. The purpose of \texttt{ApproxCount} is to compute an approximation $n'$ of the total number of agents $n$ (or report various types of failure). It takes as input a vista $\mathcal V$ of an agent, the number of leaders $\ell$, and two integer parameters $s$ and $x$, representing the index of a level of $\mathcal V$ and the anonymity of a leader node in $L_s$, respectively.

\smallskip
\mypar{Discrepancy $\delta$.} Suppose that \texttt{ApproxCount} is invoked with arguments $\mathcal V$, $s$, $x$, $\ell$, where $1\leq x\leq \ell$, and let $\vartheta$ be the first leader node in level $L_s$ of $\mathcal V$ (if $\vartheta$ does not exist, the procedure immediately returns the error message $n'=\texttt{"MissingNodes"}$). We define the \emph{discrepancy} $\delta$ as the ratio $x/a(\vartheta)$. Clearly, $\delta\leq \ell$. Note that, since $a(\vartheta)$ is not a-priori known by the agent executing \texttt{ApproxCount}, then neither is $\delta$.

\smallskip
\mypar{Conditional anonymity.} \texttt{ApproxCount} starts by assuming that the anonymity of $\vartheta$ is $x$, and makes deductions on other anonymities based on this assumption. Thus, we will distinguish between the actual anonymity of a node $a(v)$ and the \emph{conditional anonymity} $a'(v)=\delta a(v)$ that \texttt{ApproxCount} may compute under the initial assumption that $a'(\vartheta)=x=\delta a(\vartheta)$.

\smallskip
\mypar{Overview of \texttt{ApproxCount}.} The procedure $\texttt{ApproxCount}$ scans the levels of $\mathcal V$ starting from $L_s$, making ``guesses'' on the conditional anonymities of nodes based on already known conditional anonymities. Generalizing some lemmas from~\cite{DV22}, we develop a criterion to determine when a guess is correct. This yields more nodes with known conditional anonymities, and therefore more guesses (the details are in \cref{s:approxcount}). As soon as it has obtained enough information, the procedure stops and returns $(n',t)$, where $L_t$ is the level scanned thus far. If the information gathered satisfies certain criteria, then $n'$ is an approximation of $n$. Otherwise, $n'$ is an error message, as detailed below.

\smallskip
\mypar{Error messages.} If $L_s$ contains no leader nodes, the procedure returns the error message $n'=\texttt{"MissingNodes"}$. If, before gathering enough information on $n$, the procedure encounters a descendant of $\vartheta$ with more than one child in $\mathcal V$, it returns the error message $n'=\texttt{"StrandTooShort"}$. If it determines that the conditional anonymity of a node is not an integer, it returns the error message $n'=\texttt{"WrongGuess"}$. Finally, if it determines that the sum $\ell'$ of the conditional anonymities of the leader nodes is not $\ell$, it returns $n'=\texttt{"MissingNodes"}$ if $\ell'<\ell$ and $n'=\texttt{"WrongGuess"}$ if $\ell'>\ell$.

\smallskip
\mypar{Correctness of \texttt{ApproxCount}.} The following lemma gives some conditions that guarantee that \texttt{ApproxCount} has the expected behavior; it also gives bounds on the number of rounds it takes for \texttt{ApproxCount} to produce an approximation $n'$ of $n$, as well as a criterion to determine if $n'=n$. The lemma's proof is rather lengthy and technical, and is found in \cref{s:approxcount}.

\begin{lemma}\label{l:approxcount}
Let \textup{$\texttt{ApproxCount}(\mathcal V, s, x, \ell)$} return $(n',t)$. Assume that $\vartheta$ exists and $x\geq a(\vartheta)$. Let $\vartheta'$ be the (unique) descendant of $\vartheta$ in $\mathcal V$ at level $L_{t}$, and let $L_{t'}$ be the last level of $\mathcal V$. Then:
\begin{itemize}
\item[(i)] If $x=a(\vartheta)=a(\vartheta')$, then $n'\neq \texttt{"WrongGuess"}$.
\item[(ii)] If $n'$ is not an error message and $a(\vartheta)=a(\vartheta')$, and if either $t'\geq t+n'$ or level $L_t$ is complete, then $n'=n$.
\item[(iii)] If $t'\geq s+(\ell+2)(n-1)$, then $s\leq t\leq s+(\ell+1)(n-1)$ and $n'\neq \texttt{"MissingNodes"}$. Moreover, if $n'=\texttt{"StrandTooShort"}$, then $L_t$ contains a leader node with at least two children in $\mathcal V$.\qed
\end{itemize}
\end{lemma}

Our terminating algorithm assumes that all agents know the number of leaders $\ell\geq 1$ and the dynamic disconnectivity $\tau$ (but no knowledge of $n$ is assumed). Again, this is justified by \cref{t:knowL} and \cref{o:timeimposs2}.

\begin{theorem}\label{t:multileadterm}
For every $\tau\geq 1$ and $\ell\geq 1$, there is an algorithm that computes the Input Multiset function $F_{IM}$ on all networks in $\bigcup_{n\geq \ell} \mathcal N^{\tau}_{n,\ell}$ and terminates in at most $\tau((\ell^2+\ell+1)(n-1)+1)$ rounds.
\end{theorem}
\begin{proof}
Due to \cref{p:time}, since $\tau$ is known and appears as a factor in the claimed running time, we can assume that $\tau=1$ without loss of generality. Also, note that determining $n$ is enough to compute $F_{IM}$. Indeed, if an agent determines $n$ at round~$t'$, it can wait until round~$\max\{t',\tau(2n-2)\}$ and run the algorithm in \cref{t:multileadstab}, which is guaranteed to give the correct output by that time.

In order to determine $n$ assuming that $\tau=1$, we let each agent run the algorithm in \cref{l:algorithm2} with input $(\mathcal V, \ell)$, where $\mathcal V$ is the vista of the agent at the current round $t'$. We will prove that this algorithm returns a positive integer (as opposed to \texttt{"Unknown"}) within $(\ell^2+\ell+1)(n-1)+1$ rounds, and the returned number is indeed the correct size of the system $n$.

\lstset{style=mystyle}
\begin{lstlisting}[caption={Solving the Counting problem with $\ell\geq 1$ leaders.\label{l:algorithm2}},captionpos=t,float,abovecaptionskip=-\medskipamount,mathescape=true]
# Input: a vista $\mathcal V$ and a positive integer $\ell$
# Output: either a positive integer $n$ or "Unknown"

Assign $n^\ast :=-1$ and $s:=0$ and $c:=0$
Let $b$ be the number of leader branches in $\mathcal V$
While $c\leq \ell-b$
   Assign $t^\ast := -1$
   For $x:=\ell$ downto $1$
      Assign $(n',t):=\text{ApproxCount}(\mathcal V, s, x, \ell)$      # see $\text{\cref{l:algorithm3}}$ in $\text{\cref{s:approxcount}}$
      Assign $t^\ast := \max\{t^\ast, t\}$
      If $n'=\;$"MissingNodes", return "Unknown"
      If $n'=\;$"StrandTooShort", break out of the for loop
      If $n'\neq\;$"WrongGuess"
         If $n^\ast=-1$, assign $n^\ast:=n'$
         Else if $n^\ast \neq n'$, return "Unknown"
         Assign $c:=c+1$ and break out of the for loop
   Assign $s := t^\ast+1$
Let $L_{t'}$ be the last level of $\mathcal V$
If $t'\geq t^\ast+n^\ast$, return $n^\ast$
Else return "Unknown"
\end{lstlisting}

\smallskip
\mypar{Algorithm description.} Let $b$ be the number of branches in $\mathcal V$ representing leader agents (Line~5). The initial goal of the algorithm is to compute $\ell-b+1$ approximations of $n$ using the information found in as many disjoint intervals $\mathcal L_1$, $\mathcal L_2$, \dots, $\mathcal L_{\ell-b+1}$ of levels of $\mathcal V$ (Lines~6--17).

If there are not enough levels in $\mathcal V$ to compute the desired number of approximations, or if the approximations are not all equal, the algorithm returns \texttt{"Unknown"} (Lines~11 and~15).

In order to compute an approximation of $n$, say in an interval of levels $\mathcal L_i$ starting at $L_s$, the algorithm goes through at most $\ell$ phases (Lines~8--16). The first phase begins by calling \texttt{ApproxCount} with starting level $L_s$ and $x=\ell$, i.e., the maximum possible value for the anonymity of a leader node (Line~9). Specifically, \texttt{ApproxCount} chooses a leader node in $\vartheta\in L_s$ and tries to estimate $n$ using as few levels as possible.

Let $(n',t)$ be the pair of values returned by \texttt{ApproxCount}. If $n'=\texttt{"MissingNodes"}$, this is evidence that $\mathcal V$ is still missing some relevant nodes, and therefore \texttt{"Unknown"} is immediately returned (Line~11). If $n'=\texttt{"StrandTooShort"}$, then a descendant of $\vartheta$ with multiple children in $\mathcal V$ was found, say at level $L_t$, before an approximation of $n$ could be determined. As this is an undesirable event, the algorithm moves $\mathcal L_i$ after $L_t$ and tries again to estimate $n$ (Line~12). If $n'=\texttt{"WrongGuess"}$, then $x$ may not be the correct anonymity of the leader node $\vartheta$ (see the description of \texttt{ApproxCount}), and therefore the algorithm calls \texttt{ApproxCount} again with the same starting level $L_s$, but now with $x=\ell-1$. If $n'=\texttt{"WrongGuess"}$ is returned again, then $x=\ell-2$ is tried, and so on. After all possible assignments down to $x=1$ have failed, the algorithm just moves $\mathcal L_i$ forward and tries again from $x=\ell$.

As soon as $n'$ is an integer (hence not an error message), it represents an approximation of $n$ that is stored in the variable $n^\ast$. If it is different from the previous approximations, then \texttt{"Unknown"} is returned (Line~15). Otherwise, the algorithm proceeds with the next approximation in a new interval of levels $\mathcal L_{i+1}$, and so on.

Finally, when $\ell-b+1$ approximations of $n$ (all equal to $n^\ast$) have been found, a correctness check is performed: the algorithm takes the last level~$L_{t^\ast}$ visited thus far; if the current round~$t'$ satisfies $t'\geq t^\ast+n^\ast$, then $n^\ast=n$ is accepted as correct and returned; otherwise \texttt{"Unknown"} is returned (Lines~18--20).

\smallskip
\mypar{Correctness and running time.} We will prove that, if the output of \cref{l:algorithm2} is not \texttt{"Unknown"}, then it is indeed the number of agents, i.e., $n^\ast=n$. Since the $\ell-b+1$ approximations of $n$ have been computed on disjoint intervals of levels, there is at least one such interval, say $\mathcal L_j$, where no leader node in the history tree has more than one child (because there can be at most $\ell$ leader branches). With the notation of \cref{l:approxcount}, this implies that $a(\vartheta)=a(\vartheta')$ whenever \texttt{ApproxCount} is called in $\mathcal L_j$. Also, since the option $x=\ell$ is tried first, the assumption $x\geq a(\vartheta)$ of \cref{l:approxcount} is initially satisfied. Note that \texttt{ApproxCount} cannot return $n'=\texttt{"MissingNodes"}$ or $n'=\texttt{"StrandTooShort"}$, or else $\mathcal L_j$ would not yield any approximation of $n$. Moreover, by \cref{l:approxcount}~(ii) and by the termination condition (Line~19), if $n'$ is not an error message while $x\geq a(\vartheta)$, then $n^\ast=n'=n$. On the other hand, due to \cref{l:approxcount}~(i), by the time $x=a(\vartheta)$ we necessarily have $n'\neq \texttt{"WrongGuess"}$ and therefore $n'$ is not an error message.

It remains to prove that \cref{l:algorithm2} actually gives an output other than \texttt{"Unknown"} within the claimed number of rounds; it suffices to show that it does so if it is executed at round~$t'=(\ell^2+\ell+1)(n-1)+1$. By \cref{xl:propa}, all nodes in the first $t'-n+1=\ell(\ell+1)(n-1)+1$ levels of the history tree are contained in the vista $\mathcal V$ at round $t'$. It is straightforward to prove by induction that the assumption $t'\geq s+(\ell+2)(n-1)$ of \cref{l:approxcount}~(iii) holds every time \texttt{ApproxCount} is invoked. Indeed, as long as this condition is satisfied, \cref{l:approxcount}~(iii) implies that $n'\neq \texttt{"MissingNodes"}$, and so \texttt{"Unknown"} is not returned at Line~11. Also, reasoning as in the previous paragraph, we infer that $n'\neq \texttt{"WrongGuess"}$ by the time $x=a(\vartheta)$. Thus, within $\Delta=(\ell+1)(n-1)$ levels, either a branching leader node is found (hence $n'=\texttt{"StrandTooShort"}$) or a new approximation of $n$ is computed (hence $n'$ is not an error message). Every time either event occurs, $s$ is increased by at most $\Delta$ at Line~17. Thus, after $\ell-1$ (or fewer) updates of $s$, we have $s\leq (\ell-1)\Delta= t' - (\ell+2)(n-1)-1$. Hence the condition of \cref{l:approxcount}~(iii) holds again, and the induction goes through for at least $\ell$ steps.

Observe that, since there can be at most $\ell-1$ branching leader nodes in $\mathcal V$, at least one approximation $n'>0$ of $n$ is computed within the $\ell$th iteration of the loop at Lines~6--17. This occurs within $t^\ast \leq \ell\Delta = t'-n$ levels. Because every level through $L_{t^\ast}$ is complete in $\mathcal V$, the condition $a(\vartheta)=a(\vartheta')$ holds in every interval $\mathcal L_1$, $\mathcal L_2$, \dots, $\mathcal L_{\ell-b+1}$. Hence, by \cref{l:approxcount}~(ii), every non-error approximation yielded by these intervals is equal to $n$. Thus, every time Line~15 is executed, we have $n^\ast=n'$, and the algorithm does not return \texttt{"Unknown"}. Finally, when Line~19 is reached, we have $t^\ast\leq t'-n=t'-n^\ast$, and therefore \texttt{"Unknown"} is not returned. Thus, an output other than \texttt{"Unknown"} is returned within the desired number of rounds.
\end{proof}

\subsection{Subroutine \texttt{ApproxCount}}\label{s:approxcount}
We will now define the subroutine $\texttt{ApproxCount}(\mathcal V, s, x, \ell)$ introduced in \cref{s:main} and invoked in \cref{l:algorithm2}. We will also give a proof to the technical \cref{l:approxcount}.

\smallskip
\mypar{Overview.} In $\texttt{ApproxCount}$, we are given a vista $\mathcal V$ with a strand of leader nodes hanging from the first leader node $\vartheta$ in level~$L_s$, where the anonymity $a(\vartheta)$ is an unknown number not greater than $\ell$. The algorithm begins by assuming that $a(\vartheta)$ is the given parameter $x$, and then it makes deductions on the anonymities of other nodes until it is able to either make an estimate $n'>0$ on the total number of agents or report failure in the form of an error message $n'\in\{\texttt{"MissingNodes"},\texttt{"StrandTooShort"},\texttt{"WrongGuess"}\}$. In particular, since the algorithm requires the existence of a long-enough strand hanging from $\vartheta$, it reports failure if some descendants of $\vartheta$ (in the relevant levels of $\mathcal V$) have more than one child.

An important difficulty that is unique to the multi-leader case is that, even if the vista $\mathcal V$ contains a long-enough strand of leader nodes, some nodes in the strand may still be branching in the history tree, and all branches except one may be missing from $\mathcal V$.

We remark that \texttt{ApproxCount} assumes that the network is 1-union-connected, as this is sufficient for the main result of \cref{s:main} to hold for any $\tau$-union-connected network (see the proof of \cref{t:multileadterm}).

\smallskip
\mypar{Discrepancy $\delta$.} \texttt{ApproxCount} is invoked with arguments $\mathcal V$, $s$, $x$, $\ell$, where $1\leq x\leq \ell$. If there are no nodes representing leaders in level $L_s$ of $\mathcal V$, the procedure immediately returns the error message $n'=\texttt{"MissingNodes"}$. Otherwise, we denote by $\vartheta$ the first such node. We define the \emph{discrepancy} $\delta$ as the ratio $x/a(\vartheta)$. Since $a(\vartheta)$ is not known in advance by the agent executing \texttt{ApproxCount}, then neither is $\delta$. Nonetheless, since both $x$ and $a(\vartheta)$ range between $1$ and $\ell$, it follows that $1/\ell\leq \delta\leq \ell$.

\smallskip
\mypar{Conditional anonymity.} The procedure \texttt{ApproxCount} begins by assuming that the anonymity of $\vartheta$ is $x$, and then makes deductions on other anonymities based on this assumption. Thus, for every node $v$ of $\mathcal V$, we distinguish between its actual anonymity $a(v)$ and its \emph{conditional anonymity}, defined as $a'(v)=\delta a(v)$. Essentially, conditional anonymities are values that \texttt{ApproxCount} computes under the initial assumption that $a'(\vartheta)=x=\delta a(\vartheta)$. Clearly, these values are correct if and only if $\delta=1$.

\smallskip
\mypar{Guessing conditional anonymities.} Let $u$ be a node of a history tree, and assume that the conditional anonymities of all its children $u_1$, $u_2$, \dots, $u_k$ have been computed: such a node $u$ is called a \emph{guesser}. If $v$ is not among the children of $u$ but is at their same level, and the red edge $\{v, u\}$ is present with multiplicity $m\geq 1$, we say that $v$ is \emph{guessable} by $u$ (see \cref{fig:guess}). In this case, we can make a \emph{guess} $g(v)$ on the conditional anonymity $a'(v)$:
\begin{equation}\label{xe:guess}
g(v)=\frac{a'(u_1)\cdot m_1+a'(u_2)\cdot m_2+\dots+a'(u_k)\cdot m_k}m,
\end{equation}
where $m_i$ is the multiplicity of the red edge $\{u_i, v'\}$ for all $1\leq i\leq k$, and $v'$ is the parent of $v$ (possibly, $m_i=0$). Note that $g(v)$ may not be an integer. Although a guess may be inaccurate, it never underestimates the conditional anonymity:

\begin{figure}
\centering
\includegraphics[scale=0.65]{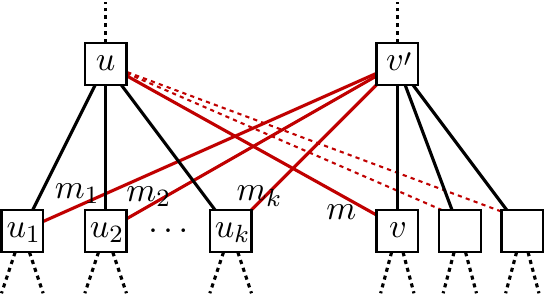}
\caption{If the anonymities of $u$, $u_1$, $u_2$, \dots, $u_k$ are known, then $v$ is \emph{guessable} by $u$.}
\label{fig:guess}
\end{figure}

\begin{lemma}\label{xl:guess2}
If $v$ is guessable, then $g(v)\geq a'(v)$. Moreover, if $v$ has no siblings, $g(v)=a'(v)$.
\end{lemma}
\begin{proof}
Let $u,v'\in L_t$, and let $P_1$ and $P_2$ be the sets of agents represented by $u$ and $v'$, respectively. By counting the links between $P_1$ and $P_2$ in $G_{t+1}$ in two ways, we have
$$\sum_i a(u_i)\,m_i = \sum_i a(v_i)\,m'_i\geq a(v)\,m,$$
where the two sums range over all children of $u$ and $v'$, respectively (note that $v=v_j$ for some $j$), and $m'_i$ is the multiplicity of the red edge $\{v_i,u\}$ (so, $m=m'_j$). Therefore, we have the inequality
$$a(v)\leq \frac{\sum_i a(u_i)\,m_i}{m}$$
which becomes an equality if $v$ has no siblings. Thus,
$$a'(v)=\delta a(v)\leq \frac{\sum_i \delta a(u_i)\,m_i}{m}= \frac{\sum_i a'(u_i)\,m_i}{m}=g(v)$$
and so $g(v)\geq a'(v)$, with equality if $v$ has no siblings.
\end{proof}

\mypar{Heavy nodes.} The subroutine \texttt{ApproxCount} assigns guesses in a \emph{well-spread} fashion, that is, in such a way that no two sibling nodes are assigned a guess. In other words, at most one of the children of each node is assigned a guess.

Suppose now that a node $v$ has been assigned a guess. We define its \emph{weight} $w(v)$ as the number of nodes in the subtree hanging from $v$ that have been assigned a guess (this includes $v$ itself). Recall that subtrees are determined by black edges only. We say that $v$ is \emph{heavy} if $w(v)\geq \lfloor g(v)\rfloor$.

\begin{lemma}\label{xl:limit}
Assume that $\delta\geq 1$. In a well-spread assignment of guesses, if $w(v)>a'(v)$, then some descendants of $v$ are heavy (the \emph{descendants} of $v$ are the nodes in the subtree hanging from $v$ other than $v$ itself).
\end{lemma}
\begin{proof}
Our proof is by well-founded induction on $w(v)$. Assume for a contradiction that no descendants of $v$ are heavy. Observe that some descendants of $v$ have been assigned guesses, because $w(v)\geq 2$. Indeed,
$$w(v)>a'(v)=\delta a(v) \geq a(v)\geq 1.$$
Thus, $v$ has $k\geq 1$ ``immediate'' descendants $v_1$, $v_2$, \dots, $v_k$ that have been assigned guesses. That is, for all $1\leq i\leq k$, the node $v_i$ has been assigned a guess, and no internal nodes of the unique black path with endpoints $v$ and $v_i$ have been assigned guesses.

By the basic properties of history trees, $a(v)\geq \sum_i a(v_i)$, and therefore $a'(v)\geq \sum_i a'(v_i)$. Also, the induction hypothesis implies that $w(v_i) \leq a'(v_i)$ for all $1\leq i\leq k$, or else one of the $v_i$'s would have a heavy descendant. Therefore,
$$w(v)-1 = \sum_i w(v_i) \leq \sum_i a'(v_i)\leq a'(v) < w(v).$$
Observe that all the terms in this chain of inequalities are between the two consecutive integers $w(v)-1$ and $w(v)$. It follows that
$$w(v_i)\leq a'(v_i)< w(v_i)+1$$
for all $1\leq i\leq k$ (recall that $a'(v_i)$ may not be an integer). Also,
$$a'(v)-1< \sum_i a'(v_i)\leq a'(v).$$
However, since every conditional anonymity is an integer multiple of the discrepancy $\delta\geq 1$, we conclude that $a'(v)=\sum_i a'(v_i)$. Hence, $a(v)=\sum_i a(v_i)$.

Let $1\leq d\leq k$ be such that $v_d$ has maximum depth. Since the assignment of guesses is well spread, no sibling of $v_d$ has been assigned a guess. However, since $a(v)=\sum_i a(v_i)$, it follows that $v_d$ has no siblings at all, for otherwise $a(v)>\sum_i a(v_i)$. Due to \cref{xl:guess2}, we have $g(v_d)=a'(v_d)$. Thus,
$$w(v_d)\leq a'(v_d) = g(v_d) < w(v_d)+1,$$
which implies that $\lfloor g(v_d)\rfloor = w(v_d)$, and so $v_d$ is heavy.
\end{proof}

\mypar{Correct guesses.} We say that a node $v$ has a \emph{correct} guess if $v$ has been assigned a guess and $g(v)=a'(v)$. The next lemma gives a criterion to determine if a guess is correct.

\begin{lemma}\label{xl:crit}
Assume that $\delta\geq 1$. In a well-spread assignment of guesses, if a node $v$ is heavy and no descendant of $v$ is heavy, then $v$ has a correct guess or the guess on $v$ is not an integer.
\end{lemma}
\begin{proof}
If $g(v)$ is not an integer, there is nothing to prove. Otherwise, because $v$ is heavy, $g(v)=\lfloor g(v)\rfloor\leq w(v)$. Since $v$ has no heavy descendants, \cref{xl:limit} implies $w(v)\leq a'(v)$. Also, by \cref{xl:guess2}, $a'(v)\leq g(v)$. We conclude that
$$g(v)\leq w(v)\leq a'(v)\leq g(v).$$
Therefore $g(v)=a'(v)$, and $v$ has a correct guess.
\end{proof}

When the criterion in \cref{xl:crit} applies to a node $v$, we say that $v$ has been \emph{counted}. So, counted nodes are nodes that have been assigned a guess, which was then confirmed to be the correct conditional anonymity.

\smallskip
\mypar{Cuts and isles.} Fix a vista $\mathcal{V}$ of a history tree $\mathcal H$. A set of nodes $C$ in $\mathcal{V}$ is said to be a \emph{cut} for a node $v\notin C$ of $\mathcal{V}$ if two conditions hold: (i)~for every leaf $v'$ of $\mathcal{V}$ that lies in the subtree hanging from $v$, the black path from $v$ to $v'$ contains a node of $C$, and (ii)~no proper subset of $C$ satisfies condition~(i). A cut for the root $r$ whose nodes are all counted is said to be a \emph{counting cut} (see \cref{fig:cuts}).

Let $s$ be a counted node in $\mathcal{V}$, and let $F$ be a cut for $s$ whose nodes are all counted. Then, the set of nodes spanned by the black paths from $s$ to the nodes of $F$ is called \emph{isle}; $s$ is the \emph{root} of the isle, while each node in $F$ is a \emph{leaf} of the isle. The nodes in an isle other than the root and the leaves are called \emph{internal}. An isle is said to be \emph{trivial} if it has no internal nodes.

If $s$ is an isle's root and $F$ is its set of leaves, we have $a(s)\geq \sum_{v\in F} a(v)$, which is equivalent to $a'(s)\geq \sum_{v\in F} a'(v)$. Note that equality does not necessarily hold, because $s$ may have some descendants in the history tree $\mathcal H$ that do not appear in the vista $\mathcal{V}$. If equality holds, then the isle is said to be \emph{complete}. In this case, given the conditional anonymities of $s$ and of all nodes in $F$, we can easily compute the conditional anonymities of all the internal nodes by adding them up starting from the nodes in $F$ and working our way up to $s$.

\begin{figure}
\centering
\includegraphics[scale=0.75]{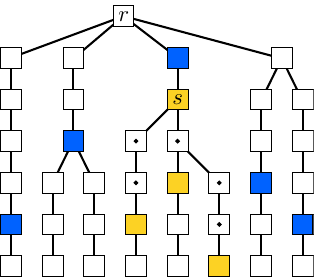}
\caption{The first levels of (a vista of) a history tree, where the colored nodes are counted. The blue nodes form a \emph{counting cut}, and the orange ones define a non-trivial \emph{isle} with root $s$, where the nodes with a dot are internal.}
\label{fig:cuts}
\end{figure}

\lstset{style=mystyle}
\begin{lstlisting}[caption={The subroutine \texttt{ApproxCount} invoked in \cref{l:algorithm2}.\label{l:algorithm3}},captionpos=t,float,abovecaptionskip=-\medskipamount,mathescape=true]
# Input: a vista $\mathcal V$ and three integers $s$, $x$, $\ell$
# Output: a pair $(n', t)$, where $n'$ is an integer or an error message, and $t$ is an integer

Let $L_{-1}$, $L_0$, $L_1$, $\dots$ be the levels of $\mathcal V$
Assign $t:=s$
If $L_s$ does not contain any leader nodes, return $(\texttt{"MissingNodes"},t)$
Let $\vartheta$ be the first leader node in $L_s$
Mark all nodes in $\mathcal V$ as not guessed, not counted, and not locked
Assign $u:=\vartheta$; assign $a'(u) := x$; mark $u$ as counted
While $u$ has a unique child $u'$ in $\mathcal V$
   Assign $u:=u'$; assign $a'(u) := x$; mark $u$ as counted
While there are guessable, non-counted, non-locked nodes in $\mathcal V$ and a counting cut has not been found
   Let $v$ be a guessable, non-counted, non-locked node of smallest depth in $\mathcal V$
   Let $L_{t'}$ be the level of $v$; assign $t:=\max\{t, t'\}$
   Assign a guess $g(v)$ to $v$ as in $\text{\cref{xe:guess}}$; mark $v$ as guessed; lock $v$ and all its siblings
   Let $P_v$ be the black path from $v$ to its ancestor in $L_s$
   If there is a heavy node in $P_v$
      Let $v'$ be the heavy node in $P_v$ of maximum depth
      If $g(v')$ is not an integer, return $(\texttt{"WrongGuess"},t)$
      Assign $a'(v') := g(v')$; mark $v'$ as counted and not guessed; unlock $v'$ and all its siblings
      If $v'$ is the root or a leaf of a non-trivial complete isle $I$
         For each internal node $w$ of $I$
            Assign $a'(w) := \sum_{w'\text{ leaf of }I\text{ and descendant of }w} a'(w')$
            Mark $w$ as counted, not guessed, and not locked
If no counting cut has been found, return $(\texttt{"StrandTooShort"},t)$
Else
   Let $C$ be a counting cut (between $L_s$ and $L_t$)
   Assign $n' := \sum_{v \in C}a'(v)$
   Assign $\ell' := \sum_{v \text{ leader node in } C}a'(v)$
   If $\ell'<\ell$, return $(\texttt{"MissingNodes"},t)$
   Else if $\ell'>\ell$, return $(\texttt{"WrongGuess"},t)$
   Else return $(n',t)$
\end{lstlisting}

\smallskip
\mypar{Overview of \texttt{ApproxCount}.} Our subroutine \texttt{ApproxCount} is found in \cref{l:algorithm3}. It repeatedly assigns guesses to nodes based on known conditional anonymities, starting from $\vartheta$ and its descendants. Eventually some nodes become heavy, and the criterion in \cref{xl:crit} causes the deepest of them to become counted. In turn, counted nodes eventually form isles; the internal nodes of complete isles are marked as counted, which gives rise to more guessers, and so on. In the end, if a counting cut is created, the algorithm checks whether the conditional anonymities of the leader nodes in the cut add up to $\ell$.

\smallskip
\mypar{Algorithmic details of \texttt{ApproxCount}.} The algorithm \texttt{ApproxCount} uses flags to mark nodes as ``guessed'', ``counted'', or ``locked''; initially, no node is marked. Thanks to these flags, we can check if a node $u\in \mathcal V$ is a guesser: let $u_1$, $u_2$, \dots, $u_k$ be the children of $u$ that are also in $\mathcal V$ (recall that a vista does not contain all nodes of a history tree); $u$ is a \emph{guesser} if and only if it is marked as counted, all the $u_i$'s are marked as counted, and $a'(u)=\sum_i a'(u_i)$ (which implies $a(u)=\sum_i a(u_i)$, and thus no children of $u$ are missing from $\mathcal V$).

\texttt{ApproxCount} will ensure that nodes marked as guessed are well-spread at all times: if a node of $\mathcal V$ is guessed, all its siblings (including itself) become \emph{locked}. While a node is locked, it cannot receive guesses. As defined earlier, a node $v$ in level $L_t$ of $\mathcal V$ is \emph{guessable} if there is a guesser $u$ in $L_{t-1}$ and the red edge $\{v,u\}$ is present in $\mathcal V$ with positive multiplicity.

The algorithm starts by assigning a conditional anonymity $a'(\vartheta)=x$ to the first leader node $\vartheta\in L_s$. If no leader node exists in $L_s$, the error message \texttt{"MissingNodes"} is immediately returned (Line~6). The algorithm also finds the longest strand $P_\vartheta$ hanging from $\vartheta$, assigns the same conditional anonymity $x$ to all of its nodes (including the unique child of the last node of $P_\vartheta$) and marks them as counted (Lines~7--11). Then, as long as there are nodes that can receive a guess and no counting cut has been found, the algorithm keeps assigning guesses to nodes. A node can receive a guess if it is guessable, not counted, and not locked (Line~12).

When a guess is made on a node $v$, the node itself and all of its siblings become locked (Line~15). This is to ensure that guessed nodes will always be well spread. Moreover, as a result of $v$ becoming guessed, some nodes in the path from $v$ to its ancestor in $L_s$ may become heavy; if this happens, let $v'$ be the deepest heavy node (Line~18). If $g(v')$ is not an integer, the algorithm returns the error message \texttt{"WrongGuess"} (Line~19). (As we will prove later, this can only happen if $\delta\neq 1$ or some nodes in the strand $P_\vartheta$ have children that are not in the vista $\mathcal V$.) Otherwise, if $g(v')$ is an integer, the algorithm marks $v'$ as counted and not guessed, in accordance with \cref{xl:crit}. Also, since $v'$ is no longer guessed, its siblings become unlocked and are again eligible to receive guesses (Line~20). Furthermore, if the newly counted node $v'$ is either the root or a leaf of a complete isle $I$, then the conditional anonymities of all the internal nodes of $I$ are determined, and such nodes are marked as counted (Lines~21--24).

In the end, the algorithm performs a ``reality check'' and possibly returns an estimate $n'$ of $n$ as follows. If no counting cut was found, the algorithm returns the error message \texttt{"StrandTooShort"} (Line~25). Otherwise, a counting cut $C$ has been found. The algorithm computes $n'$ (respectively, $\ell'$) as the sum of the conditional anonymities of all nodes (respectively, all leader nodes) in $C$. If $\ell'= \ell$, then the algorithm returns $n'$ (Line~32). Otherwise, it returns the error message \texttt{"MissingNodes"} if $\ell'<\ell$ (Line~30) or the error message \texttt{"WrongGuess"} if $\ell'>\ell$ (Line~31). In all cases, the algorithm also returns the maximum depth $t$ of a guessed or counted node (excluding $\vartheta$ and its descendants), or $s$ if no such node exists.

\smallskip
\mypar{Consistency condition.} In order for our algorithm to work properly, a condition has to be satisfied whenever a new guess is made. Indeed, note that all of our previous lemmas on guesses rest on the assumption that the conditional anonymities of a guesser and all of its children are known. However, while the node $\vartheta$ has a known conditional anonymity (by definition, $a'(\vartheta)=x$), the same is not necessarily true of the descendants of $\vartheta$ and all other nodes that are eventually marked as counted by the algorithm. This justifies the following definition.

\begin{condition}\label{cond:1}
During the execution of \textup{$\texttt{ApproxCount}$}, if a guess is made on a node $v$ at level $L_{t'}$ of $\mathcal V$, then $\vartheta$ has a (unique) descendant $\vartheta'\in L_{t'}$ and $a(\vartheta)=a(\vartheta')$.
\end{condition}

As we will prove next, as long as \cref{cond:1} is satisfied during the execution of \texttt{ApproxCount}, all of the nodes between levels $L_s$ and $L_t$ that are marked as counted do have correct guesses (i.e., their guesses coincide with their conditional anonymities). Note that in general there is no guarantee that \cref{cond:1} will be satisfied at any point; it is the job of our main Counting algorithm in \cref{s:main} to ensure that the condition is satisfied often enough for our computations to be successful.

\smallskip
\mypar{Correctness.} In order to prove the correctness of \texttt{ApproxCount}, it is convenient to show that it also maintains some \emph{invariants}, i.e., properties that are always satisfied as long as some conditions are met.

\begin{lemma}\label{xl:invariants}
Assume that $\delta\geq 1$. Then, as long as \cref{cond:1} is satisfied, the following statements hold.
\begin{itemize}
\item[(i)] The nodes marked as guessed are always well spread.
\item[(ii)] Whenever Line~13 is reached, there are no heavy nodes.
\item[(iii)] Whenever Line~13 is reached, all complete isles are trivial.
\item[(iv)] The conditional anonymity of any node between $L_s$ and $L_t$ that is marked as counted has been correctly computed.
\end{itemize}
\end{lemma}
\begin{proof}
Statement~(i) is true by design with no additional assumptions, because the algorithm only makes new guesses on unlocked nodes. In turn, a node is locked if and only if it is marked as guessed or has a sibling that is marked as guessed. Thus, no two nodes marked as guessed can be siblings, e.g., guesses are well spread.

All other statements can be proved collectively by induction. They certainly hold the first time Line~13 is ever reached. Indeed, the only nodes marked as counted up to this point are $\vartheta$ and some of its descendants, which are assigned the conditional anonymity $x$. Since $s=t$ and $\vartheta$ has conditional anonymity $x$ by definition, statement~(iv) is satisfied. Note that some descendants of $\vartheta$ that are marked as counted may not have been assigned their correct conditional anonymities, because some branches of the history tree may not appear in $\mathcal V$. However, no guesses have been made yet, and therefore no nodes are heavy; thus, statement~(ii) is satisfied. Moreover, the only isles are formed by $\vartheta$ and its descendants, and are obviously all trivial; so, statement~(iii) is satisfied.

Now assume that statements~(ii), (iii), and~(iv) are all satisfied up to some point in the execution of the algorithm. In particular, due to statement~(iv), all nodes that have been identified as guessers by the algorithm up to this point were in fact guessers according to our definitions. For this reason, all guesses have been computed as expected, and all of our lemmas on guesses apply (because $\delta\geq 1$).

The next guess on a new node $v$ is performed properly, as well. Indeed, \cref{cond:1} states that $\vartheta$ has a descendant $\vartheta'$ at the same level as $v$ such that $a(\vartheta')=a(\vartheta)$, and therefore $a'(\vartheta')=a'(\vartheta)=x$; so, $\vartheta'$ has the correct conditional anonymity. Thus, regardless of what the guesser of $v$ is (either the parent of $\vartheta'$ or some other counted node), the guess at Line~15 is computed properly.

Hence, if a node is identified as heavy at Lines~17--18, it is indeed heavy according to our definitions. Because statement~(ii) held before making the guess on $v$, it follows that any heavy node must have been created after the guess, and therefore should be on the path $P_v$, defined as in Line~16. If no heavy nodes are found on the path, then nothing is done and statements~(ii), (iii), and~(iv) keep being true.

Otherwise, by \cref{xl:crit}, the deepest heavy node $v'$ on $P_v$ has a correct guess and can be marked as counted, provided that the guess is an integer. Thus, statement~(iv) is still true after Line~20. At this point, there are no heavy nodes left, because $v'$ is no longer guessed and all of its ancestors along $P_v$ end up having the same weight they had before the guess on $v$ was made.

Now, because statement~(iii) held before marking $v'$ as counted, there can be at most one non-trivial complete isle, and $v'$ must be its root or one of its leaves. Note that, due to statement~(iv), any isle $I$ identified as complete at Line~21 is indeed complete according to our definitions. Since $I$ is complete, computing the conditional anonymities of its internal nodes as in Line~23 is correct, and therefore statement~(iv) is still true after Line~24. Also, the unique non-trivial isle $I$ is reduced to trivial isles, and statement~(iii) holds again. Finally, since Lines~21--24 may only cause weights to decrease, statement~(ii) keeps being true.
\end{proof}

\mypar{Running time.} We will now study the running time of \texttt{ApproxCount}. We will prove two lemmas that allow us to give an upper bound on the number of rounds it takes for the algorithm to return an output.

Recall that a node $v$ of the history tree $\mathcal H$ is said to be missing from level $L_i$ of the vista $\mathcal V$ if $v$ is at the level of $\mathcal H$ corresponding to $L_i$ but does not appear in $\mathcal V$. Clearly, if a level of $\mathcal V$ has no missing nodes, all previous levels have no missing nodes, either.

\begin{lemma}\label{xl:bound1}
Assume that $\delta\geq 1$. Then, as long as \cref{cond:1} holds, whenever Line~13 is reached, at most $\delta (n-1)$ levels contain locked nodes.
\end{lemma}
\begin{proof}
Note that the assumptions of \cref{xl:invariants} are satisfied, and therefore all the conditional anonymities and weights assigned to nodes up to this point are correct according to our definitions.

We will begin by proving that, if the subtree hanging from a node $v$ of $\mathcal V$ contains more than $a'(v)$ nodes marked as guessed, then it contains a node $v'$ marked as guessed such that $w(v')>a'(v')$. The proof is by well-founded induction based on the subtree relation in $\mathcal V$. If $v$ is guessed, then we can take $v'=v$, for in this case $w(v)>a'(v)$. Otherwise, by the pigeonhole principle, $v$ has at least one child $u$ whose hanging subtree contains more than $a'(u)$ guessed nodes. Thus, $v'$ is found in this subtree by the induction hypothesis.

Now, assume for a contradiction that more than $\delta (n-1)$ levels of $\mathcal V$ contain locked nodes; in particular, $\mathcal V$ contains more than $\delta (n-1)$ nodes marked as guessed. Consider the nodes in level $L_s$ other than $\vartheta$; the sum of their anonymities is at most $n-a(\vartheta)$ (note that some nodes may be missing from $L_s$), and so the sum of their conditional anonymities is at most $\delta (n - a(\vartheta))\leq \delta (n-1)$. Thus, by the pigeonhole principle, there is a a node $v\neq \vartheta$ in $L_s$ whose hanging subtree contains more than $a'(v)$ nodes marked as guessed.

Therefore, as proved above, the subtree hanging from $v$ contains a guessed node $v'$ such that $w(v')>a'(v')$. Since $\delta\geq 1$ and \cref{xl:invariants}~(i) holds, we can apply \cref{xl:limit} to $v'$, which implies that there exist heavy nodes. In turn, this contradicts \cref{xl:invariants}~(ii). We conclude that at most $\delta (n-1)$ levels contain locked nodes.
\end{proof}

\begin{lemma}\label{xl:bound2}
Assume that $\delta\geq 1$. Then, as long as level $L_{t}$ of $\mathcal V$ is not missing any nodes (where $t$ is defined and updated as in \textup{$\texttt{ApproxCount}$}), whenever Line~13 is reached, there are at most $n-2$ levels in the range from $L_{s+1}$ to $L_{t}$ where all guessable nodes are already counted.
\end{lemma}
\begin{proof}
By definition of $t$, either $t=s$ or the algorithm has performed at least one guess on a node at level $L_t$ with a guesser at level $L_{t-1}$. It is easy to prove by induction that the first guesser to perform a guess on this level must be the unique descendant $\vartheta'\in L_{t-1}$ of the selected leader node $\vartheta\in L_s$. Moreover, both $\vartheta'$ and its unique child in $\mathcal V$ have been assigned conditional anonymity $x$ at Lines~9--11, and the same is true of all nodes in the black path $P_\vartheta$ from $\vartheta$ to $\vartheta'$, which is a strand in $\mathcal V$. Since level $L_t$ is not missing any nodes, then each of the nodes in $P_\vartheta$ has a unique child in the history tree, as well. It follows that all descendants of $\vartheta$ up to level $L_t$ have the same anonymity as $\vartheta$. Also, by definition of $t$ and the way it is updated (Line~14), no guesses have been made on nodes at levels deeper than $L_t$, and hence \cref{cond:1} is satisfied up to this point. Thus, \cref{xl:invariants} applies.

Observe that there are no counting cuts, or Line~13 would not be reachable. Due to Lines~9--11, all of the nodes in $P_\vartheta$ initially become guessers. Hence, all levels between $L_s$ and $L_{t-1}$ must have a non-empty set of guessers at all times. Consider any level $L_i$ with $s< i\leq t$ such that all the guessable nodes in $L_i$ are already counted. Let $S$ be the set of guessers in $L_{i-1}$; note that not all nodes in $L_{i-1}$ are guessers, or else they would give rise to a counting cut. Since the network is 1-union-connected, there is a red edge $\{u,v\}$ (with positive multiplicity) such that $u\in S$ and the parent of $v$ is not in $S$. By definition, the node $v$ is guessable; therefore, it is counted. Also, since the parent of $v$ is not a guesser, $v$ must have a non-counted parent or a non-counted sibling; note that such a non-counted node is in $\mathcal V$, because the levels up to $L_t$ are not missing any nodes.

We have proved that every level between $L_{s+1}$ and $L_{t}$ where all guessable nodes are counted contains a counted node $v$ having a parent or a sibling that is not counted: we call such a node $v$ a \emph{bad} node. To conclude the proof, it suffices to show that there are at most $n-2$ bad nodes between $L_{s+1}$ and $L_{t}$. Observe that no nodes in $P_\vartheta$ can be bad.

We will prove that, if a subtree $\mathcal W$ of $\mathcal V$ contains the root $r$, the leader node $\vartheta$, no counting cuts, and no non-trivial isles, then $\mathcal W$ contains at most $f-1$ bad nodes, where $f$ is the number of leaves of $\mathcal W$ not in the subtree hanging from $\vartheta$. We stress that, in the context of $\mathcal W$, a bad node is defined as a counted node in $\mathcal W$ (other than $\vartheta$ and its descendants) that has a non-counted parent or a non-counted sibling in $\mathcal W$.

The proof is by induction on $f$. The case $f=0$ is impossible, because the single node $\vartheta$ yields a counting cut. Thus, the base case is $f=1$, which holds because any bad node $v$ in $\mathcal W$ and not in $P_\vartheta$ gives rise to the counting cut $\{\vartheta, v\}$ (recall that a bad node is counted by definition).

For the induction step, let $v$ be a bad node of maximum depth in $\mathcal W$. Let $(v_1, v_2,\dots, v_k)$ be the black path from $v_1=v$ to the root $v_k=r$, and let $1<i\leq k$ be the smallest index such that $v_i$ has more than one child in $\mathcal W$ ($i$ must exist, because this path eventually joins the black path from $\vartheta$ to $r$). Let $\mathcal W'$ be the tree obtained by deleting the black edge $\{v_{i-1},v_i\}$ from $\mathcal W$, as well as the subtree hanging from it. 

Note that the induction hypothesis applies to $\mathcal W'$: since $v_1$ is counted, and each of the nodes $v_2$, \dots, $v_{i-1}$ has a unique child in $\mathcal W$, the removal of $\{v_{i-1},v_i\}$ does not create counting cuts or non-trivial isles. Also, $v_2$ is not counted (unless $v_2=v_i$), because $v_1$ is bad. Furthermore, none of the nodes $v_3$, \dots, $v_{i-1}$ is counted, or else $v_2$ would be an internal node of a (non-trivial) isle in $\mathcal W$. Therefore, none of the nodes $v_2$, \dots, $v_{i-1}$ is counted. In particular, none of these nodes is bad in $\mathcal W$.

Moreover, a node of $\mathcal W'$ is bad if and only if it is bad in $\mathcal W$. This is trivial for all nodes, except for the siblings of $v_{i-1}$, which require a careful proof. Let $u\neq v_{i-1}$ be a sibling of $v_{i-1}$ in $\mathcal W$. If $u$ is not counted, then it is not a bad node in $\mathcal W$ nor in $\mathcal W'$. Hence, let us assume that $u$ is counted. If $v_i$ is not counted, then $u$ is a bad node in $\mathcal W$ and in $\mathcal W'$. Thus, let us assume that $v_i$ is counted.

Assume for a contradiction that all children of $v_i$ other than $v_{i-1}$ are counted. Then, if $i>2$, the nodes $v$, $v_i$, and the children of $v_i$ (other than $v_{i-1}$) are all counted and form a non-trivial isle in $\mathcal W$, which is impossible. On the other hand, if $i=2$, then $v_i$ and all its children (including $v=v_1$) are counted, which contradicts the fact that $v$ is a bad node. Thus, we have proved that $v_i$ must have a non-counted child in $\mathcal W'$ (hence in $\mathcal W$), and therefore $u$ is a bad node in both $\mathcal W$ and $\mathcal W'$.

It follows that $\mathcal W'$ has exactly one less bad node than $\mathcal W$ and at most $f-1$ leaves (because there is at least one leaf in the subtree of $\mathcal W$ hanging from $v$). Thus, the induction hypothesis implies that $\mathcal W'$ contains at most $f-2$ bad nodes, and therefore $\mathcal W$ contains at most $f-1$ bad nodes.

Observe that the subtree $\mathcal V'$ of $\mathcal V$ formed by all levels up to $L_{t}$ satisfies all of the above conditions, as it contains $\vartheta\in L_s$, the root $r$, and has no counting cuts, because a counting cut for $\mathcal V'$ would be a counting cut for $\mathcal V$, as well (recall that $\mathcal V$ has no counting cuts). Also, \cref{xl:invariants}~(iii) ensures that $\mathcal V'$ contains no non-trivial complete isles. However, since no nodes are missing from the levels of $\mathcal V'$, all isles in $\mathcal V'$ are complete, and thus must be trivial. We conclude that, if $\mathcal V'$ has $f$ leaves not in the subtree hanging from $\vartheta$, it contains at most $f-1$ bad nodes. Since such leaves induce a partition of the at most $n-1$ agents not represented by $\vartheta$, we have $f\leq n-1$, implying that the number of bad nodes up to $L_{t}$ is at most $n-2$.
\end{proof}

Observe that the statement of \cref{xl:bound2} holds for $n=1$ as well, because in this case the single node $\vartheta$ constitutes a counting cut, and Line~13 is never reached.

\smallskip
\mypar{Main lemma.} We are now ready to prove the salient properties of \texttt{ApproxCount} as summarized in \cref{l:approxcount}, which we restate next.

\setcounter{theorem}{0}

\begin{lemma}
Let \textup{$\texttt{ApproxCount}(\mathcal V, s, x, \ell)$} return $(n',t)$. Assume that $\vartheta$ exists and $x\geq a(\vartheta)$. Let $\vartheta'$ be the (unique) descendant of $\vartheta$ in $\mathcal V$ at level $L_{t}$, and let $L_{t'}$ be the last level of $\mathcal V$. Then:
\begin{itemize}
\item[(i)] If $x=a(\vartheta)=a(\vartheta')$, then $n'\neq \texttt{"WrongGuess"}$.
\item[(ii)] If $n'$ is not an error message and $a(\vartheta)=a(\vartheta')$, and if either $t'\geq t+n'$ or level $L_t$ is complete, then $n'=n$.
\item[(iii)] If $t'\geq s+(\ell+2)(n-1)$, then $s\leq t\leq s+(\ell+1)(n-1)$ and $n'\neq \texttt{"MissingNodes"}$. Moreover, if $n'=\texttt{"StrandTooShort"}$, then $L_t$ contains a leader node with at least two children in $\mathcal V$.
\end{itemize}
\end{lemma}
\begin{proof}
Note that $\vartheta'$ is well defined, because the returned pair is $(n',t)$, which means that either $t=s$, and thus $\vartheta=\vartheta'$, or $t>s$, and hence some guesses have been made on nodes in level $L_t$, the first of which must have had the parent of $\vartheta'$ as the guesser.

Let us prove statement~(i). The assumption $x=a(\vartheta)$ implies $\delta=1$. Moreover, since $a(\vartheta)=a(\vartheta')$, \cref{cond:1} is satisfied whenever a guess is made (this is a straightforward induction). Therefore, by \cref{xl:invariants}~(iv), all nodes marked as counted up to $L_t$ indeed have the correct guesses. So, the conditional anonymity that is computed for any node is equal to its anonymity ($a'(v)=\delta a(v)=a(v)$), and hence is an integer. This implies that \texttt{ApproxCount} cannot return the error message \texttt{"WrongGuess"} at Line~19. Also, either $\ell'=\ell$ if all leader agents have been counted, or $\ell'<\ell$ if some leader nodes are missing from the vista. Either way, \texttt{ApproxCount} cannot return the error message \texttt{"WrongGuess"} at Line~31. We conclude that $n'\neq \texttt{"WrongGuess"}$.

Let us prove statement~(ii). Again, because $a(\vartheta)=a(\vartheta')$, \cref{cond:1} is satisfied, and all nodes marked as counted have correct guesses. Also, $x\geq a(\vartheta)$ is equivalent to $\delta\geq 1$. By assumption, \texttt{ApproxCount} returns $(n',t)$, where $n'$ is not an error message and either $t'\geq t+n'$ or level $L_t$ is complete. Since $n'$ is not an error message, a counting cut $C$ was found whose nodes are within levels up to $L_t$, and $n'$ is the sum of the conditional anonymities of all nodes in $C$. Let $S_C$ be the set of agents represented by the nodes of $C$; note that $n'\geq |S_C|$, because $\delta\geq 1$. We will prove that $S_C$ includes all agents in the system. If level $L_t$ is complete, this follows immediately from the fact that $C$ is a counting cut with no nodes after $L_t$. Otherwise, we have $t'\geq t+n'$. Assume the contrary; \cref{xl:propamain} implies that, since $t'\geq t+n'\geq t+|S_C|$, there is a node $z\in L_t$ representing some agent not in $S_C$. Thus, the black path from $z$ to the root $r$ does not contain any node of $C$, contradicting the fact that $C$ is a counting cut with no nodes after $L_t$. Therefore, $|S_C|=n$, i.e., the nodes in $C$ represent all agents in the system. Since \texttt{ApproxCount} does not return an error message, the ``reality check'' $\ell'=\ell$ succeeds (Lines~30--32). However, $\ell'$ is the sum of the conditional anonymities of all leader nodes in $C$, and hence $\ell'=\delta\ell$, implying that $\delta=1$. Thus, $n'=\delta n=n$, as claimed.

Let us prove statement~(iii). Once again, $x\geq a(\vartheta)$ is equivalent to $\delta\geq 1$. By \cref{xl:propa}, if $L_{t'}$ is the last level of $\mathcal V$, then no nodes are missing from level $L_{t'-n+1}$. In fact, since $t'-n+1\geq s+(\ell+1)(n-1)$, no nodes are missing from any level up to $L_{s+(\ell+1)(n-1)}$. Let $\vartheta''$ be the deepest descendant of $\vartheta$ that is marked as counted at Lines~9--11, and let $L_{p}$ be the level of $\vartheta''$. By construction, either all children of $\vartheta''$ are missing from $L_{p+1}$ or at least two children of $\vartheta''$ are in $L_{p+1}$. Also note that $\vartheta'$ must be an ancestor of $\vartheta''$, and so $t\leq p$.

Assume that $p<s+(\ell+1)(n-1)$. This implies that no nodes are missing from level $L_{p+1}$, and therefore $\vartheta''$ must have at least two children in $L_{p+1}$. Since $t\leq p$, we have $t< s+(\ell+1)(n-1)$, as desired. Now assume that $n'=\texttt{"StrandTooShort"}$, which implies that the algorithm was unable to find a counting cut. We claim that in this case $t=p$. So, assume for a contradiction that $t\leq p-1$. It follows that $\vartheta'\in L_t$ is a guesser. Recall that $L_t$ and $L_{t+1}$ are not missing any nodes, because $t\leq p$. Since the network is connected at round~$t+1$, there is at least one node $v\in L_{t+1}$ that is guessable by $\vartheta'$. Also, since no guess has ever been made in level $L_{t+1}$ (due to the way $t$ is updated in Line~14), it follows that $v$ is not counted and not locked. However, the algorithm cannot return $n'=\texttt{"StrandTooShort"}$ as long as there are nodes such as $v$ (Line~12). Thus, $t=p$, which means that $\vartheta'=\vartheta''$, and hence $\vartheta'$ has at least two children in $\mathcal V$, as desired.

Assume now that $p\geq s+(\ell+1)(n-1)$. If $\vartheta$ is the only node in $L_s$, it constitutes a counting cut. In this special case, Line~13 is never reached, $n'=\texttt{"StrandTooShort"}$ is not returned at Line~25, and $t=s$. Hence, we may assume that there are nodes in $L_s$ other than $\vartheta$ and Line~13 is reached (so, we have $n>1$). Recall that no nodes are missing from any level up to $L_{s+(\ell+1)(n-1)}$. In particular, no nodes are missing from the levels in the non-empty interval $\mathcal L$ consisting of the $(\ell+1)(n-1)$ levels from $L_{s+1}$ to $L_{s+(\ell+1)(n-1)}$. Thus, by definition of $p$, as long as no guesses are made outside of $\mathcal L$, \cref{cond:1} holds, and therefore \cref{xl:bound1,xl:bound2} apply. Hence, as long as no guesses are made outside of $\mathcal L$, there are at most $\delta (n-1)$ levels of $\mathcal L$ containing locked nodes (\cref{xl:bound1}) and there are at most $n-2$ levels of $\mathcal L$ where all guessable nodes are already counted (\cref{xl:bound2}). So, there are at most
$$\delta (n-1) + n-2 \leq \ell (n-1) + n-2 = (\ell+1)(n-1)-1$$
levels in $\mathcal L$ were the algorithm can make no new guesses. We conclude that $\mathcal L$ always contains at least one node where a guess can be made, and no guesses are ever made outside of $\mathcal L$ until either a counting cut is found or $n'=\texttt{"WrongGuess"}$ is returned at Line~19. In both cases, $n'=\texttt{"StrandTooShort"}$ is not returned, and moreover $t\leq s+(\ell+1)(n-1)$, as desired.

It remains to prove that $n'\neq \texttt{"MissingNodes"}$. Since $\vartheta$ exists by assumption, the error message \texttt{"MissingNodes"} cannot be returned at Line~6 and can only be returned at Line~30. In turn, this can only occur if a counting cut has been found and $\ell'<\ell$. However, we have already proved that no level up to $L_t$ is missing any nodes, which implies that the counting cut contains nodes representing all agents, and in particular $\ell'=\delta\ell$. Since $\delta\geq 1$, we have $\ell'\geq \ell$, and the condition at Line~30 is not satisfied.
\end{proof}

\mypar{Worst-case example.} For $\ell=\tau=1$, our Counting algorithm in \cref{t:multileadterm} yields a running time of $3n-2$ rounds. The example in \cref{xfig:counter2}, which can easily be generalized to networks of any size $n$, shows that our Counting algorithm can in fact terminate in $3n-3$ rounds. Indeed, the last node in a counting cut is at level $L_{2n-3}$, and then it takes an extra $n$ rounds for the termination condition $t\geq t^\ast +n^\ast$ of \cref{l:algorithm2} to be satisfied.

Note that removing all the double edges from this network and eliminating round~$1$ yields a simple dynamic network where our Counting algorithm terminates in $3n-4$ rounds and stabilization occurs in $2n-3$ rounds, providing a close-to-worst-case example for both \cref{t:multileadterm,t:multileadstab}.

\begin{figure}
\centering
\includegraphics[scale=0.5]{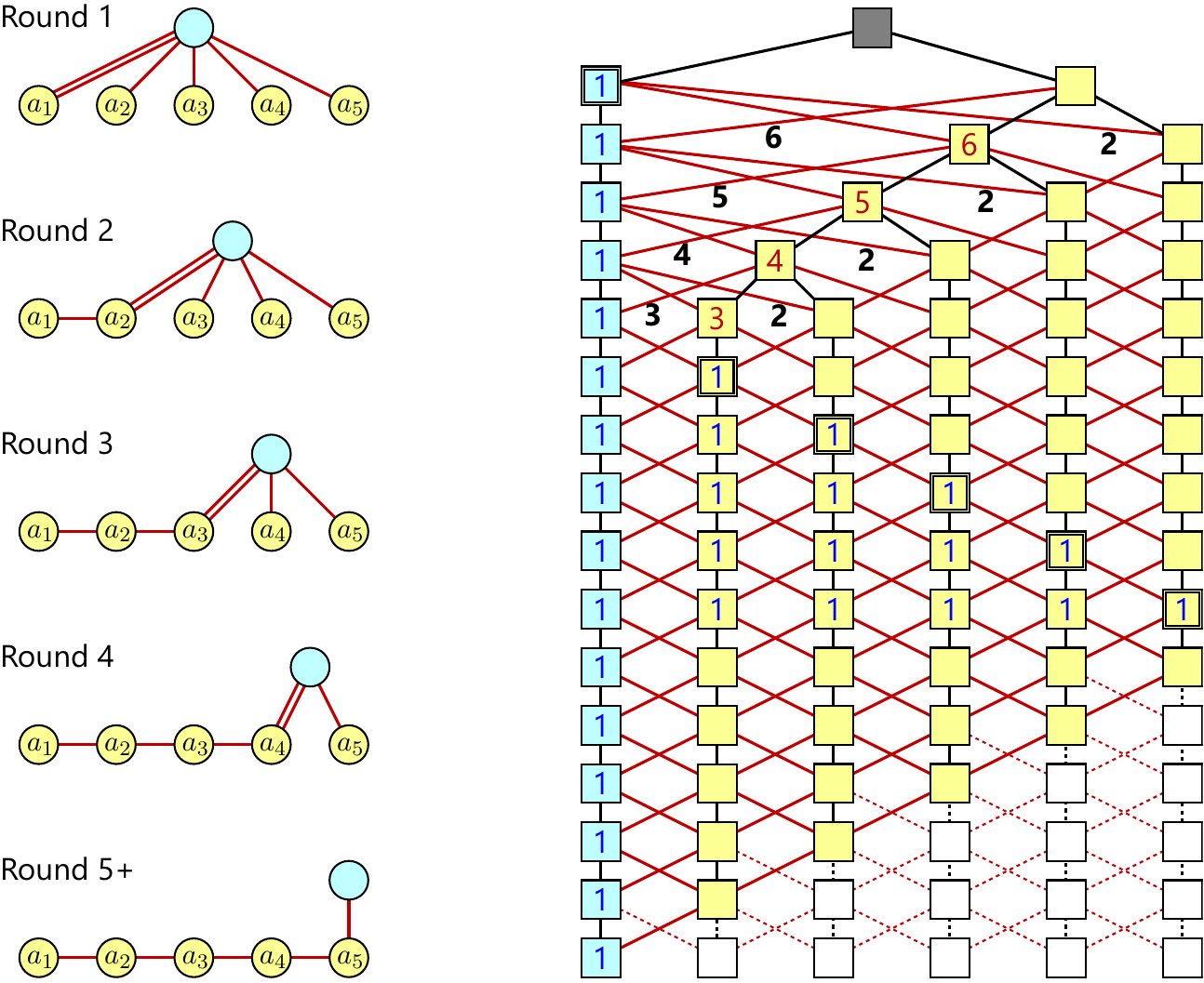}
\caption{An example of a dynamic network with $\ell=\tau=1$ and $n=6$ where the algorithm of \cref{t:multileadterm} terminates in $3n-3$ rounds, almost matching the upper bound of $3n-2$ rounds. The white nodes in the history tree are not in the vista of the leader at the last round. The numbers inside nodes represent guesses; the blue ones correspond to nodes that are marked as counted. The six highlighted nodes constitute a counting cut.}
\label{xfig:counter2}
\end{figure}

\section{Negative Results}\label{s:negative}
In this section we collect several negative results and counterexamples, which not only provide lower bounds asymptotically matching our algorithms' running times, but also justify all of the assumptions made in \cref{s:intermediate,s:termleader} about the a-priori knowledge of agents. Although some of these facts were previously known, here we offer simple and self-contained proofs based on history trees.

We point out that all the impossibility results concerning the assumptions of our algorithms and all the principal lower bounds in this section already hold for simple networks. Moreover, all the counterexamples in \cref{s:unsolvable} are static networks, implying that our impossibility results hold for dynamic, as well as static networks.

\subsection{Unsolvable Problems}\label{s:unsolvable}
Showing that certain functions are not computable in certain network models allows us to argue that our algorithms are \emph{universal} within the topology-independent function class considered in this paper, i.e., they can be applied to every solvable problem in that class.

\smallskip
\mypar{Networks with leaders.} We will prove that the multiset-based functions introduced in \cref{s:2} are the only functions that can be computed deterministically in anonymous networks with leaders, even when restricted to connected static simple networks. This result, together with \cref{xth:compl}, shows that the algorithms in \cref{s:stableader,s:termleader} can be generalized to all functions that can be computed in networks with leaders, i.e., the multiset-based functions.

\begin{proposition}\label{t:aggr}
For any $\ell\geq 1$ and $n\geq \ell$, no algorithm computes a non-multiset-based function (with or without termination) on all static simple networks in $\mathcal N^1_{n,\ell}$.
\end{proposition}
\begin{proof}
Let us consider the (static) network whose topology at round~$t$ is the complete graph $G_t=K_n$, i.e., each agent receives messages from all other agents at every round. We can prove by induction that all nodes of the history tree other than the root have exactly one child. This is because any two agents with the same input always receive equal multisets of messages, and are therefore always indistinguishable. Thus, the history tree is completely determined by the multiset $\mu_\lambda$ of all agents' inputs; moreover, an agent's vista at any given round only depends on the agent's own input and on $\mu_\lambda$ (the fact that $\ell$ is a constant does not affect this). By the fundamental theorem of history trees \cref{xth:view}, this is enough to conclude that if an agent's output stabilizes, that output must be a function of the agent's own input and of $\mu_\lambda$, which is the defining condition of a multiset-based function.
\end{proof}

The following result justifies the assumption made in \cref{s:stableader,s:termleader} that agents have a-priori knowledge of the number of leaders $\ell$ in the system. It states that no algorithm can compute the Counting function $F_C$ (which, as we recall, is a multiset-based function) without knowledge of $\ell$, even when restricted to connected static simple networks with a known and arbitrarily large ratio $n/\ell=k$.

\begin{proposition}\label{t:knowL}
For any integers $k\geq 1$ and $\ell\geq 1$ with $k\ell\geq 3$, no algorithm computes the Counting function $F_C$ (with or without termination) on all static simple networks in $\mathcal N^1_{k\ell,\ell}\cup \mathcal N^1_{k(\ell+1),\ell+1}$.
\end{proposition}
\begin{proof}
Let $\mathcal G$ (respectively, $\mathcal G'$) be the static network consisting of a cycle of $k\ell$ (respectively, $k(\ell+1)$) agents of which $\ell$ are leaders, such that the leaders are evenly spaced among the non-leaders. That is, each leader has a sequence of exactly $k-1$ consecutive non-leaders on each side. Assume that all agents in $\mathcal G$ and $\mathcal G'$ get the same input (apart from their leader flags). Then, at any round, all leaders in both networks have isomorphic vistas, and therefore always give equal outputs. However, in order to compute the Counting function, the leaders of $\mathcal G$ have to eventually output $k\ell$, while the leaders of $\mathcal G'$ have to eventually output $k(\ell+1)$. Hence, at most one of these networks can stabilize on the correct output.
\end{proof}

\mypar{Leaderless networks.} We will now prove that the frequency-based functions introduced in \cref{s:2} are the only  functions that can be computed deterministically in anonymous multigraph networks without leaders, even when restricted to connected static networks. This result, together with \cref{th:concentration}, shows that the algorithms in \cref{s:stableaderless,s:termleaderless} can be generalized to all (topology-independent) functions that can be computed in leaderless networks, i.e., the frequency-based functions.

Although the next theorem is implied by~\cite[Theorem~III.1]{HOT}, here we provide an alternative and simpler proof.

\begin{proposition}\label{t:scale}
No algorithm computes a non-frequency-based function (with or without termination) on all static multigraph networks in $\bigcup_{n\geq 1}\mathcal N^1_{n,0}$.
\end{proposition}
\begin{proof}
The complete-graph argument used in the proof of \cref{t:aggr}, which does not rely on the presence of leaders, shows that every function computable on all static multigraph networks must be multiset-based. Assume for a contradiction that a computable multiset-based function $F(p,\lambda)=\psi(\lambda(p),\mu_\lambda)$ is not frequency-based. Then there are an input multiset $\mu$, an integer $\alpha>1$, and an input $z$ occurring in $\mu$ such that $\psi(z,\mu)\neq\psi(z,\alpha\cdot\mu)$.

Let $n$ be the cardinality of $\mu$ (counting multiplicities), and consider a static cycle $C_n$ whose $n$ agents' inputs realize $\mu$. If $n=1$, interpret $C_n$ as two self-loops; if $n=2$, interpret it as two parallel links. Now consider the cycle $C_{\alpha n}$ obtained by repeating the same sequence of inputs $\alpha$ times. Every agent in $C_n$ and its corresponding agents in $C_{\alpha n}$ have isomorphic vistas at every round, as follows immediately by induction. Hence, due to the fundamental theorem of history trees \cref{xth:view}, corresponding agents with input $z$ must give the same output, contradicting $\psi(z,\mu)\neq\psi(z,\alpha\cdot\mu)$.
\end{proof}

The use of multigraphs for $n\leq 2$ is essential. For example, on connected static simple networks, the non-frequency-based function that outputs~$1$ if $n=1$ and~$0$ otherwise is computable in one round: an agent outputs~$1$ if and only if it receives no messages. A similar counterexample can be constructed for $n=2$, as well. However, the cycles used in the proof of \cref{t:scale} are simple whenever the smaller network has more than two agents. Hence, every function computable on all connected static simple leaderless networks must be frequency-based when restricted to instances with $n>2$ agents.

The following result justifies the need for some a-priori information in \cref{s:termleaderless}, such as an upper bound on $n$ or on the dynamic diameter $d$ of the network. It states that, without any information restricting the possible extent of the network, the leaderless Average Consensus problem (which, as we recall, is a frequency-based function) cannot be solved with explicit termination, even when restricted to connected static simple networks.

\begin{proposition}\label{t:knowN}
No algorithm solves the Average Consensus problem with termination on all static simple networks in $\bigcup_{n\geq 1}\mathcal N^1_{n,0}$.
\end{proposition}
\begin{proof}
Assume for a contradiction that there is such an algorithm $\mathcal A$. Let $\mathcal G$ be a static network consisting of three agents forming a cycle, and assign input~$0$ to all of them. If the agents execute $\mathcal A$, they eventually output the mean value~$0$ and terminate, say in $t$ rounds.

Now construct a static network $\mathcal G'$ consisting of a cycle of $2t+3$ agents $p_1$, $p_2$, \dots, $p_{2t+3}$; assign input~$1$ to $p_1$ and input~$0$ to all other agents. The agent $p_{t+2}$ has distance $t+1$ from $p_1$, and therefore, from round~$0$ to round~$t$, its vista is isomorphic to the vista of any agent in $\mathcal G$. Hence, if $p_{t+2}$ executes $\mathcal A$, it terminates in $t$ rounds with the incorrect output~$0$. Thus, $\mathcal A$ is incorrect.
\end{proof}

\subsection{Lower Bounds}\label{s:lower}
We will now give some lower bounds on the complexity of problems for anonymous dynamic networks. Since our algorithms have linear running times, our focus is on optimizing the multiplicative constants of the leading terms.

\smallskip
\mypar{Preliminary results.} We first prove some simple statements that will be used to derive lower bounds for stabilizing and terminating algorithms.

\begin{lemma}\label{l:lowerlemma}
Let $\mathcal G$ and $\mathcal G'$ be two networks on $n$ and $n'$ agents respectively, where $n\neq n'$. Assume that there is an agent $p$ in $\mathcal G$ and an agent $p'$ in $\mathcal G'$ such that $p$ and $p'$ have isomorphic vistas at round~$t$. Then,
\begin{itemize}
\item No algorithm computes the Counting function $F_C$ and stabilizes within $t$ rounds in both $\mathcal G$ and $\mathcal G'$.
\item No algorithm computes the Counting function $F_C$ in both $\mathcal G$ and $\mathcal G'$ and terminates within $t$ rounds in $\mathcal G$ (or $\mathcal G'$).
\end{itemize}
\end{lemma}
\begin{proof}
Since $p$ and $p'$ have isomorphic vistas at round $t$, they have isomorphic vistas at all rounds up to $t$. Thus, by \cref{xth:view}, if $p$ and $p'$ execute the same algorithm, they give equal outputs up to round $t$. Since the Counting function $F_C$ prescribes that $p$ must output $n$ and $p'$ must output $n'\neq n$, it is impossible for both agents to simultaneously give the correct output within $t$ rounds. In particular, no Counting algorithm can stabilize within $t$ rounds.

Moreover, if the execution of $p$ terminated within $t$ rounds, then the execution of $p'$ would terminate at the same round, as well (again, due to \cref{xth:view}). In that case, both agents would return the same output, which would be incorrect for at least one of them. In particular, no Counting algorithm can terminate within $t$ rounds in $\mathcal G$.
\end{proof}

\begin{corollary}\label{c:lowerlinear}
Let $\mathcal G$ and $\mathcal G'$ be networks on $m$ and $m+k$ agents respectively, with $k>0$. Assume that there are agents in $\mathcal G$ and in $\mathcal G'$ that have isomorphic vistas at round $am+b$, with $a\geq 0$. Then, there is no algorithm that computes the Counting function $F_C$ on both $\mathcal G$ and $\mathcal G'$ and stabilizes in less than $an-ak+b+1$ rounds in each of them, where $n$ is the size of the network (i.e., $n=m$ for $\mathcal G$ and $n=m+k$ for $\mathcal G'$). If termination is required, the bound improves to $an+b+1$ rounds.
\end{corollary}
\begin{proof}
Let $t=am+b$. Assume for a contradiction that a Counting algorithm stabilizes in both networks within $an-ak+b$ rounds. In $\mathcal G$, where $n=m$, this is $am-ak+b\leq t$ rounds; in $\mathcal G'$, where $n=m+k$, this is $a(m+k)-ak+b=t$ rounds. Hence the Counting algorithm stabilizes in both networks within $t$ rounds, contradicting \cref{l:lowerlemma}.

\cref{l:lowerlemma} also states that no Counting algorithm can terminate within $t$ rounds in $\mathcal G$. Hence, if termination is required, no algorithm succeeds in less than $t+1=am+b+1=an+b+1$ rounds (recall that $n=m$ in $\mathcal G$).
\end{proof}

\mypar{Unique-leader networks.} We will now prove a lower bound of roughly $2n$ rounds on the Counting problem for always connected networks with a unique leader.

We first introduce a family of 1-union-connected dynamic networks. For any $m\geq 3$, we consider the dynamic network $\mathcal G_m$ whose topology at round~$t$ is the graph $G^{(m)}_t$ defined on the system $\{p_1,p_2,\dots,p_m\}$ as follows. If $t\geq m-2$, then $G^{(m)}_t$ is the path graph $P_m$ spanning all agents $p_1$, $p_2$, \dots, $p_m$ in order. If $1\leq t\leq m-3$, then $G^{(m)}_t$ is $P_m$ with the addition of the single edge $\{p_{t+1}, p_m\}$. We assume $p_1$ to be the leader and all other agents to have the same input.

\begin{proposition}\label{xth:lower}
For any $m\geq 3$, no algorithm computes the Counting function $F_C$ on all simple networks in $\bigcup_{n=m}^{m+1}\mathcal N^1_{n,1}$ in less than $2n-6$ rounds. If termination is required, the bound improves to $2n-4$ rounds.
\end{proposition}
\begin{proof}
Let us consider the network $\mathcal G_m$ as defined above. It is straightforward to prove by induction that, at every round~$t\leq m-3$, the agent $p_{t+1}$ gets disambiguated, while all agents $p_{t+2}$, $p_{t+3}$, \dots, $p_m$ are still indistinguishable. So, the history tree of $\mathcal G_m$ has a very regular structure, which is illustrated in \cref{xfig:lower1,xfig:lower2}. By comparing the history trees of $\mathcal G_m$ and $\mathcal G_{m+1}$, we see that the leaders of the two systems have identical vistas up to round~$2m-5$. Our claim now follows immediately from \cref{c:lowerlinear}, with $k=1$, $a=2$, and $b=-5$
\end{proof}

A slightly better lower bound can be obtained if we allow self-loops in the network. Consider the static network $\mathcal P_n$ consisting of a path graph spanning all $n$ agents, where one endpoint of the path is the unique leader and the other endpoint has a self-loop. It is easy to see that the leaders in $\mathcal P_n$ and $\mathcal P_{n+1}$ have identical vistas up to round~$2n-2$. This implies a lower bound of $2n-3$ rounds for stabilization and $2n-1$ rounds for termination in static networks with self-loops and a unique leader.

\begin{figure}
\centering
\includegraphics[scale=0.5]{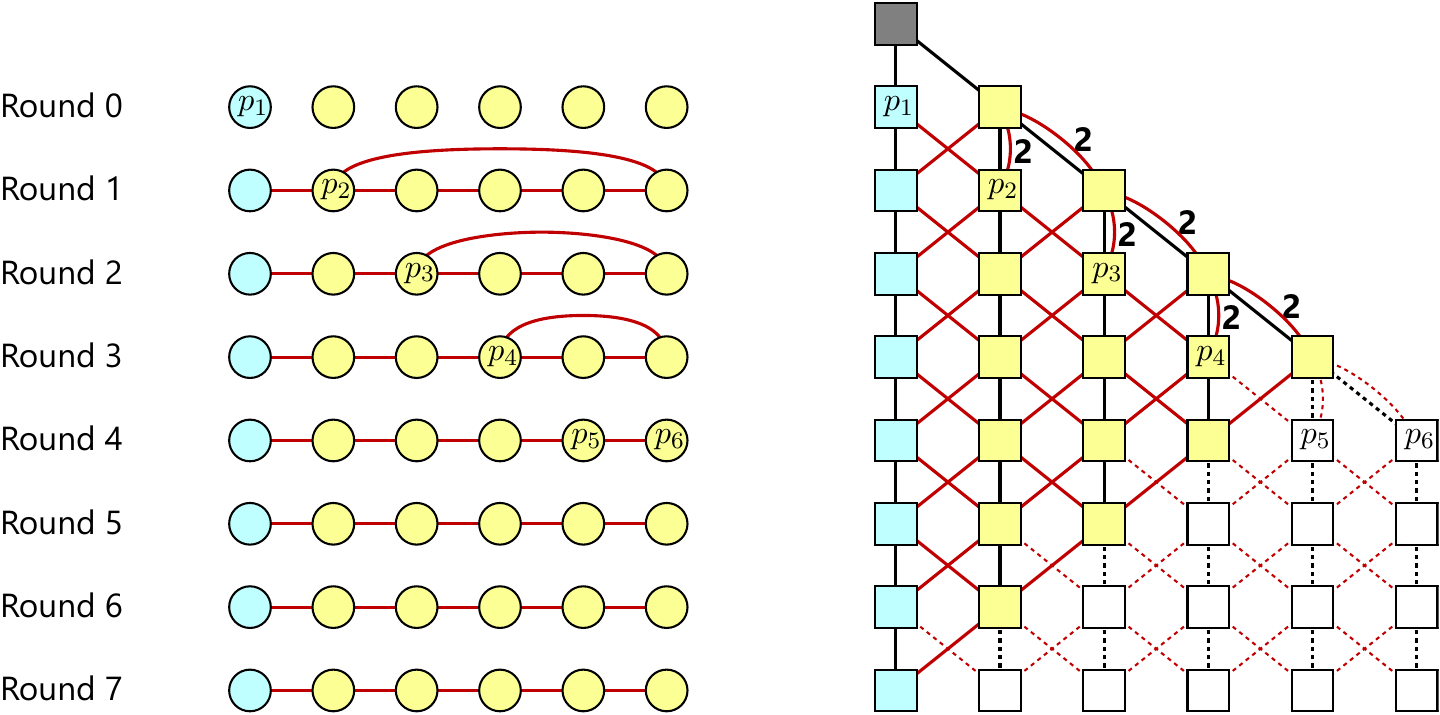}
\caption{The first rounds of the dynamic network $\mathcal G_m$ used in \cref{xth:lower} (left) and the corresponding levels of its history tree (right), where $m=6$; the agent in blue is the leader. The white nodes and the dashed edges in the history tree are not in the vista of the leader at round~$7$. The labels $p_1$, \dots, $p_6$ have been added for the reader's convenience, and mark the agents that get disambiguated, as well as their corresponding nodes of the history tree, which have anonymity~$1$.}
\label{xfig:lower1}
\end{figure}

\begin{figure}
\centering
\includegraphics[scale=0.5]{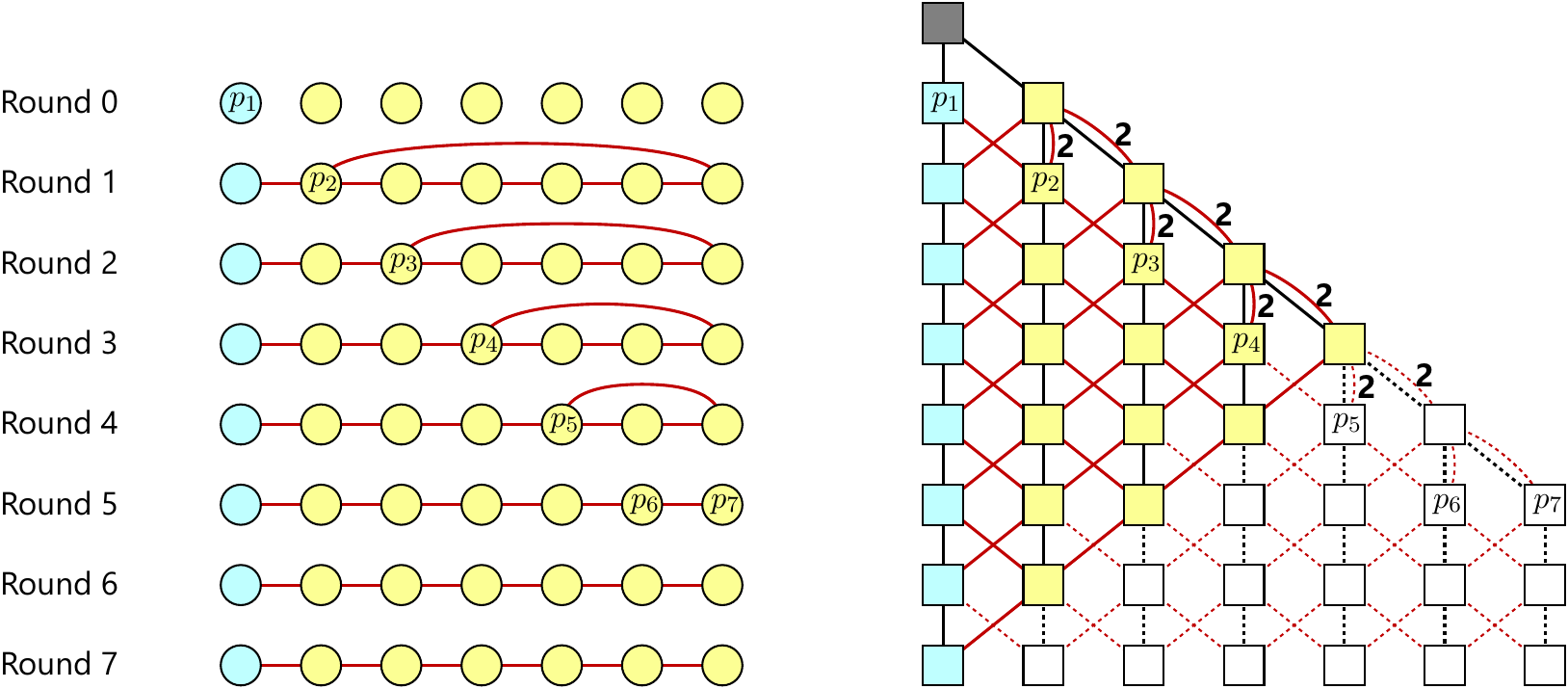}
\caption{The first rounds of the dynamic network $\mathcal G_{m+1}$ with $m=6$. Observe that the vista of the leader at round~$7$ is identical to the vista highlighted in \cref{xfig:lower1}. The intuitive reason is that, from round~$1$ to round~$m-3$, both networks have a cycle whose agents are all indistinguishable (and are therefore represented by a single node in the history tree), except for the one agent with degree~$3$. Thus, the history trees of $\mathcal G_m$ and $\mathcal G_{m+1}$ are identical up to level~$m-3$. After that, the two networks get disambiguated, but this information takes another $m-3$ rounds to reach the leader. Therefore, if the leader of $\mathcal G_m$ and the leader of $\mathcal G_{m+1}$ execute the same algorithm, they must have the same internal state up to round~$2m-5$, due to \cref{xth:view}. In particular, they cannot give different outputs up to that round, which leads to our lower bounds on stabilization and termination for the Counting problem.}
\label{xfig:lower2}
\end{figure}

\smallskip
\mypar{Multi-leader networks.} We can now generalize \cref{xth:lower} to any $\tau$ and any $\ell\geq 1$. This lower bound shows that the algorithms in \cref{s:stableader,s:termleader} are asymptotically optimal for any constant number of leaders $\ell$.

\begin{proposition}\label{t:lower2}
For any $\tau\geq 1$, $\ell\geq 1$, and $m\geq \ell+2$, no algorithm computes the Counting function $F_C$ on all simple networks in $\bigcup_{n=m}^{m+1} \mathcal N^\tau_{n,\ell}$ in less than $\tau(2n-\ell - 5)$ rounds. If termination is required, the bound improves to $\tau(2n-\ell - 3)$ rounds.
\end{proposition}
\begin{proof}
As shown in \cref{xth:lower}, there is a family of simple 1-union-connected networks $\mathcal G_m$, with $m\geq 3$, with the following properties. $\mathcal G_m$ has $\ell=1$ leader and $m$ agents in total; moreover, up to round~$2m-5$, the leaders of $\mathcal G_m$ and $\mathcal G_{m+1}$ have isomorphic vistas.

Let us fix $\ell\geq 1$, and let us construct $\mathcal G'_m$, for $m\geq\ell+2$, as follows. Start from $\mathcal G_{m-\ell+1}$ and rename its agents from $p_1$, \dots, $p_{m-\ell+1}$ to $p_\ell$, \dots, $p_m$, respectively. In the resulting network, at every round, attach a chain of $\ell-1$ additional leaders $p_1$, $p_2$, \dots, $p_{\ell-1}$ to the single leader $p_\ell$. Note that the resulting network $\mathcal G'_m$ has $m$ agents in total and a stable subpath $(p_1, p_2, \dots, p_\ell)$ which is attached to the rest of the network via $p_\ell$.

It is straightforward to see that the agent $p_\ell$ in $\mathcal G'_m$ and the agent $p_\ell$ in $\mathcal G'_{m+1}$, which correspond to the leaders of $\mathcal G_{m-\ell+1}$ and $\mathcal G_{m-\ell+2}$ respectively, have isomorphic vistas up to round~$2(m-\ell)-3$. Since the vista of $p_1$ is completely determined by the vista of $p_\ell$, and it takes $\ell-1$ rounds for any information to travel from $p_\ell$ to $p_1$, we conclude that the agent $p_1$ in $\mathcal G'_m$ and the agent $p_1$ in $\mathcal G'_{m+1}$ have isomorphic vistas up to round~$2m-\ell-4$.

It follows from \cref{c:lowerlinear} (with $k=1$, $a=2$, $b=-\ell-4$) that the Counting function with $\ell\geq 1$ leaders and $\tau=1$ cannot be computed in less than $2n-\ell-5$ rounds in a stabilizing fashion or in less than $2n-\ell-3$ if termination is required. These bounds generalize to an arbitrary $\tau$ by \cref{p:time}.
\end{proof}

Again, a slightly better lower bound can be obtained if we allow self-loops in the network. Starting from the static networks $\mathcal P_{m-\ell+1}$ and $\mathcal P_{m-\ell+2}$ defined above, attach a chain of $\ell-1$ additional leaders to each original leader, as in the proof of \cref{t:lower2}. The new endpoint leaders have identical vistas up to round~$2m-\ell-1$. By \cref{c:lowerlinear,p:time}, this implies a lower bound of $\tau(2n-\ell-2)$ rounds for stabilization and $\tau(2n-\ell)$ rounds for termination in networks with self-loops and $\ell$ leaders (if $\tau=1$, this holds even for static networks).

\smallskip
\mypar{Leaderless networks.} Finally, we can use \cref{xth:lower} to obtain a lower bound for the Average Consensus problem in leaderless networks. This lower bound shows that the leading term in the running time of the algorithm in \cref{s:stableaderless}, i.e., $2\tau n$, is optimal.

\begin{proposition}\label{t:lower1}
For any $\tau\geq 1$ and $m\geq 3$, no algorithm solves the Average Consensus problem on all simple networks in $\bigcup_{n=m}^{m+1} \mathcal N^\tau_{n,0}$ in less than $\tau(2n-6)$ rounds. If termination is required, the bound improves to $\tau(2n-4)$ rounds.
\end{proposition}
\begin{proof}
According to \cref{xth:lower}, the number of agents $n$ in a network with $\ell=1$ and $\tau=1$ cannot be determined in less than $2n-6$ rounds ($2n-4$ rounds if termination is required). We can reduce this problem to Average Consensus with $\ell=0$ and $\tau=1$ as follows. In any given network with $\ell=\tau=1$, assign input~$1$ to the leader and clear its leader flag; assign input~$0$ to all other agents. If the agents can compute the mean input value, $1/n$, they can invert it to obtain $n$ in the same number of rounds. It follows that Average Consensus with $\ell=0$ and $\tau=1$ cannot be solved in less than $2n-6$ rounds ($2n-4$ rounds if termination is required); this immediately generalizes to an arbitrary $\tau$ by \cref{p:time}.
\end{proof}

By setting $m=N-1$, this theorem immediately yields a lower bound of $\tau(2n-4)$ rounds for the terminating algorithm in \cref{t:noleadterm}, provided that $N\geq 4$.

Observe that \cref{t:lower1} says nothing about networks with a fixed dynamic diameter $d$, because $\mathcal G_m$ and $\mathcal G_{m+1}$ have different dynamic diameters. The following result shows that the algorithm in \cref{c:noleadterm}, which pertains to leaderless networks with a known $d$, is asymptotically optimal, as well.
\begin{proposition}\label{t:lower3}
For any $\tau\geq 1$ and $m\geq 3$, no algorithm solves the Average Consensus problem on all simple networks in $\bigcup_{n=m}^{m+1}\mathcal N^{[d]}_{n,0}$ in less than $\tau(2n-6)$ rounds, where $d=\tau(m-1)$. If termination is required, the bound improves to $\tau(2n-4)$ rounds.
\end{proposition}
\begin{proof}
It is easy to see that the network $\mathcal G_m$ has dynamic diameter $m-1$. Let us construct the network $\mathcal G'_{m+1}$ by taking $\mathcal G_{m+1}$ and adding the edge $\{p_{m-1},p_{m+1}\}$ to $G^{(m+1)}_t$, for all $t\geq m-1$. As a result, $\mathcal G'_{m+1}$ has dynamic diameter $m-1$, as well. Also, the leaders of $\mathcal G_{m+1}$ and $\mathcal G'_{m+1}$ have isomorphic vistas up to round $(m-1)+(m-2)-1=2m-4$ (they may become distinguishable only when they are influenced by $p_{m-1}$ after round $m-1$). Thus, since the leaders of $\mathcal G_m$ and $\mathcal G_{m+1}$ have isomorphic vistas up to round $2m-5$, the same holds for the leaders of $\mathcal G_m$ and $\mathcal G'_{m+1}$.

Now we can reason as in \cref{t:lower1}, with $\mathcal G'_{m+1}$ in lieu of $\mathcal G_{m+1}$, to conclude that no algorithm solves the Average Consensus problem on all simple networks in $\bigcup_{n=m}^{m+1}\mathcal N^{[d]}_{n,0}$ in less than $2n-6$ rounds ($2n-4$ if termination is required), where $d=m-1$. Due to \cref{p:time}, this generalizes to an arbitrary $\tau$, provided that $d=\tau(m-1)$.
\end{proof}

\section{Conclusions}\label{s:conclusion}
We introduced the novel concept of \emph{history tree} and used it as our main investigation technique to study computation in anonymous dynamic networks, modeled as sequences of multigraphs. History trees are a powerful tool that completely and naturally captures the concept of symmetry and indistinguishability among agents. In fact, the history tree of a (static or dynamic) network encodes all the information that can be extracted by the agents in the network (\cref{xth:view}). We have demonstrated the effectiveness of our methods by optimally solving a wide class of fundamental problems (\cref{tab:summres}), and we believe that our techniques will find numerous applications in other settings, as well.

We have shown that anonymous agents in $\tau$-union-connected dynamic networks can compute all the multiset-based functions and no other (topology-independent) functions, provided that the network contains a known number of leaders $\ell\geq 1$. If there are no leaders or the number of leaders is unknown, the class of computable functions reduces to the frequency-based functions. We have also identified the Input Frequency function and the Input Multiset function as the complete problems for each class. Notably, the network's dynamic disconnectivity $\tau$ does not affect the computability of functions, but only makes computation proportionally slower.

Moreover, we gave efficient stabilizing and terminating algorithms for computing all the aforementioned functions. Some of our algorithms make assumptions on the agents' a-priori knowledge about the network; we proved that these assumptions are necessary, in the sense that removing any one of them without supplying alternative information makes the corresponding problem unsolvable. All our algorithms have optimal linear running times in terms of $\tau$ and the size of the network $n$.

In one case, there is still a small gap in terms of the number of leaders $\ell$. Namely, for terminating computations with $\ell\geq 1$ leaders, we have a lower bound of roughly $\tau(2n-\ell)$ rounds (\cref{t:lower2}) and an upper bound of roughly $\tau(\ell^2+\ell+1)n$ rounds (\cref{t:multileadterm}). Although these bounds asymptotically match if the number of leaders $\ell$ is a constant (which is a realistic assumption in most applications), determining the optimal dependence on $\ell$ remains an open problem. More fundamentally, it is unclear whether multiple indistinguishable leaders are an advantage or an obstacle: they provide more distinguished agents, but eliminate the unique point of reference supplied by a single leader. Determining which of these effects prevails is an interesting direction for future research.

It is worth noting that for stabilizing computation (i.e., when explicit termination is not required) in networks with a fixed number of leaders, our lower and upper bounds are essentially $2\tau n$ rounds; hence, in this case we were able to optimize the multiplicative constant, as well. As for terminating computation with a unique leader, we have a lower bound of $2\tau n$ rounds and an upper bound of $3\tau n$ rounds. Although we are still unable to completely close this gap, we emphasize that our findings demonstrate the practical feasibility of general computations in anonymous dynamic networks with a unique leader, which was a major open problem in this research area prior to our work.

Observe that our stabilizing algorithms use an unbounded amount of memory, as agents keep adding nodes to their vistas at every round. This can be avoided if the dynamic disconnectivity $\tau$ (as well as an upper bound on $n$, in the case of a leaderless network) is known: in this case, agents can run the stabilizing and the terminating version of the relevant algorithm in parallel, and stop adding nodes to their vistas when the terminating algorithm halts.

Our algorithms require agents to send each other explicit representations of their history trees, which have roughly cubic size in the worst case. It would be interesting to develop algorithms that only send messages of logarithmic size, possibly with a trade-off in terms of running time. We are currently able to do so for leaderless networks and networks with a unique leader, but not for networks with more than one leader~\cite{DVDistComp}. Combining logarithmic-size messages with only $O(\log n)$ bits of internal memory per agent, apart from the space required to represent the agents' inputs and outputs, is another open problem.

Finally, we wonder if our results extend to other communication models, most notably directed networks and networks where communications are not necessarily synchronous. For directed dynamic networks with leaders, an exponential-time terminating Counting algorithm is given in~\cite{VIG24}; whether a polynomial-time algorithm exists remains open. As for asynchrony, we conjecture that our algorithms can be generalized to networks where messages may be delayed by a bounded number of rounds or agents may be inactive for some rounds (provided that a ``global fairness'' condition is met).

\smallskip
\mypar{Acknowledgments.} The authors would like to thank the anonymous reviewers of the Journal of the ACM for corrections and suggestions that greatly improved the readability of this paper. G.\,A.~Di Luna was supported in part by the Ministero dell'Universit\`a e della Ricerca (MUR) National Recovery and Resilience Plan funded by the European Union-NextGenerationEU through the Project SEcurity and RIghts In the CyberSpace (SERICS) under Grant PE00000014 and project ATENEO Sapienza RM1241911201548A. G.~Viglietta was partially funded by the JSPS KAKENHI Grants 23K10985 and 26K14709 and the University of Aizu Competitive Research Fund.

\addcontentsline{toc}{section}{References}
\bibliographystyle{plainurl}
\bibliography{bibliography}
\end{document}